\documentclass[letterpaper,11pt]{article}

\usepackage{xcolor}

\usepackage{fullpage}
\usepackage{amsthm}
\usepackage{authblk}
\newtheorem{theorem}{Theorem}[section]
\newtheorem{definition}[theorem]{Definition}

\newtheorem{claim}[theorem]{Claim}
\newtheorem{lemma}[theorem]{Lemma}

\newenvironment{citemize}{
  \begin{itemize}
}{
  \end{itemize}
}

\newenvironment{cenumerate}{
  \begin{enumerate}
}{
  \end{enumerate}
}

\usepackage{amsmath}
\usepackage{amsfonts}
\usepackage{amssymb}
\usepackage{graphpap}
\usepackage{graphics}

\usepackage[normalem]{ulem}
\usepackage{array}
\usepackage{enumerate}
\usepackage{framed}

\usepackage{enumitem}
\usepackage{float}
\usepackage{mathrsfs}
\usepackage{etoolbox}
\usepackage{graphicx}
\usepackage{hyperref}
\usepackage{xcolor}
\hypersetup{
    colorlinks,
    linkcolor={black},
    citecolor={black},
    urlcolor={black}
}
\usepackage{wrapfig}
\usepackage{multirow}

\usepackage{pbox}
\usepackage{mathtools}
\usepackage[hang,flushmargin]{footmisc}
\usepackage{booktabs}

\usepackage{changepage}

\newcommand{\hide}[1]{}
\newcommand{\ignore}[1]{}

\newcommand{\esm}[1]{\ensuremath{#1}}
\newcommand{\mr}[1]{\esm{\mathrm{#1}}}
\newcommand{\ms}[1]{\esm{\mathsf{#1}}}

\DeclareMathAlphabet{\mathcalalt}{OMS}{zplm}{m}{n}
\newcommand{\mc}[1]{\esm{\mathcalalt{#1}}}
\newcommand{\mbb}[1]{\esm{\mathbb{#1}}}

\newcommand{\Z}{\esm{\mathbb{Z}}}

\newcommand{\N}{\esm{\mathbb{N}}}

\newcommand{\calA}{\mc{A}}
\newcommand{\calB}{\mc{B}}
\newcommand{\calC}{\mc{C}}
\newcommand{\calD}{\mc{D}}

\newcommand{\calK}{\mc{K}}

\newcommand{\calM}{\mc{M}}

\newcommand{\calO}{\mc{O}}

\newcommand{\calS}{\mc{S}}

\newcommand{\calZ}{\mc{Z}}

\newcommand{\bbG}{\mbb{G}}

\newcommand{\abs}[1]{\esm{\left| #1 \right|}}
\newcommand{\set}[1]{\esm{\left\{ #1 \right\}}}

\newcommand{\zo}{\{0,1\}}

\newcommand{\ord}[1]{\esm{{#1}^{\mr{th}}}}
\newcommand{\getsr}{\xleftarrow{\textsc{r}}}
\newcommand{\sind}[2]{#1^{(#2)}}

\newcommand{\para}[1]{\paragraph{#1}}

\newcommand{\mathsc}[1]{{\normalfont \textsc{#1}}}
\newcommand{\msc}[1]{\esm{\mathsc{#1}}}
\newcommand{\mpk}{\msc{mpk}}
\newcommand{\msk}{\msc{msk}}
\newcommand{\id}{\msc{id}}
\newcommand{\hid}{\overline{\id}}

\newcommand{\sk}{\msc{sk}}
\newcommand{\skid}{\sk_\id}
\newcommand{\ct}{\msc{ct}}
\newcommand{\hct}{\overline{\ct}}

\newcommand{\sig}{\textsc{sig}}

\newcommand{\prg}{\ms{PRG}}
\newcommand{\hyb}{\ms{H}}
\newcommand{\Extract}{\ms{Extract}}

\newcommand{\ksk}{\atk}

\newcommand{\setup}{\ms{Setup}}
\newcommand{\extract}{\ms{Extract}}
\newcommand{\enc}{\ms{Encrypt}}
\newcommand{\dec}{\ms{Decrypt}}

\newcommand{\Test}{\ms{Test}}

\newcommand{\protoleft}[1]{\xleftarrow{\mathmakebox[26em]{\displaystyle #1}}}
\newcommand{\protoright}[1]{\xrightarrow{\mathmakebox[26em]{\displaystyle #1}}}

\newcommand{\protoleftlong}[1]{\xleftarrow{\mathmakebox[34em]{\displaystyle #1}}}
\newcommand{\protorightlong}[1]{\xrightarrow{\mathmakebox[34em]{\displaystyle #1}}}

\newcommand{\peenc}{\mathsf{PE.Enc}}
\newcommand{\pedec}{\mathsf{PE.Dec}}

\newcommand{\kdf}{\mathsf{KDF}}
\newcommand{\htk}{\mathsf{htk}}
\newcommand{\exk}{\mathsf{exk}}
\newcommand{\atk}{\mathsf{atk}}
\newcommand{\eadk}{\mathsf{eadk}}
\newcommand{\sid}{\mathsf{sid}}
\newcommand{\bid}{\mathsf{bid}}
\newcommand{\hsid}{\overline{\sid}}
\newcommand{\hbid}{\overline{\bid}}
\newcommand{\hhtk}{\overline{\htk}}
\newcommand{\hexk}{\overline{\exk}}
\newcommand{\hatk}{\overline{\atk}}
\newcommand{\headk}{\overline{\eadk}}
\newcommand{\expt}{\mathsf{Expt}}

\newcommand{\statereveal}{\ms{StateReveal}}
\newcommand{\keyreveal}{\ms{KeyReveal}}
\newcommand{\broadcastreveal}{\ms{BroadcastReveal}}
\newcommand{\corrupt}{\ms{Corrupt}}

\newcommand{\ptest}{T}
\mathchardef\mhyphen="2D
\newcommand{\indidcca}{\ms{IND \mhyphen ID \mhyphen CCA}}
\newcommand{\indcpa}{\ms{IND \mhyphen CPA}}

\newcommand{\negl}{\ms{negl}}
\newcommand{\Funs}{\ms{Funs}}
\newcommand{\Enc}{\ms{Encrypt}}
\newcommand{\Dec}{\ms{Decrypt}}
\newcommand{\LoR}{\textbf{LoR}}

\begin{document}

\title{Privacy, Discovery, and Authentication for the Internet of Things \\ {\normalsize (Extended Version)}}

\author[1]{David J. Wu}
\author[2]{Ankur Taly}
\author[2]{Asim Shankar}
\author[1]{Dan Boneh}
\affil[1]{Stanford University}
\affil[2]{Google}
\date{}
\maketitle

\begin{abstract}
Automatic service discovery is essential to realizing the full potential of
the Internet of Things (IoT). While discovery protocols like Multicast DNS,
Apple AirDrop, and Bluetooth Low Energy have gained widespread adoption across
both IoT and mobile devices, most of these protocols do not offer any form of privacy control
for the service, and often leak sensitive information such as service type,
device hostname, device owner's identity, and more in the clear.

To address the need for better privacy in both the IoT and the mobile landscape,
we develop two protocols for private service discovery and private mutual
authentication. Our protocols provide private and authentic service
advertisements, zero round-trip (0-RTT) mutual authentication, and are
provably secure in the Canetti-Krawczyk key-exchange model. In contrast to
alternatives, our protocols are lightweight and require minimal modification
to existing key-exchange protocols. We integrate our protocols into an
existing open-source distributed applications framework, and provide
benchmarks on multiple hardware platforms: Intel Edisons, Raspberry Pis, smartphones,
laptops, and desktops. Finally, we discuss some privacy limitations of the Apple
AirDrop protocol (a peer-to-peer file sharing mechanism) and show how to
improve the privacy of Apple AirDrop using our private mutual authentication
protocol.
\end{abstract}


\section{Introduction}

Consider a smart home with dozens of IoT devices: an alarm system, a
nanny camera, health monitoring devices, house controls (e.g.,
lighting, heating), and electronics.  Many of these devices need to be
controlled by multiple people, including residents, guests, employees,
and repairmen.  The devices must be easily discoverable by 
all these people.

To provide a good experience, IoT devices advertise the
services they offer using a service discovery mechanism.
Examples include Mutlticast DNS (mDNS)~\cite{CK13a, CK13b}, Apple
Bonjour~\cite{Bonjour}, Bluetooth Low Energy (BLE)~\cite{BLE}, and
Universal Plug-N-Play (UPnP)~\cite{UPnP}.  These mechanisms require
only a broadcast communication channel between the devices (unlike
older discovery protocols~\cite{Jini, CZHJK99, ZMN05} that need
a directory service).  Moreover, these protocols adhere to the zero
configuration networking charter (\emph{Zeroconf})~\cite{Zeroconf} and
can operate with minimal user intervention.

Privacy is an important feature often missing in zero-configuration
service discovery protocols (e.g., Zeroconf)~\cite{KM14a, KM14b,
  KBSW13, PGMSW07}. Services broadcast extensive information about
themselves in the clear to make it easy for clients to discover
them. Advertisements often include sensitive information such as
service type, device hostname, and the device owner's identity.
\ignore{While this may not pose a privacy threat when the service is
  running on a public device (e.g., a vending machine), it certainly}
This poses a threat when the service is running on a private device
(e.g., an alarm system or a smart watch). Identities obtained from
personal devices can be used for user profiling, tracking, and
launching social engineering attacks. A recent study~\cite{KBSW13}
revealed that 59\% of all devices advertise their owner's name in the
clear, which is considered harmful by more than 90\% of the
device owners.  Indeed, one would not want random visitors, or
passerbys, to ``discover'' the alarm system in their home.  Only
authorized clients, such as the home owner and her family, a
technician, or local police, should be able to discover this device.

\medskip
In this work, we address this problem by building a new discovery and
authentication mechanism that respects the privacy of both sides. 

\para{Private service discovery.} Our goal is to ensure that services
are only discoverable by an authorized set of clients. This problem is
challenging as on one hand, services want to advertise themselves only
after confirming that the client trying to discover them is authorized
to see them.  On the other hand, clients want to reveal their identity
only after verifying that the service they are talking to is the
desired one.  In particular, a client device, such as a smartphone,
should not simply identify itself to every device in the wild that requests
it.  This leads to a chicken-and-egg problem reminiscent of the
settings addressed by secret handshakes and hidden
credentials~\cite{BDSSSW03,LDB05,JKT06,AKB07,HBSO03,FAL04}.

\para{Private mutual authentication.} A closely related privacy problem arises
during  authentication between mutually suspicious entities. Most existing
mutual authentication protocols (SIGMA~\cite{CK02,Kra03}, JFK~\cite{ABBCIKR04},
and TLS~\cite{DR08}) require one of the parties (typically the server) to
reveal its identity to its peer before the other, effectively making that
party's identity public to anyone who communicates with it. This is
undesirable when the participants are personal end-user devices,
where neither device is inclined to reveal its identity before learning that
of its peer. Private mutual authentication is the problem of designing a
mutual authentication protocol wherein each end learns the identity of its peer
only if it satisfies the peer's authorization policy.\footnote{While
protocols like SIGMA-I~\cite{Kra03,CK02} and TLS~1.3~\cite{KW15,Res15}
can ensure privacy against passive adversaries, they
do {\em not} provide privacy against active attackers.}

\para{An application.}
Our private discovery protocols apply broadly to many identification 
and key-exchange settings.
Here we describe a common mobile-to-mobile example: peer-to-peer file sharing.
Protocols such as AirDrop and Shoutr
have become popular among mobile
users for sharing photos and other content with their friends. These
peer-to-peer protocols typically work by having a participant
start a sharing service and making it publicly discoverable. The
other device then discovers the service and connects to it to complete the file
transfer. While this offers a seamless sharing experience, it compromises
privacy for the device that makes itself discoverable---nearby
devices on the same network can also listen to the advertisement and obtain 
identifiers from it. A private service discovery
mechanism would make the service advertisement available only to the intended
devices and no one else. The AirDrop protocol offers a ``contacts-only''
mode for additional privacy, but as we show in
Section~\ref{sec:airdrop-case-study}, this mechanism leaks significant 
private information.
The private discovery protocols we develop in this paper provide an
efficient solution to  these problems.

\subsection{Our Contributions}
This paper presents private mutual authentication and service discovery
protocols for IoT and mobile settings. Given the network connectivity constraints
implicit to these settings, our protocols do not require devices to maintain
constant connectivity to an external proxy or directory service in
the cloud. Furthermore, the protocols do not require the participants
to have an out-of-band shared secret, thereby allowing
devices with no pre-existing relationships to discover each other (in accordance with
their respective privacy policies).

In Section~\ref{sec:threatmodel}, we motivate the desired features we seek in our protocols
by presenting a case study of the Apple AirDrop protocol---specifically, its ``contacts-only" mode
for private file sharing. We describe several privacy vulnerabilities in the design of 
AirDrop, which we have disclosed to Apple. In light of these vulnerabilities,
we define the robust privacy guarantees that we seek in our protocols.

\para{Protocol construction.}
Our protocols are designed for distributed public-key infrastructures, such
as the Simple Distributed Security Infrastructure (SDSI)~\cite{RL96}.
Each principal has a public and private key-pair (for a signature scheme), and a
hierarchical human-readable name bound to its public key by a certificate
chain. The key primitive in our design is an encryption scheme that
allows messages to be encrypted under
an authorization policy so that it can be decrypted only by principals
satisfying the policy. Using this primitive, we design a mutual authentication
protocol where one party sends its identity (certificate chain)
encrypted under its authorization policy. This protects the privacy of
that party. The other party maintains its privacy by revealing
its identity only {\em after} verifying the first party's identity. The same primitive is also used
to construct a private service discovery protocol by having a service encrypt
its advertisement under its authorization policy before broadcasting.

The service advertisements in our discovery protocol carry a signed semi-static Diffie-Hellman (DH)
key. The signature provides authenticity for the advertisements
and protects clients from 
connecting to an impostor service.
The semi-static DH key enables clients to establish a secure 
session with the service using zero round-trips (0-RTT), similar to what
is provided in TLS 1.3~\cite{Res15,KW15}.

The authorization policies considered in this work are based on name prefixes. For instance,
a technician Bob from HomeSecurity Corp.\ may have the name
\texttt{HomeSecurityCorp/Technician/Bob}, and a home security system might have a 
policy that only principals whose name starts with \texttt{HomeSecurityCorp/Technician} are allowed
to discover it. Encrypting messages under a prefix-based authorization
policy is possible using a prefix encryption scheme~\cite{LW14}, which can be constructed 
using off-the-shelf identity-based encryption (IBE) schemes~\cite{BF01,BB04a}. 

\para{Protocol analysis.} We give a full specification of our
private mutual authentication and service discovery protocols in
Sections~\ref{sec:private-mutual-auth} and~\ref{sec:discovery-protocol}.
We also discuss a range of practical issues related to our protocol such
as replay protection, ensuring perfect forward secrecy, and amortizing the
overhead of the prefix encryption.
In Appendices~\ref{app:mutual-auth-analysis}
and~\ref{app:discovery-protocol-security},
we provide a rigorous proof of the security and privacy of
both protocols in the Canetti-Krawczyk key-exchange model~\cite{CK01,CK02,Kra03}.

\para{Implementation and evaluation.}
We implemented and deployed our protocols in the \emph{Vanadium} open-source
distributed application framework~\cite{Vanadium}. 
We measured the end-to-end latency overhead for our
private mutual authentication protocol on an Intel Edison, a Raspberry Pi, a smartphone,
a laptop, and a desktop. On the desktop, the protocol completes in 9.5~ms, which corresponds to a
1.8x slowdown over the SIGMA-I protocol that does {\em not} provide mutual privacy.
On the Nexus~5X and the Raspberry Pi, the protocol completes in just over 300~ms 
(about a 3.8x slowdown over SIGMA-I), which makes it suitable for
user-interactive services such as AirDrop and home security system controls
that do not have high throughput requirements.

For the discovery protocol, a service's private discovery message consists
of approximately~$820$ bytes of data. Since mDNS broadcasts support
up to $1300$ bytes of data, it is straightforward to deploy our
discovery protocol over mDNS.
In Section~\ref{sec:experiments}, we also discuss mechanisms for
deploying the protocol over Bluetooth Low Energy and other protocols
where the size of the advertisement packets are
much more constrained.

Based on our benchmarks, our protocols are practical on a range of IoT devices,
such as thermostats (e.g., Nest), security systems
(e.g., Dropcam), and smart
switches (e.g., Belkin Wemo).
All of these devices have hardware comparable to a Pi or an Intel Edison.
In fact, the Intel Edison is marketed primarily
as a platform for building IoT applications. Moreover,
as our AirDrop analysis demonstrates, many of the privacy issues we describe
are not limited to only the IoT setting. Indeed, in
Section~\ref{sec:fixing-airdrop}, we show how our private mutual authentication and
discovery protocols can be efficiently deployed to solve privacy problems
in peer-to-peer interactions on smartphones. On more constrained processors
such as the ARM Cortex M0, however, we expect
the handshakes to take several seconds to complete. This makes our protocols less
suitable in Cortex M0 applications that require fast session setup. Nonetheless,
our protocols are sufficient for a wide range of existing IoT and mobile scenarios.


\section{Desired Protocol Features}\label{sec:threatmodel}
In this section, we define the privacy properties and features that we
seek in our protocols.  We begin with a case study of Apple's AirDrop protocol,
and use it to motivate our privacy concerns and desired features.

\subsection{Case Study: Apple AirDrop}
\label{sec:airdrop-case-study}

AirDrop is a protocol for 
transferring files between two devices running
OS X (Yosemite or later) or
iOS (version 7 or later). It is designed to work
whenever two AirDrop-enabled devices are close to 
each other and even when they do not have Internet access. 
AirDrop uses both Bluetooth Low Energy (BLE) and 
Apple's peer-to-peer WiFi technology ({\tt awdl}) for device
discovery and file transfer. 

To receive files, devices make themselves discoverable by
senders. AirDrop offers two modes for making devices discoverable:
\emph{everyone}, which makes the device discoverable by all nearby devices, and
\emph{contacts-only} (default), which makes the receiving device discoverable only by
senders in its contacts.  The contacts-only mode is meant to be
a privacy mechanism and can be viewed as a solution to the private service
discovery problem for the ``contacts-only'' policy.

\para{Protocol overview.} \label{airdrop-protocol}
We analyzed the AirDrop protocol to understand its privacy
properties and see how it solves the chicken-and-egg
problem of private mutual authentication.
Below, we present a high-level description of the AirDrop protocol in
contacts-only mode, based on the iOS9 security guide~\cite{Apple-iOS15} and 
our experiments with observing AirDrop flows between a MacBook Pro
and an iPhone.
\begin{cenumerate}
  \item When a sender opens the sharing pane on her device, the
  device advertises a truncated hash of the sender's identity
   over BLE, and simultaneously
  queries for a service with label {\tt \_AirDrop.\_tcp.local}
  over mDNS.
  \item A nearby receiving device matches the hash of the sender's identity
  against its contacts list. If a match is found, the receiving device
  starts a service and advertises
  its instance name, {\tt awdl} IP address, and port number over mDNS, under
  the label {\tt \_AirDrop.\_tcp.local}.
  \item The sending device obtains the receiving device's service
  advertisement, and initiates a TLS (version 1.2) connection to it over {\tt awdl}.
  The TLS handshake uses client authentication, wherein each device sends
  its iCloud identity certificate\footnote{All AirDrop-enabled devices 
  have an RSA public and private key pair and an iCloud certificate
  for the owner's identity.}
  to the other in the clear. The handshake
  fails if either party receives a certificate for someone not in their
  contacts list.
  \item Once the TLS connection is established, the
  receiver's description is displayed on the sending device's sharing pane. The
  sender selects the receiver as well as the files to be shared which are then
  sent over the established TLS channel.
\end{cenumerate}

\para{Privacy weaknesses in Apple AirDrop.}
Our analysis indicates that AirDrop employs two main privacy checks in
contacts-only mode. First, a receiving device responds only if the sender's
identifier (received over BLE) matches one of its contacts, and second, a
communication channel is established between a sender and receiver only if
their respective certificates match a contact on their peer's device.  While
necessary, these checks are insufficient to protect the privacy of the sender
and receiver. Below, we enumerate some of the privacy problems with the existing
protocol.
\begin{citemize}
  \item \textbf{Sender and receiver privacy and tracking.} The use of TLS~1.2 with client
  authentication causes both the sender and receiver to exchange certificates 
  in the clear. This makes their identities, as specified by their certificates,
  visible to even a {\em passive} eavesdropper on the network. Moreover, the public keys in the certificates 
  allow the eavesdropper to track the sender and receiver in the future. While using a key-exchange
  protocol like SIGMA-I~\cite{Kra03} or TLS 1.3~\cite{Res15} 
  would hide the certificate for one of the parties (the sender in this case), the certificate
  of the other party (the receiver) would still be revealed to an active attacker.
  Protecting the privacy of {\em both} parties against active
  attackers, requires {\em private} mutual authentication, as constructed
  in Section~\ref{sec:private-mutual-auth}. 

  \item \textbf{Sender impersonation.} Another privacy problem is that the
  sender's identifier advertised over BLE can be forged or replayed by an
  attacker to trick an honest receiver into matching
  it against its contacts. Based on the receiver's response, the attacker
  learns whether the receiver has the sender in their contacts, and moreover,
  could try to initiate a TLS session with the receiver to obtain its certificate.
  To protect against this kind of impersonation attack, discovery
  broadcasts must provide some kind of {\em authenticity}, as 
  in Section~\ref{sec:discovery-protocol}.  
\end{citemize}

\subsection{Protocol Design Goals}\label{sec:desired-features}

The privacy properties of AirDrop are insufficient to solve
the private service discovery problem. While our case study in
Section~\ref{sec:airdrop-case-study} focuses exclusively
on the AirDrop protocol, most existing key-exchange and
service discovery protocols do not provide robust privacy and authenticity
guarantees. We survey some of these alternative protocols in Section~\ref{sec:related-work}.
In this section, we describe the strong privacy properties we seek in our
  protocols. We want these properties to hold even against an active network
  attacker (i.e., an attacker that can intercept, modify, replay and drop packets
  arbitrarily). We begin by describing two concrete privacy objectives:

\begin{citemize}
  \item \textbf{Mutual privacy.} The protocols must ensure that the identities
  and any identifying attributes of the protocol participants are only
  revealed to authorized recipients. For service discovery, this
  applies to both the service being advertised and the clients trying to
  discover it.

  \item \textbf{Authentic advertisements.} Service advertisements should
  be unforgeable and authentic. Otherwise, an attacker may forge a
  service advertisement to determine if a client is interested in the service.

\end{citemize}
Finally, to ensure that our protocols are
applicable in both IoT and peer-to-peer settings, we impose
additional constraints on the protocol design:
\begin{citemize}
  \item \textbf{No out-of-band pairing for participants.}
  The protocol should not require participants to exchange certain
  information or secrets out-of-band. This is especially important for the
  discovery protocol as the service may not know all the clients that might
  try to discover it in the future.
    Besides impacting feasibility, the pairing requirement also degrades the
    user experience. For instance, the main charm of AirDrop is that it ``just
    works'' without needing any {\em a priori} setup between the sender and
    recipient.

  \item \textbf{No cloud dependency during protocol execution.}
  The protocol should not rely on an external service in the cloud, such as a
  proxy or a directory service. Protocols that depend on cloud-based services
  assume that the participating devices maintain reliable Internet access.
  This assumption fails for many IoT devices, including devices that only
  communicate over Bluetooth, or ones present in spaces
  where Internet access is unreliable.
    Again, a nice feature of the AirDrop protocol is that it works even if
    neither device is connected to the Internet. Thus, we seek a protocol that
    maximizes peer-to-peer communication and avoids dependence on global
    services.
\end{citemize}


\section{Preliminaries}\label{sec:prelim}

In this section, we describe our identity and authorization model, as well
as introduce the cryptographic primitives we use in our constructions.

\para{Identity and authorization model.}
We define our protocols for a generic distributed public-key infrastructure, such as
SDSI~\cite{RL96}.  We assume each principal has a public and private key-pair
for a signature scheme and one or more hierarchically-structured human-readable names
bound to its public key via a certificate
chain. For instance, a television set owned by Alice might have a certificate chain binding
the name {\tt Alice/Devices/TV} to it. The same television set may also have a certificate chain
with the name {\tt PopularCorp/Products/TV123} from its manufacturer. Our protocols
are agnostic to the specific format of certificates and how they are distributed.

Principals authenticate each other by exchanging certificate chains and providing
a signature on a fresh (session-specific) nonce.
During the authentication protocol, a principal validates its peer's certificate
chain, and extracts the name bound to the certificate chain. Authorization
decisions are based on this extracted name, and {\em not} the public key. For example,
Alice may authorize all principals with names matching the prefix pattern {\tt Alice/Devices/*} 
to access her television set. In this work, we consider prefix-based authorization policies.

Prefix-based policies can also be used for group-based access control. For instance, Alice may issue
certificate chains with names {\tt Alice/Family/Bob}, and {\tt Alice/Family/Carol}
to her family members Bob and Carol. She can then authorize her family to discover and 
access her home security system simply by setting the authorization policy to {\tt Alice/Family/*}.

\subsection{Cryptographic and Protocol Building Blocks}
\label{sec:crypto-primitives}
We write $\Z_p$ to denote the group of integers modulo $p$.
For a distribution $\calD$, we write $x \gets \calD$ to denote that $x$ is
drawn from $\calD$. For a finite set $S$, we write $x \getsr S$ to denote that
$x$ is drawn uniformly at random from $S$. A function $f(\lambda)$ is
negligible in a security parameter $\lambda$ if $f = o(1/\lambda^c)$ for all
$c \in \N$.

\para{Identity-based encryption and prefix encryption.}
Identity-based encryption
(IBE)~\cite{Sha84,BF01,Coc01,BB04a} is a generalization of public-key encryption
where public keys can be arbitrary strings, or {\em identities}. 
We give more details in
Appendix~\ref{app:prelims-add}.
Prefix encryption~\cite{LW14} is a generalization of
IBE where the secret key $\skid$ for an identity $\id$
can decrypt all ciphertexts encrypted to any
identity $\id'$ that is a prefix of $\id$ (in
IBE, decryption succeeds only if $\id = \id'$).\footnote{Lewko and
Waters~\cite{LW14} considered the reverse setting where
decryption succeeds if $\id$ is a prefix of $\id'$, but their construction
is easily adaptable to our setting.} Prefix encryption allows for messages to be encrypted under
a prefix-based policy such that the resulting ciphertext
can only be decrypted by principals satisfying the policy.

It is straightforward to construct prefix encryption from
IBE. The following construction is adapted from the Lewko-Waters scheme~\cite{LW14}.
The key for an identity $\id = s_1 / s_2 / \cdots / s_n$ consists of $n$ different
IBE keys for the following sequence of identities:
$(s_1), (s_1 / s_2), \ldots, (s_1 / s_2 / \cdots / s_n)$. Encryption
to an identity $\id'$ is just IBE encryption to the identity $\id'$.
Given a secret key $\skid$ for $\id$, if $\id'$ is a prefix
of $\id$, then $\skid$ contains an IBE identity key for $\id'$.

The syntax of a prefix encryption scheme is very similar to that of an IBE
scheme. Secret keys are still associated with identities, but ciphertexts are now
associated with prefix-constrained policies. In the following,
we write $\peenc(\mpk, \pi, m)$ to denote an encryption algorithm
that takes as input the public key
$\mpk$, a message $m$, a prefix-constrained
policy $\pi$, and outputs a ciphertext $\ct$. When there is no ambiguity,
we will treat $\mpk$ as an implicit parameter to $\peenc$. We write
$\pedec(\sk_\id, \ct)$ for the decryption algorithm that takes in a
ciphertext $\ct$ and a secret key $\sk_\id$ (for an identity $\id$)
and outputs a message if $\id$ matches the ciphertext policy $\pi$, and
a special symbol $\bot$ otherwise.

\para{Other cryptographic primitives.} We write $\set{m}_k$ to denote an
authenticated encryption~\cite{BN00,Rog02,BRW03}
of a message $m$ under a key $k$, and $\kdf(\cdot)$ to
denote a key-derivation function~\cite{DGHKR04,Kra10}.
We describe these additional primitives
as well as the cryptographic assumptions (Hash Diffie-Hellman
and Strong Diffie-Hellman~\cite{ABR01}) we use in our security analysis
in Appendix~\ref{app:prelims-add}.

\para{Key-exchange model.} We analyze the security of our private mutual
authentication and privacy service discovery protocols in the
Canetti-Krawczyk~\cite{CK01,CK02,Kra03} key-exchange model, which models
the capabilities of an active network adversary. We defer the formal
specification of this model and our generalization of it to the 
service discovery setting to
Appendix~\ref{app:key-exchange-model}.


\section{Private Mutual Authentication Protocol}
\label{sec:private-mutual-auth}

In this section, we describe our private mutual authentication protocol and
discuss some of its features and limitations. We use the identity and
authorization model described in Section~\ref{sec:prelim}.

\para{Protocol execution environment.}
In our setting, each principal has a signing/verification key-pair
and a set of names
(e.g., {\tt Alice/Devices/TV}) bound to its public verification key
via certificate chains.
For each name, a principal possesses an identity secret key
(for the prefix encryption scheme) extracted for that name.
The secret key extraction is carried out by IBE
root authorities (who possess the IBE master secret key $\msk$), which may coincide with
certificate authorities. For example, Alice could be an IBE root that issued both
an identity secret key and a certificate chain to its television set for the name
{\tt Alice/Device/TV}. Similarly, the television set may also have an identity secret key from its
manufacturer Popular Corp.\ for the name {\tt PopularCorp/Products/TV123}.
Finally, each principal also has one or more
prefix-constrained authorization policies.

In our protocol description, we refer to the initiator of the
protocol as the {\em client} and the responder as the {\em server}. For a
party $P$, we write $\id_P$ to denote a certificate chain binding $P$'s
public key to one of its identities. For a message $m$, we write
$\sig_P(m)$ to denote $P$'s signature on $m$. We refer to each instantiation
of the key-exchange protocol as a ``session,'' and each session is identified by a
unique session id, denoted $\sid$.

\para{Protocol specification.} Our starting point is the 3-round SIGMA-I
protocol~\cite{Kra03,CK02} which provides mutual authentication as well as
privacy against passive adversaries. Similar to the SIGMA-I protocol, our
protocol operates over a cyclic group $\bbG$ of prime order where the Hash-DH~\cite{ABR01}
assumption holds. Let $g$ be a generator of $\bbG$. We now describe our private mutual
authentication protocol. The message flow is illustrated in Figure~\ref{fig:auth-proto-2}.
\begin{enumerate}
  \item To initiate a session with id $\sid$,
  the client $C$ chooses $x \getsr \Z_p$, and sends $(\sid, g^x)$ to the server.

  \item Upon receiving a start message $(\sid, g^x)$ from a client,
  the server $S$ chooses $y \getsr \Z_p$, and does the following:
  \begin{enumerate}
    \item Encrypt its name $\id_S$ using the prefix encryption scheme under its policy
    $\pi_S$ to obtain an encrypted identity $\ct_S \gets \peenc(\pi_S, \id_S)$.

    \item Derive authenticated encryption keys
    $(\htk, \atk) = \kdf(g^x, g^y, g^{xy})$ for the handshake
    and application-layer messages, respectively.

    \item Compute a signature $\sigma = \sig_S(\sid, \ct_S, g^x, g^y)$ on its encrypted
    identity and the ephemeral session state, and encrypt
    $(\ct_S, \sigma)$ using $\htk$ to obtain a ciphertext $c$.
  \end{enumerate}
  The server replies to the client with $(\sid, g^y, c)$.

  \item When the client receives a response $(\sid, g^y, c)$, it
  derives the keys $(\htk, \atk) = \kdf(g^x, g^y, g^{xy})$. It tries to
  decrypt $c$ with $\htk$ and aborts if decryption fails. It parses
  the decrypted value as $(\ct_S, \sigma_S)$ and checks whether its identity
  $\id_C$ satisfies the server's policy $\pi_S$ (revealed by $\ct_S$).
  If the client satisfies the server's policy, it
  decrypts $\ct_S$ using its identity key $\sk_C$ to obtain the server's identity
  $\id_S$. If $\id_S$ satisfies the client's policy $\pi_C$ and $\sigma_S$
  is a valid signature on $(\sid, \ct_S, g^x, g^y)$ under the
  public key identified by $\id_S$, the client replies to the server with
  the session id $\sid$ and an
  encryption $c'$ of $(\id_C, \sig_C(\sid, \id_C, g^x, g^y))$ under $\htk$.
  Otherwise, the client aborts.

  \item Upon receiving the client's response $(\sid, c')$, the server
  tries to decrypt $c'$ using $\htk$ and aborts if decryption fails.
  It parses  the decrypted value as $(\id_C, \sigma_C)$ and verifies that
  $\id_C$ satisfies its policy and that $\sigma_C$ is a
  valid signature on $(\sid, \id_C, g^x, g^y)$ under the public key identified by $\id_C$. If so, the handshake
  completes with $\atk$ as the shared session key and where
  the client believes it is talking to $\id_S$ and the server believes
  it is talking to $\id_C$.
  Otherwise, the server aborts.
\end{enumerate}

\begin{figure}
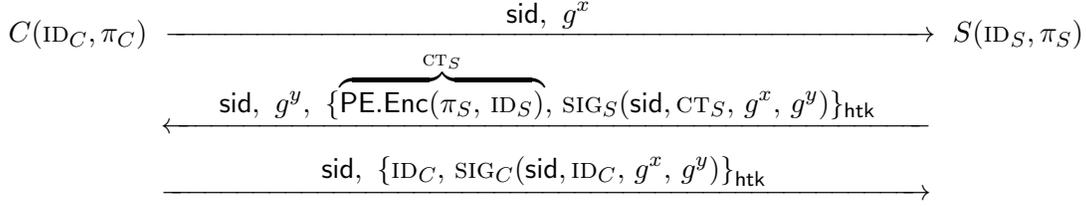

  \[ C(\id_C, \pi_C) \ \protoright{\sid,\ g^x} \ S(\id_S, \pi_S) \]
  \vspace{-1.7em}
  \[ \protoleft{\sid,\ g^y,\ \{ \overbrace{\peenc(\pi_S,\, \id_S)}^{\ct_S}, \, \sig_S(\sid, \ct_S,\, g^x,\, g^y) \}_\htk} \]
  \[ \protoright{\sid, \ \set{ \id_C,\, \sig_C(\sid, \id_C,\, g^x,\, g^y)}_\htk} \]
\caption{Message flow between the client $C$ (with certificate $\id_C$ and policy $\pi_C$)
and the server $S$ (with certificate $\id_S$ and policy $\pi_S$) for the
private mutual authentication protocol. Both the client and the server possess a secret
signing key. The associated verification keys are bound to their identities via
the certificates $\id_C$ and $\id_S$,
respectively. For a message $m$, $\sig_C(m)$ and $\sig_S(m)$ denote
the client's and server's signature on $m$, respectively.
Both the client and server know the master public key for the prefix-based
encryption scheme, and the client possesses a secret key $\sk_C$ for the prefix-based
encryption scheme for the identity associated with its certificate $\id_C$.}
\label{fig:auth-proto-2}
\end{figure}

\subsection{Protocol Analysis}\label{sec:private-mutual-auth-analysis}
In this section, we highlight some properties of our private mutual
authentication protocol.

\para{Comparison with SIGMA-I.} Our authentication
protocol is very similar to the SIGMA-I key-exchange
protocol~\cite[\S 5.2]{Kra03}, but with the following key difference:
the server's certificate, $\id_S$, is sent encrypted under a prefix encryption
scheme. Moreover, instead of deriving separate MAC and encryption keys from the shared DH key,
we combine the two primitives by using an authenticated encryption scheme.
Since we have only added an additional layer of
prefix encryption to the certificates, each party's signature verification
key is still bound to its identity as before. Thus, the proof that the \mbox{SIGMA-I}
protocol is a secure key-exchange protocol~\cite[\S5.3]{CK02} (with perfect forward secrecy)
translates to our setting.

\para{Identity privacy.}
The identity of the server is sent encrypted under its prefix policy, so by
security of the prefix encryption scheme, it is only revealed to clients that satisfy the policy.
Conversely, an honest client only reveals its identity after it has verified that the
server's identity satisfies its policy.
  We formally define our notion of mutual privacy in Appendix~\ref{app:key-exchange-privacy}.
  Then, in Appendix~\ref{app:mutual-auth-analysis},
  we show that the protocol in Figure~\ref{fig:auth-proto-2} achieves our notion
  of mutual privacy.
In contrast,
the SIGMA-I protocol  does not  provide such a guarantee as the identity of
the server is revealed to active adversaries.

\para{Policy privacy.} The protocol in
Figure~\ref{fig:auth-proto-2} ensures privacy for the client's policy,
but not for the server's policy. This issue
stems from the fact that the underlying prefix encryption scheme is not ``policy-hiding,''
meaning an encryption of a message $m$ to a prefix $\pi$ is only guaranteed to hide $m$ and not $\pi$.
Since a passive adversary can observe whether a connection is successfully
established between a client and a server, learning the server's policy allows the
adversary to then infer information about the client's identity. We describe a solution
based on anonymous IBE in Section~\ref{sec:extensions}.

\para{Caching the encrypted certificate chains.} By introducing
prefix encryption, the protocol in Figure~\ref{fig:auth-proto-2} is more
computationally expensive for both the server and the client compared to the
SIGMA-I protocol. However, for key-exchange security,
it is not essential that the server re-encrypt its certificate chain using the
prefix encryption scheme on each handshake. The server can instead reuse the
{\em same} $\ct_S$ across multiple handshakes. Thus, the cost of the IBE
encryptions needed to construct $\ct_S$ can be amortized across multiple
handshakes. The client still needs to perform a
single IBE decryption per handshake.

\para{Unlinkability.} One limitation of our private mutual authentication
protocol is that an active network attacker can fingerprint a server.
Certainly, if the server reuses the same encrypted certificate across
multiple handshakes, an active attacker that does not satisfy the server's
authorization policy can still use the encrypted certificate to identify
the server across sessions. Even without the caching optimization, an
active attacker can still track the server via its signature. This is because
in many signature schemes such as ECDSA, the public verification key can
effectively be recovered from a message-signature pair~\cite{Bro09}.
Thus, as long as the server's signing key (and
correspondingly, verification key) is long-lived, the server's public
verification key can be used to identify and track the server across
sessions. In fact, when deploying mutual authentication protocols, it is
good practice to include the verification key in the signature in order to
defend against Duplicate Signature Key Selection (DSKS)
attacks~\cite{MS04,KM13}.

Depending on the use case, this linkability issue can be problematic. A simple
fix is to instead apply the prefix encryption to the entirety of the server's
response message in the SIGMA-I protocol. Then, both the server's certificate
chain as well as its signature on the ephemeral shares are encrypted under the
prefix encryption scheme.
As a result, the server's response is unique per session and
can only be decrypted by clients whose identities
satisfy the server's policy. Security of the resulting key-exchange protocol
still reduces naturally to the SIGMA-I protocol and privacy still reduces to
the security of the prefix encryption scheme. The main disadvantage of this
protocol is that the server can no longer cache and reuse the prefix-encrypted
ciphertext. Thus, the server must perform a prefix encryption during each handshake,
increasing the computational load on the server. We give the message flow for
this modified protocol in Figure~\ref{fig:auth-proto-1}. Using ProVerif~\cite{BSC16},
we are able to formally verify
(in a Dolev-Yao model of protocol logic~\cite{DY81})
that this modified protocol satisfies the following notion of unlinkability: no client can
distinguish between two servers that have different signing and verification keys
but use the same policy. For the other properties (key-exchange security and mutual privacy),
our analysis (in the standard computational model)
in Appendix~\ref{app:mutual-auth-analysis} applies.

\begin{figure}
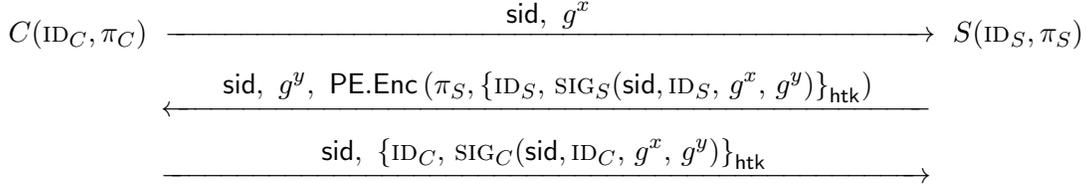

  \[ C(\id_C, \pi_C) \  \protoright{\sid,\  g^x} \  S(\id_S, \pi_S) \]
  \[ \protoleft{\sid, \ g^y,\ \peenc \left(\pi_S, \set{\id_S, \, \sig_S(\sid, \id_S,\, g^x,\, g^y)}_\htk \right)} \]
  \[ \protoright{\sid,\  \set{ \id_C,\, \sig_C(\sid, \id_C,\, g^x,\, g^y)}_\htk} \]
\caption{Message flow between the client $C$ and the server $S$ for the variant of the
         private mutual authentication protocol that provides unlinkability. The main
         difference between this protocol and the private mutual authentication
         protocol in Figure~\ref{fig:auth-proto-2} is that here, the
         server encrypts its entire response using prefix-based encryption. In contrast, in the
         protocol in Figure~\ref{fig:auth-proto-2}, the server only encrypts its
         certificate chain using the prefix-based encryption scheme.}
\label{fig:auth-proto-1}
\end{figure}

\para{Security theorem.} We state the security theorem
for our private mutual authentication protocol here, but defer the formal
proof to
Appendix~\ref{app:mutual-auth-analysis}.

\begin{theorem}[Private Mutual Authentication]
  \label{thm:private-mutual-auth}
  The protocol in Figure~\ref{fig:auth-proto-2} is a secure and private
  key-exchange protocol in the Canetti-Krawczyk key-exchange model
  assuming the Hash Diffie-Hellman assumption in $\bbG$ and the security of all
  underlying cryptographic primitives.
\end{theorem}

\section{Private Service Discovery Protocol}
\label{sec:discovery-protocol}

In this section, we describe our private service discovery protocol. The primary
goal is to make a service discoverable only by parties that
satisfy its authorization policy. Additionally, once a client has discovered a service,
it should be able to authenticate to the server using zero round-trips (0-RTT), i.e.,
include application data on the first flow of the handshake.
0-RTT protocols are invaluable for IoT since devices are often constrained
in both computation and bandwidth.

The key idea in our design is to have the service
include a fresh DH share and a signature in its advertisement.
The DH share allows 0-RTT client authentication, and the
signature provides authenticity for the service advertisement. Next, the service
encrypts its advertisement under its policy $\pi_S$ before broadcasting
to ensure that only authorized clients are able to discover it.
A similar mechanism for
(non-private) 0-RTT authentication is present in OPTLS and the
TLS~1.3 specification~\cite{Res15,KW15}, although OPTLS
only provides server authentication.

\para{Protocol specification.}
Our protocol works over a cyclic group $\bbG$ of prime order $p$ with
generator $g$ where the Strong-DH and Hash-DH assumptions~\cite{ABR01} hold.
The private discovery protocol can be separated into a broadcast protocol and
a 0-RTT mutual authentication protocol. Each broadcast is associated with a
unique broadcast identifier $\bid$ and each session with a unique session
identifier $\sid$. The protocol execution environment is the same as that
described in
Section~\ref{sec:private-mutual-auth}. The basic message flow for the
private discovery protocol is illustrated
in Figure~\ref{fig:discovery-proto}.

\para{Service broadcast message.} To setup a new broadcast
with broadcast id $\bid$, the server $S$ chooses a fresh DH exponent
$s \getsr \Z_p$, and encrypts $(\id_S, g^s, \sig_S(\bid, \id_S, g^s))$ using the
prefix encryption scheme under its authorization policy $\pi_S$ to obtain
a broadcast ciphertext $\ct_S$. The server broadcasts $(\bid, \ct_S)$.

\para{0-RTT mutual authentication.} Upon receiving a broadcast $(\bid, \ct_S)$,
a client performs the following steps to establish a
session $\sid$ with the server:
\begin{enumerate}
  \item The client $C$ checks that its identity $\id_C$ satisfies the
  server's authorization policy $\pi_S$ (included with $\ct_S$). If so, it
  decrypts $\ct_S$ using its prefix encryption secret key and parses the decrypted
  value as  $(\id_S, g^s, \sigma_S)$. It verifies that $\id_S$
  satisfies its policy $\pi_C$ and that $\sigma_S$ is a valid
  signature on $(\bid, \id_S, g^s)$ under the public key identified by $\id_S$. If any step
  fails, the client aborts.

  \item To initiate a session with id $\sid$, the client
  first chooses an ephemeral DH exponent $x \getsr \Z_p$.
  It derives authenticated encryption keys
  $(\htk, \htk', \eadk) = \kdf(g^s, g^x, g^{sx})$, where $\htk$ and $\htk'$
  are used to encrypt handshake messages, and $\eadk$ is used to encrypt
  any early application data the client
  wants to include with its connection request. The client encrypts
  the tuple $(\id_S, \id_C, \sig_C(\bid, \sid, \id_S, \id_C, g^s, g^x))$
  under $\htk$ to obtain a ciphertext $c_1$ and any early application
  data under $\eadk$ to obtain a ciphertext $c_2$. It sends
  $(\bid, \sid, g^x, c_1, c_2)$ to the server.

  \item When the server receives a message from a client of the form
  $(\bid, \sid, g^x, c_1, c_2)$, it first derives the encryption
  keys $(\htk, \htk', \eadk) = \kdf(g^s, g^x, g^{sx})$, where $s$
  is the DH exponent it chose for broadcast $\bid$.
  Then, it tries to decrypt $c_1$ with $\htk$ and $c_2$ with
  $\eadk$. If either decryption fails, the server aborts the protocol. Otherwise, let
  $(\id_1, \id_2, \sigma)$ be the message obtained
  from decrypting $c_1$. The server verifies
  that $\id_1 = \id_S$ and that $\id_2$ satisfies its authorization policy $\pi_S$.
  Next, it checks that $\sigma$ is a valid signature on
  $(\bid, \sid, \id_1, \id_2, g^s, g^x)$ under the public key identified by
  $\id_2$. If all these checks pass, the server chooses a new ephemeral
  DH exponent $y \getsr \Z_p$ and derives the session key
  $\atk = \kdf(g^s, g^x, g^{sx}, g^y, g^{xy})$.\footnote{In this step,
  the server samples a fresh ephemeral DH
  share $g^y$ that is used to derive the application-traffic
  key $\atk$. This is essential
  for ensuring perfect forward secrecy for all subsequent
  application-layer messages (encrypted under $\atk$).
  We discuss the perfect forward secrecy properties of this protocol
  in Section~\ref{sec:disc-proto-analysis}.}
  The server encrypts
  the tuple $(\bid, \sid, \id_1, \id_2, g^s, g^x, g^y)$ under $\htk'$ to
  obtain a ciphertext $c_1'$, and any application messages under
  $\atk$ to obtain a ciphertext $c_2'$. It replies to the client
  with $(\bid, \sid, g^y, c_1', c_2')$.

  \item When the client receives a response message
  $(\bid, \sid, g^y, c_1', c_2')$, it first decrypts $c_1'$ using $\htk'$
  and verifies that $c_1'$ decrypts to $(\bid, \sid, \id_S, \id_C, g^s, g^x, g^y)$,
  where $s \in \Z_p$ is the server's semi-static DH share associated with
  broadcast $\bid$, and
  $x \in \Z_p$ is the ephemeral DH share it used in session $\sid$ with
  broadcast $\bid$. If so,
  it derives $\atk = \kdf(g^s, g^x, g^{sx}, g^y, g^{xy})$ and uses $\atk$ to
  decrypt $c_2'$. The handshake then concludes with $\atk$ as the shared
  session key.
\end{enumerate}

\begin{figure}
Server's Broadcast:
  \vspace{-0.5em}
  \[ \bid, \peenc(\pi_S, (\id_S, g^s, \sig_S(\bid, \id_S, g^s))) \]
0-RTT Mutual Authentication:
  \[ C(\id_C, \pi_C) \protoright{\bid, \sid, g^x, \set{ \id_S, \id_C, \sig_C(\bid, \sid, \id_S, \id_C, g^s, g^x)}_\htk} S(\id_S, \pi_S) \]
  \vspace{-2em}
  \[ \protoleft{\bid, \sid, g^y, \set{ (\bid, \sid, \id_S, \id_C, g^s, g^x, g^y) }_{\htk'}} \]
\caption{Basic message flow between the client $C$ (with certificate $\id_C$ and policy $\pi_C)$
and the server $S$ (with certificate $\id_S$ and policy $\pi_S$) for the private
discovery protocol. As in the private mutual authentication protocol
(Figure~\ref{fig:auth-proto-2}), the client and server each possess a secret signing
key, and the associated verification keys are bound to their identities via
the certificates $\id_C$ and $\id_S$, respectively. Both the client and server know
the master public key for the prefix-based encryption scheme, and the client possesses
a secret key $\sk_C$ for the prefix-based encryption scheme for the identity associated
with its certificate $\id_C$. In this protocol, the client can also include early application data
in the first flow of the 0-RTT mutual authentication protocol under
a key derived from the client's ephemeral DH share $g^x$ and the server's
semi-static DH share $g^s$.}
\label{fig:discovery-proto}
\end{figure}

\subsection{Protocol Analysis}
\label{sec:disc-proto-analysis}
We now describe some of the properties of our private service discovery
protocol in Figure~\ref{fig:discovery-proto}.

\para{0-RTT security.} The security analysis of the 0-RTT mutual
authentication protocol in Figure~\ref{fig:discovery-proto} is similar
to that of the OPTLS protocol in TLS~1.3~\cite{KW15} and relies on the
Strong-DH and Hash-DH assumptions~\cite{ABR01} in the random oracle
model~\cite{BR93a}. Note that in contrast to the
OPTLS protocol which only provides server authentication, our protocol
provides {\em mutual authentication}.

\para{Replay attacks.} One limitation of the 0-RTT mode is that the
early-application data is vulnerable to replay attacks. A typical
replay-prevention technique (used by QUIC~\cite{FG14,LJBN15})
is to have the server maintain a list of client
nonces in the 0-RTT messages and reject duplicates for the lifetime of the
service advertisement.
While this mechanism is potentially suitable if the service
advertisements are short-lived and the service does not have to maintain large
amounts of state, the general problem remains open.

\para{Authenticity of broadcasts.} Because the service broadcasts are signed,
a client is assured of the authenticity of a broadcast before
establishing a session with a service. This ensures that the client will not
inadvertently send its credentials to an impostor service. However, an
adversary that intercepts a service broadcast and recovers the
associated semi-static DH exponent can replay the broadcast for an honest
client. If the client then initiates a session using the DH share from the
replayed advertisement, the adversary compromises the client's privacy. To
protect against this kind of replay attack, the server should include an
expiration time in its broadcasts, and more importantly, sign this expiration.
Then, by having short-lived broadcasts, the adversary's
window of opportunity in mounting a replay attack is greatly reduced. Note though
that recovering the semi-static DH exponent $s$ from $g^s$ amounts to solving the
discrete log problem in the underlying group, which is conjectured to be a difficult
problem. Thus, it is unlikely that authenticity
will be compromised even for long-lived broadcasts. However, perfect
forward secrecy (discussed next) is compromised for all early-application
data for the duration of the broadcast, and so, it is still preferable to have
short-lived broadcasts.

\para{Forward secrecy.} Since the server's semi-static DH share persists
across sessions, perfect forward secrecy (PFS) is lost for early-application
data and handshake messages sent during the lifetime of each
advertisement. Specifically, an adversary that
compromises a server's semi-static DH share
for a particular broadcast is able to break the security on all
early-application data and handshake messages (which contain the client's
identity) in any key-exchange session that occurred during the lifetime of the
particular broadcast. To mitigate this risk in practical deployments, it is
important to periodically refresh the DH-share in the server's broadcast
(e.g., once every hour). The refresh interval corresponds to the window where
forward secrecy may be compromised.

While PFS is not achievable for early-application and handshake messages for the
lifetime of a service's broadcast, PFS is ensured for
all application-layer messages (as long as the server's semi-static secret
was not already compromised at the time of session negotiation). In particular, after
processing a session initiation request, the server responds with a fresh
ephemeral DH share that is used to derive the session key for all subsequent
messages.
In Appendix~\ref{app:discovery-protocol-security}, we
show that the security of the session is preserved even if
the server's semi-static secret is compromised (after the completion of the handshake) but the
ephemeral secret is uncompromised.
This method of combining a semi-static key with an ephemeral key is also used
in the OPTLS~\cite{KW15} protocol.

\para{Identity privacy.} As was the case in our private mutual
authentication protocol from Section~\ref{sec:private-mutual-auth},
privacy for the server's identity is ensured by the prefix-based encryption
scheme. Privacy for the client's identity is ensured since all handshake messages
are encrypted
under handshake traffic keys $\htk$ and $\htk'$. We formally state and prove mutual
privacy for the protocol in Appendix~\ref{app:discovery-protocol-privacy}.

\para{Security theorem.} We conclude by stating the security theorem for our
private service discovery protocol. We give the formal proof
in Appendix~\ref{app:discovery-protocol-security}.
\begin{theorem}[Private Service Discovery]
  \label{thm:private-discovery}
  The protocol in Figure~\ref{fig:discovery-proto} is a secure and private
  service discovery protocol in a Canetti-Krawczyk-based model of
  key-exchange in the random oracle model, assuming the Hash Diffie-Hellman
  and Strong Diffie-Hellman assumptions in $\bbG$,
  and the security of the underlying
  cryptographic primitives.
\end{theorem}


\section{Protocol Evaluation and Deployment}\label{sec:experiments}

In this section, we describe the implementation and deployment of our private
mutual authentication and service discovery protocols in the Vanadium
framework~\cite{Vanadium}.
We benchmark our protocols
on a wide range of architectures: an Intel Edison (0.5GHz Intel Atom),
a Raspberry Pi~2 (0.9GHz ARM Cortex-A7), a Nexus 5X smartphone
(1.8GHz 64-bit ARM-v8A), a Macbook Pro
(3.1GHz Intel Core i7), and a desktop (3.2GHz Intel Xeon).

\para{Vanadium.} We implement our private mutual authentication and service
discovery protocols as part of the Vanadium
framework for developing secure, distributed applications that can run on a multitude of devices
ranging from small computers such as Raspberry 
Pis and Intel Edisons, to computers running Linux or Mac OS, to cloud services like 
Google Compute Engine. At the core of the framework is a remote procedure call (RPC) 
system that enables applications to expose services over the network. All RPCs are
mutually authenticated, encrypted and bidirectional. The framework also offers 
an interface definition language (IDL) for defining services, a federated naming system for 
addressing them, and a discovery API for advertising and scanning for services over a
variety of protocols, including  BLE and mDNS.

The Vanadium identity model is based on a distributed PKI. 
All principals in Vanadium possess an ECDSA \mbox{P-256} signing and verification key-pair. 
Principals have a set of human-readable names bound to them via certificate chains,
called \emph{blessings}. Blessings can be extended locally and delegated from one
principal to another. All RPCs are encrypted and mutually authenticated based on the
blessings bound to each end. The access control model is based on blessing names,
and is specified in~\cite{ABPSST15}.

We implement our protocols to enhance the privacy of the
Vanadium RPC and discovery framework. Our entire implementation is in Go\footnote{\url{https://golang.org}}
(with wrappers for interfacing with third-party C libraries).
Go is also the primary
implementation language of Vanadium.

\subsection{Identity-Based Encryption}

The key primitive we require for our protocols is prefix-based encryption,
which we can construct from any IBE scheme (Section~\ref{sec:crypto-primitives}).
For our experiments, we
implemented the Boneh-Boyen ($\ms{BB}_2$) scheme~\cite[\S 5]{BB04a} over the
256-bit Barreto-Naehrig ({\tt bn256})~\cite{NNS10} pairings curve.
We chose the $\ms{BB}_2$ IBE scheme for its efficiency: it
only requires a single pairing evaluation during decryption. Other IBE
schemes either require pairing operations in both encryption and
decryption~\cite{BF01} or require multiple pairing evaluations
during decryption ($\ms{BB}_1$)~\cite[\S4.1]{BB04a}.
We apply the Fujisaki-Okamoto transformation\footnote{The
Fujisaki-Okamoto transformation~\cite{FO99} provides a generic method of
combining a semantically-secure public-key encryption scheme with a
symmetric encryption scheme to obtain a CCA-secure hybrid encryption
scheme in the random oracle model.}~\cite{FO99} to obtain
CCA-security. For the underlying symmetric encryption scheme in the Fujisaki-Okamoto
transformation, we use the authenticated encryption scheme from NaCl~\cite{Ber09,BLS12}.
All of our cryptographic primitives are chosen to
provide at least 128 bits of security.

\para{Benchmarks.} We measured the performance overhead of the
different IBE operations, as well as the pairing operation over the elliptic curve.
Our results are summarized in Table~\ref{tab:ibe-bench}. 
We use the off-the-shelf {\tt dclxvi} implementation~\cite{NNS10} 
for the {\tt bn256} pairing curve on all architectures other than the Nexus 5X. On the Nexus 5X, we use a Go
implementation.\footnote{We were not able to use
the {\tt dclxvi} implementation on the Android phone. A significant speedup should be
possible with an optimized implementation.}
On the desktop, we were able to turn on assembly optimizations which
led to a significant speedup in the benchmarks.

\para{Deployment.} 
The Vanadium ecosystem includes \emph{identity providers}, which are principals that sign the root
certificate of a blessing (certificate chain). While any principal can become an identity provider, 
Vanadium runs a prototype
identity provider for granting blessings to users based on their Google account. These blessings begin with the
name {\tt dev.v.io}. We augment this service so that it also acts as an IBE root and issues
prefix encryption secret keys. More specifically, we run an RPC service for exchanging any 
Vanadium blessing prefixed with {\tt dev.v.io} (e.g., {\tt dev.v.io/u/Alice/Devices/TV}) 
with a  prefix encryption secret key for the blessing name.

\begin{table}
\begin{center}
\begin{tabular}{lrrrrr}
\toprule
 & \hspace{0.5em} Intel Edison & \hspace{0.5em} Raspberry Pi 2
 & \hspace{0.5em} Nexus 5X$^\star$ & \hspace{0.5em} Laptop
 & \hspace{0.5em} Desktop \\ \midrule
Pairing & 254.5 ms & 101.2 ms & 219.2 ms &  5.0 ms & 1.0 ms \\ \midrule
Encrypt & 406.7 ms & 157.2 ms & 161.3 ms &  8.2 ms & 3.5 ms \\
Decrypt & 623.0 ms & 235.2 ms & 235.7 ms & 11.6 ms & 3.7 ms \\
Extract & 107.5 ms &  41.6 ms &  47.2 ms &  2.1 ms & 0.5 ms \\
\bottomrule
\end{tabular}
\end{center}

\footnotesize{
  \begin{adjustwidth}{5cm}{}
    $\star$ Unoptimized Go implementation of {\tt bn256}.
  \end{adjustwidth}
}
\caption{IBE micro-benchmarks}\label{tab:ibe-bench}
\end{table}

\subsection{Private Mutual Authentication}
We implemented the private mutual authentication protocol from 
Section~\ref{sec:private-mutual-auth} within the Vanadium RPC system
as a means to offer a ``private mode'' for 
Vanadium services.
We implemented the protocol from 
Figure~\ref{fig:auth-proto-2} that allows caching of the encrypted server certificate
chain. The implementation uses a 
prefix encryption primitive implemented on top of our IBE library.

\para{Benchmarking.} We measure the end-to-end connection setup time for our protocol
on various platforms. To eliminate network latency, we instantiate a server and client in the
same process. Since the encrypted server certificate chain can be reused across multiple
handshakes, we precompute it before executing the protocol.\footnote{If we use the
variant of our private mutual authentication protocol that provides unlinkability
(Figure~\ref{fig:auth-proto-1}), the server needs to perform a prefix-based
encryption on every handshake. In our case, prefix-based encryption is just
IBE encryption. Based on the micro-benchmarks from Table~\ref{tab:ibe-bench}, the
extra (per-policy) overhead is just a few hundred milliseconds on the small devices and under
10 milliseconds on the laptops and desktops.} Both the client and the server use
a prefix-based policy of length three. Note that the encryption and decryption times in our prefix encryption 
scheme are not affected by the length of the policy. 

\para{Results.} We compare the performance of our protocol to the traditional
SIGMA-I protocol in Table~\ref{tab:mutual-auth-bench}. 
The end-to-end latency on the desktop is only
9.5~ms, thanks to an assembly-optimized IBE implementation. The latency on smaller
devices is typically around a third of a second,
which is quite suitable for user-interactive applications like
AirDrop. Even on the Intel Edisons (a processor marketed specifically for IoT),
the handshake completes in just over 1.5 s, which is still reasonable for
many applications. Moreover, with an optimized implementations of the IBE library (e.g.,
taking advantage of assembly optimizations like on the desktop), these latencies
should be significantly reduced.

The memory and storage requirements of our protocol are very modest and
well-suited for the computational constraints of IoT and mobile devices. Specifically, the
pairing library is just 40 KB of code on the ARM processors (and 64 KB on x86).
The public parameters for the IBE scheme are 512 bytes, and each IBE secret key is just
$160$ bytes. For comparison, a typical certificate chain (of length~3) is about
$500$ bytes in Vanadium. Also, our protocols are not memory-bound,
and in particular, do not require much additional memory on top of the existing
non-private SIGMA-I key-exchange protocol supported by Vanadium.

\para{Deployment.} We implemented the private mutual authentication protocol from
Section~\ref{sec:private-mutual-auth} within the Vanadium RPC system.
Prior to our work, the RPC system
only supported the \mbox{SIGMA-I}~\cite{Kra03,CK02}
protocol for mutual authentication, which does not provide
privacy for services. We implemented our
private mutual authentication protocol to offer a ``private mode'' for 
Vanadium \ \  services.

\begin{table}
\begin{center}
\begin{tabular}{lrrrrr}
\toprule
 & \hspace{0.5em} Intel Edison & \hspace{0.5em} Raspberry Pi 2 & \hspace{0.5em} Nexus 5X
 & \hspace{0.5em} Laptop & \hspace{0.5em} Desktop \\ 
\midrule
SIGMA-I              &  252.1 ms &  88.0 ms &  91.6 ms &  6.3 ms & 5.3 ms \\
Private Mutual Auth. & 1694.3 ms & 326.1 ms & 360.4 ms & 19.6 ms & 9.5 ms \\ \midrule
Slowdown & 6.7x & 3.7x & 3.9x & 3.1x & 1.8x \\
\bottomrule
\end{tabular}
\end{center}
\caption{Private mutual authentication benchmarks.}\label{tab:mutual-auth-bench}
\end{table}

\subsection{Private Discovery}
We also integrated the private discovery protocol from Section~\ref{sec:discovery-protocol}
into Vanadium.

\para{Benchmarks.} 
We benchmark the cryptographic
overhead of processing service advertisements, and measure the size of the service advertisements.
Processing service advertisements requires a single IBE decryption
and one ECDSA signature verification. This cost can be estimated analytically from
Table~\ref{tab:ibe-bench} and standard ECDSA benchmarks. For instance, 
on the Nexus 5X smartphone, which is a typical client for processing service advertisements, the cost 
is approximately $\text{236 ms (IBE decryption)} + \text{11 ms (ECDSA signature verification)} = \text{247 ms}$.

The advertisement size can also be estimated analytically.
From Figure~\ref{fig:discovery-proto}, our service advertisement has the following form:
 \vspace{-0.4em} \[\vspace{-0.4em} \bid, \peenc(\pi_S, (\id_S, g^s, \sig_S(\bid, \id_S, g^s))) \]
Our implementation of prefix encryption ($\peenc$) has a ciphertext
overhead\footnote{This 
includes the size of the $\mathsf{BB}_2$ ciphertext and the
overhead of the Fujisaki-Okamoto transformation.
The {\tt dclxvi} library uses an optimized
floating-based representation~\cite{NNS10} of points on the elliptic curve, which is not very
space-efficient. By representing the points directly as one (using point
compression~\cite{Jiv14}) or two field elements, it is possible to significantly
reduce this overhead.} of $208$ bytes on top of
the plaintext. The Diffie-Hellman exponent ($g^s$) is $32$ bytes, the broadcast id ($\bid$) is $16$ bytes,
the ECDSA signature is $64$ bytes, and a certificate chain ($\id_S$) of length three is 
approximately $500$ bytes in size. The overall service advertisement is about $820$ bytes.

\para{Deployment.}
We deploy our service discovery protocol within the Vanadium discovery framework.
The protocol allows services to advertise themselves while restricting visibility 
to an authorized set of clients. The Vanadium discovery API allows services
to advertise over both mDNS and BLE. An mDNS {\tt TXT} record has a maximum size of
$1300$ bytes~\cite{CK13a,CK13b}, which suffices for service advertisements.
Some care must be taken to
ensure that the mDNS hostname and service label do no leak any private information
about the service. For instance, we set the service label to
{\tt \_vanadium\_private.\_tcp.local} and hide the actual service type and instance
name.

When the policy has multiple prefixes, our advertisements would
no longer fit in a single mDNS {\tt TXT} record. Furthermore, BLE advertisement payloads are
restricted to 31 bytes~\cite{BLE}, which is far too small to fit a full
service advertisement.

We can address this by running an auxiliary service (over GATT or peer-to-peer
WiFi as in Apple AirDrop) to host the encrypted advertisement for the main
service, and advertise the endpoint of this auxiliary service over BLE. The
privacy of the main service is unaffected since the endpoint of the
auxiliary service reveals nothing about it.

Encrypted service advertisements can only be interpreted by authorized
clients. Thus, these advertisements can also be served by members of the
network other than the service itself.  A volunteering node may choose to
cache the encrypted private service advertisements it sees on the network
(for example, any advertisement within radio range) and serve them to
clients.  Clients can fetch advertisements for multiple services in the
network in a single round-trip to the volunteer, without having to connect
to each advertising device.  While such volunteer caching nodes imply
additional infrastructure, the key point is that their presence improves
overall communication and power efficiency in the network without any loss
of privacy to the advertising services or their clients.

\subsection{Fixing AirDrop}
\label{sec:fixing-airdrop}
Recall from Section~\ref{airdrop-protocol}  that during an AirDrop file exchange in 
contacts-only mode, a hash of the sender's identity is
advertised over BLE and matched by potential
receivers against their contacts. If there is a match, 
the receiver starts a service that the sender can connect to using TLS (version 1.2).
In the TLS handshake, the sender and receiver exchange their certificates in the clear, which makes
them visible to eavesdroppers on the network. This privacy vulnerability
can be fixed using the private mutual authentication protocol
from Section~\ref{sec:private-mutual-auth}.
In particular, once the receiver matches the sender's hash against
one of its contacts, it uses the prefix encryption scheme to encrypt
its identity under the name of the
contact that matched the sender's hash.
This ensures that the receiver's identity is only shared with the intended
contact, even in the presence of active network attackers.\footnote{
The privacy guarantee for the sender is not as 
robust as information about its identity leaks via the hash advertised over BLE.}
To deploy this protocol, Apple would provision all AirDrop-enabled devices with
a secret key extracted for their identity in addition to their existing
iCloud certificates. Apple would be the IBE root in this case.


\section{Extensions}
\label{sec:extensions}

In this section, we describe several ways to extend our private mutual
authentication and service discovery protocols to provide additional
properties such as policy-hiding, delegation of decryption capabilities,
and support for more complex authorization policies.

\para{Policy-hiding authentication and discovery.} A limitation
of our private mutual authentication and service discovery protocols from
Sections~\ref{sec:private-mutual-auth} and~\ref{sec:discovery-protocol} is
that the server's authorization policy is revealed to a network attacker.
One solution to this problem is to use a ``policy-hiding'' prefix encryption
scheme, where ciphertexts hide the prefix to which they are encrypted. This
is possible if we build the prefix encryption scheme from an {\em anonymous}
IBE scheme such as Boneh-Franklin~\cite{BF01}. In anonymous
IBE, ciphertexts do not reveal their associated identity, and thus, we achieve
policy-hiding encryption. A drawback of the Boneh-Franklin scheme is that it 
requires evaluating a pairing during {\em both} encryption and decryption,
thus increasing the computational cost over the $\ms{BB}_2$ IBE scheme, where a
pairing is only needed during decryption.

\para{Delegating decryption abilities.} In our current framework, only
IBE roots can issue identity keys to principals. But there are many natural
scenarios where non-root identities should also be able to issue identity
keys. For example, an identity provider might issue Alice a certificate chain for
the name {\tt IDP/User/Alice}. While Alice can now issue
certificates for child domains like {\tt IDP/User/Alice/Phone},
she cannot extract an identity key for the child domain. As a result,
delegation of decryption capabilities must still go through the IBE
root for Alice's domain. To allow Alice to issue identity keys corresponding
to derivatives (i.e., further extensions) of her name, we can use a
hierarchical IBE scheme~\cite{GS02,HL02} to construct the underlying
prefix encryption scheme. In hierarchical IBE, a principal holding
a key for an identity {\tt A/B/} is able to issue keys for any
derivative identity {\tt A/B/C}. This allows 
principals to independently manage their own subdomain without needing
to go through a central identity provider.

\para{Complex authorization policies.} Our private mutual authentication
and service discovery protocols support authorization policies that can be
formulated as prefix constraints. While these are a very natural class of
policies for the SDSI model, there are scenarios that may require
 a more diverse set of policies. For example, a
server's authorization policy might require that the client's user ID be
between 0 and 100. While such range queries can be expressed using a
(large) number of prefixes, the size of the server's broadcast message
grows with the number of prefixes in the policy.
One way of supporting a more complex set of policies is to use a more
general encryption scheme, such as an attribute-based encryption
(ABE)~\cite{SW05,BSW07,GVW13,BGGHNSVV14,GVW15}.
For instance, if the policies can be expressed as
Boolean formulas, then the Bethencourt et al.\ ABE scheme~\cite{BSW07} provides
an efficient solution from pairings. Supporting more complicated policies is
possible using lattice-based methods~\cite{GVW13,BGGHNSVV14,GVW15}, but this
increased expressivity comes at the expense of performance.


\section{Related Work}
\label{sec:related-work}
In this section, we survey some of the existing work on developing
private mutual authentication and service discovery protocols.

\para{Private mutual authentication.} The term ``private authentication" was first introduced
by Abadi and Fournet~\cite{Aba03,AF04}.
These papers define a privacy property for
authentication protocols and the challenges associated with achieving it. 
Their definition of privacy is similar to the one we describe in
Section~\ref{sec:desired-features} and states that the identity of each
protocol participant is revealed to its peer only if the peer is authorized by
the participant. In~\cite{AF04}, Abadi and Fournet introduce two private
authentication protocols and formally analyze them in the applied pi calculus.
In addition to showing privacy, they also show that their protocol achieves
unlinkability. While our private mutual authentication protocol shown in
Figure~\ref{fig:auth-proto-2} does not provide unlinkability, the
variant described in Section~\ref{sec:private-mutual-auth-analysis}
(shown in Figure~\ref{fig:auth-proto-1}) does achieve unlinkability.

A key difference between our work and~\cite{AF04} is the fact that we support
prefix-based authorization policies.
The protocols in~\cite{AF04} require the
authorization policy to be specified by a set of public keys and do not scale
when the set of public keys is very large.
Moreover, in their protocol, clients and servers must either maintain
(potentially long) lists of public keys for authorized peers, or look up a
peer's public key in a directory service before each authentication request.
In contrast, using IBE both enables support for more expressive authorization
policies and eliminates the need to either maintain lists of public keys or
run a public-key directory service.

Many other cryptographic primitives have also been developed for problems related to private
mutual authentication, including secret handshakes~\cite{BDSSSW03,JKT06,AKB07}, oblivious
signature-based envelopes~\cite{LDB05}, oblivious attribute
certificates~\cite{LL06}, hidden credentials~\cite{HBSO03,FAL04,BHS04}, and
more.

Secret handshakes and their extensions are
protocols based on bilinear pairings that allow members of a group to identify each
other privately. A key limitation of secret handshakes is that the parties can only
authenticate using credentials issued by the same root authority. 
Oblivious signature-based envelopes~\cite{LDB05}, oblivious attribute
certificates~\cite{LL06} and hidden credentials~\cite{HBSO03,FAL04,BHS04}
allow a sender to send an encrypted message that can be decrypted only by a
recipient that satisfies some policy. Hidden credentials additionally 
hide the sender's policy. Closely related are the cryptographic primitives of
attribute-based encryption~\cite{BSW07,GVW13} and predicate
encryption~\cite{KSW08,GVW15}, which allow more fine-grained control over
decryption capabilities.

The protocols we have surveyed here are meant for {\em authentication},
and not authenticated key-exchange, which is usually the desired primitive.
Integrating these protocols into existing key-exchange protocols such as
SIGMA or TLS~1.3 is not always straightforward and can require non-trivial
changes to existing protocols. In contrast, our work shows how IBE-based
authentication can be very naturally integrated with existing secure
key-exchange protocols (with minimal changes) to obtain private mutual
authentication. Moreover, our techniques are equally applicable in the service
discovery setting, and can be used to obtain 0-RTT private mutual authentication.

\para{Service discovery.} There is a large body of work on designing
service discovery protocols for various environments---mobile, IoT, enterprise
and more; we refer to~\cite{ZMN05} for a survey. Broadly, these 
protocols can be categorized into two groups:
``directory-based" protocols and  ``directory-free'' protocols.

In directory-based discovery protocols~\cite{CZHJK99,ZMBDLC10,Jini}, 
there is a central directory that maintains service information and controls access
to the services. Clients query directories to discover services while services register 
with the directory to announce their presence.
While directory-based protocols allow for centralized management and tend
to be computationally efficient, their main drawback is that they force
dependence on an external service. If the directory service is unavailable then
the protocol ceases to work. Even worse, if the directory service is compromised, then
both server and client privacy is lost. Besides, mutually suspicious clients and 
servers may not be able to agree on a common directory service that they both trust.
In light of these downsides, we designed decentralized, peer-to-peer protocols in this work.

Directory-free protocols, such as~\cite{KM14a, KM14b, ZMN04, ZMN06}, typically rely 
on a shared key established between devices in a separate, out-of-band protocol. The shared
key is then used to encrypt the private service advertisements so that only paired
devices can decrypt. Other protocols like UPnP~\cite{Ell02}
rely on public key encryption, where each device maintains a set of public keys for 
the peers it is willing to talk to. In contrast, key-management in our IBE-based 
solution is greatly simplified---devices do not have
to maintain long lists of symmetric or public keys. Our protocol
is similar to the Tryst protocol~\cite{PGMSW07}, which proposes using an
anonymous IBE scheme for encrypting under the peer's name (based on using  
a mutually agreed upon convention). A distinguishing feature of our protocol over Tryst
is the support for prefix-based authorization policies.


\section{Conclusion}

Automatic service discovery is an integral component of the Internet of
Things. While numerous service discovery protocols exist, few
provide any notion of privacy. 
Motivated by the privacy shortcomings of
the Apple AirDrop protocol, we introduce several important security and
privacy goals for designing service discovery protocols for both IoT and mobile applications.
We then show how to combine off-the-shelf identity-based encryption with
existing key-exchange protocols to obtain both a private mutual authentication
and a private service discovery protocol.
Our benchmarks on the various devices
show that our protocols are viable for a wide range of practical 
deployments and IoT scenarios.


\section*{Acknowledgments}

We thank  Mart\'{i}n Abadi, Mike Burrows, Felix G{\"u}nther, and Adam Langley for many helpful
comments and suggestions. We thank Bruno Blanchet for his help in verifying
the unlinkability property of our modified private mutual authentication
protocol in Section~\ref{sec:private-mutual-auth}. We thank Felix
G{\"u}nther for pointing out an error in an earlier version of our private discovery
protocol. This work was supported by NSF, DARPA, a grant from ONR, the Simons Foundation,
and an NSF Graduate Research Fellowship.
Opinions, findings and conclusions or recommendations expressed in this material
are those of the authors and do not necessarily reflect the views of DARPA.

\bibliographystyle{alpha}
\bibliography{paper}

\newcommand{\etalchar}[1]{$^{#1}$}
\begin{thebibliography}{PGM{\etalchar{+}}07}

\bibitem[Aba03]{Aba03}
Mart{\'i}n Abadi.
\newblock Private authentication.
\newblock In {\em {PETS}}, pages 27--40, 2003.

\bibitem[ABB{\etalchar{+}}04]{ABBCIKR04}
William Aiello, Steven~M. Bellovin, Matt Blaze, Ran Canetti, John Ioannidis,
  Angelos~D. Keromytis, and Omer Reingold.
\newblock Just fast keying: Key agreement in a hostile internet.
\newblock {\em {ACM} Trans. Inf. Syst. Secur.}, 7(2):242--273, 2004.

\bibitem[ABP{\etalchar{+}}15]{ABPSST15}
Mart{\'{\i}}n Abadi, Mike Burrows, Himabindu Pucha, Adam Sadovsky, Asim
  Shankar, and Ankur Taly.
\newblock Distributed authorization with distributed grammars.
\newblock In {\em Programming Languages with Applications to Biology and
  Security}, pages 10--26, 2015.

\bibitem[ABR01]{ABR01}
Michel Abdalla, Mihir Bellare, and Phillip Rogaway.
\newblock The oracle diffie-hellman assumptions and an analysis of {DHIES}.
\newblock In {\em {CT-RSA}}, pages 143--158, 2001.

\bibitem[AF04]{AF04}
Mart\'{\i}n Abadi and C{\'e}dric Fournet.
\newblock Private authentication.
\newblock {\em Theoretical Computer Science}, 322:427--476, 2004.

\bibitem[AKB07]{AKB07}
Giuseppe Ateniese, Jonathan Kirsch, and Marina Blanton.
\newblock Secret handshakes with dynamic and fuzzy matching.
\newblock In {\em {NDSS}}, 2007.

\bibitem[BB04]{BB04a}
Dan Boneh and Xavier Boyen.
\newblock Efficient selective-id secure identity-based encryption without
  random oracles.
\newblock In {\em {EUROCRYPT}}, pages 223--238, 2004.

\bibitem[BCK98]{BCK98}
Mihir Bellare, Ran Canetti, and Hugo Krawczyk.
\newblock A modular approach to the design and analysis of authentication and
  key exchange protocols (extended abstract).
\newblock In {\em {STOC}}, pages 419--428, 1998.

\bibitem[BDS{\etalchar{+}}03]{BDSSSW03}
Dirk Balfanz, Glenn Durfee, Narendar Shankar, Diana~K. Smetters, Jessica
  Staddon, and Hao{-}Chi Wong.
\newblock Secret handshakes from pairing-based key agreements.
\newblock In {\em 2003 {IEEE} S{\&}P 2003}, pages 180--196, 2003.

\bibitem[Ber09]{Ber09}
Daniel~J. Bernstein.
\newblock {Cryptography in {NaCl}}, 2009.

\bibitem[BF01]{BF01}
Dan Boneh and Matthew~K. Franklin.
\newblock Identity-based encryption from the weil pairing.
\newblock In {\em {CRYPTO}}, pages 213--229, 2001.

\bibitem[BGG{\etalchar{+}}14]{BGGHNSVV14}
Dan Boneh, Craig Gentry, Sergey Gorbunov, Shai Halevi, Valeria Nikolaenko, Gil
  Segev, Vinod Vaikuntanathan, and Dhinakaran Vinayagamurthy.
\newblock Fully key-homomorphic encryption, arithmetic circuit {ABE} and
  compact garbled circuits.
\newblock In {\em {EUROCRYPT}}, pages 533--556, 2014.

\bibitem[BHS04]{BHS04}
Robert~W. Bradshaw, Jason~E. Holt, and Kent~E. Seamons.
\newblock Concealing complex policies with hidden credentials.
\newblock In {\em {ACM} {CCS}}, pages 146--157, 2004.

\bibitem[BLE14]{BLE}
Bluetooth specification version 4.2, 2014.

\bibitem[BLS12]{BLS12}
Daniel~J. Bernstein, Tanja Lange, and Peter Schwabe.
\newblock The security impact of a new cryptographic library.
\newblock In {\em {LATINCRYPT}}, pages 159--176, 2012.

\bibitem[BM82]{BM82}
Manuel Blum and Silvio Micali.
\newblock How to generate cryptographically strong sequences of pseudo random
  bits.
\newblock In {\em {FOCS}}, pages 112--117, 1982.

\bibitem[BN00]{BN00}
Mihir Bellare and Chanathip Namprempre.
\newblock Authenticated encryption: Relations among notions and analysis of the
  generic composition paradigm.
\newblock In {\em {ASIACRYPT}}, pages 531--545, 2000.

\bibitem[Bon13]{Bonjour}
Bonjour printing specification version 1.2, 2013.

\bibitem[BR93a]{BR93}
Mihir Bellare and Phillip Rogaway.
\newblock Entity authentication and key distribution.
\newblock In {\em {CRYPTO}}, pages 232--249, 1993.

\bibitem[BR93b]{BR93a}
Mihir Bellare and Phillip Rogaway.
\newblock Random oracles are practical: {A} paradigm for designing efficient
  protocols.
\newblock In {\em {ACM} {CCS}}, pages 62--73, 1993.

\bibitem[Bro09]{Bro09}
Daniel R.~L. Brown.
\newblock {{SEC 1}: Elliptic Curve Cryptography}, 2009.

\bibitem[BRW03]{BRW03}
Mihir Bellare, Phillip Rogaway, and David Wagner.
\newblock {EAX:} {A} conventional authenticated-encryption mode.
\newblock {\em {IACR} Cryptology ePrint Archive}, 2003:69, 2003.

\bibitem[BSC16]{BSC16}
Bruno Blanchet, Ben Smyth, and Vincent Cheval.
\newblock Proverif 1.93: Automatic cryptographic protocol verifier, user manual
  and tutorial, 2016.

\bibitem[BSW07]{BSW07}
John Bethencourt, Amit Sahai, and Brent Waters.
\newblock Ciphertext-policy attribute-based encryption.
\newblock In {\em {IEEE} S{\&}P}, pages 321--334, 2007.

\bibitem[CK01]{CK01}
Ran Canetti and Hugo Krawczyk.
\newblock Analysis of key-exchange protocols and their use for building secure
  channels.
\newblock In {\em {EUROCRYPT}}, pages 453--474, 2001.

\bibitem[CK02]{CK02}
Ran Canetti and Hugo Krawczyk.
\newblock Security analysis of {IKE}'s signature-based key-exchange protocol.
\newblock In {\em {CRYPTO}}, pages 143--161, 2002.

\bibitem[CK13a]{CK13a}
S.~Cheshire and M.~Krochmal.
\newblock {{DNS}-Based Service Discovery}.
\newblock RFC 6763 (Proposed Standard), February 2013.

\bibitem[CK13b]{CK13b}
S.~Cheshire and M.~Krochmal.
\newblock {Multicast {DNS}}.
\newblock RFC 6762 (Proposed Standard), February 2013.

\bibitem[Coc01]{Coc01}
Clifford Cocks.
\newblock An identity based encryption scheme based on quadratic residues.
\newblock In {\em Cryptography and Coding}, pages 360--363, 2001.

\bibitem[CZH{\etalchar{+}}99]{CZHJK99}
Steven~E. Czerwinski, Ben~Y. Zhao, Todd~D. Hodes, Anthony~D. Joseph, and
  Randy~H. Katz.
\newblock An architecture for a secure service discovery service.
\newblock In {\em MobiCom}, pages 24--35, 1999.

\bibitem[DGH{\etalchar{+}}04]{DGHKR04}
Yevgeniy Dodis, Rosario Gennaro, Johan H{\aa}stad, Hugo Krawczyk, and Tal
  Rabin.
\newblock Randomness extraction and key derivation using the cbc, cascade and
  {HMAC} modes.
\newblock In {\em {CRYPTO}}, pages 494--510, 2004.

\bibitem[DR08]{DR08}
T.~Dierks and E.~Rescorla.
\newblock {The Transport Layer Security ({TLS}) Protocol Version 1.2}.
\newblock RFC 5246 (Proposed Standard), August 2008.

\bibitem[DY81]{DY81}
Danny Dolev and Andrew~Chi{-}Chih Yao.
\newblock On the security of public key protocols.
\newblock In {\em {FOCS}}, pages 350--357, 1981.

\bibitem[Ell02]{Ell02}
Carl~M. Ellison.
\newblock Home network security.
\newblock {\em Intel Technology Journal}, 6(4):37--48, 2002.

\bibitem[FAL04]{FAL04}
Keith~B. Frikken, Mikhail~J. Atallah, and Jiangtao Li.
\newblock Hidden access control policies with hidden credentials.
\newblock In {\em {ACM} {WPES}}, page~27, 2004.

\bibitem[FG14]{FG14}
Marc Fischlin and Felix G{\"{u}}nther.
\newblock Multi-stage key exchange and the case of google's {QUIC} protocol.
\newblock In {\em {ACM} {CCS}}, pages 1193--1204, 2014.

\bibitem[FO99]{FO99}
Eiichiro Fujisaki and Tatsuaki Okamoto.
\newblock Secure integration of asymmetric and symmetric encryption schemes.
\newblock In {\em {CRYPTO}}, pages 537--554, 1999.

\bibitem[GGM84]{GGM84}
Oded Goldreich, Shafi Goldwasser, and Silvio Micali.
\newblock How to construct random functions (extended abstract).
\newblock In {\em {FOCS}}, pages 464--479, 1984.

\bibitem[GS02]{GS02}
Craig Gentry and Alice Silverberg.
\newblock Hierarchical id-based cryptography.
\newblock In {\em {ASIACRYPT}}, pages 548--566, 2002.

\bibitem[GVW13]{GVW13}
Sergey Gorbunov, Vinod Vaikuntanathan, and Hoeteck Wee.
\newblock Attribute-based encryption for circuits.
\newblock In {\em {STOC}}, pages 545--554, 2013.

\bibitem[GVW15]{GVW15}
Sergey Gorbunov, Vinod Vaikuntanathan, and Hoeteck Wee.
\newblock Predicate encryption for circuits from {LWE}.
\newblock In {\em {CRYPTO}}, pages 503--523, 2015.

\bibitem[HBSO03]{HBSO03}
Jason~E. Holt, Robert~W. Bradshaw, Kent~E. Seamons, and Hilarie~K. Orman.
\newblock Hidden credentials.
\newblock In {\em {ACM} {WPES}}, pages 1--8, 2003.

\bibitem[HL02]{HL02}
Jeremy Horwitz and Ben Lynn.
\newblock Toward hierarchical identity-based encryption.
\newblock In {\em {EUROCRYPT}}, pages 466--481, 2002.

\bibitem[Inc15]{Apple-iOS15}
Apple Inc.
\newblock The transport layer security ({TLS}) protocol version 1.3, September
  2015.

\bibitem[Jin13]{Jini}
{Jini(TM)} network technology specifications -- {Apache} river version 2.2.0,
  2013.

\bibitem[Jiv14]{Jiv14}
A.~Jivsov.
\newblock Compact representation of an elliptic curve point, March 2014.

\bibitem[JKT06]{JKT06}
Stanislaw Jarecki, Jihye Kim, and Gene Tsudik.
\newblock Authentication for paranoids: Multi-party secret handshakes.
\newblock In {\em {ACNS}}, pages 325--339, 2006.

\bibitem[KBSW13]{KBSW13}
Bastian K{\"{o}}nings, Christoph Bachmaier, Florian Schaub, and Michael Weber.
\newblock Device names in the wild: Investigating privacy risks of zero
  configuration networking.
\newblock In {\em {IEEE} {MDM}}, pages 51--56, 2013.

\bibitem[KE10]{KE10}
H.~Krawczyk and P.~Eronen.
\newblock {HMAC}-based extract-and-expand key derivation function ({HKDF}),
  2010.

\bibitem[KM13]{KM13}
Neal Koblitz and Alfred Menezes.
\newblock Another look at security definitions.
\newblock {\em Adv. in Math. of Comm.}, 7(1):1--38, 2013.

\bibitem[KPW13]{KPW13}
Hugo Krawczyk, Kenneth~G. Paterson, and Hoeteck Wee.
\newblock On the security of the {TLS} protocol: {A} systematic analysis.
\newblock In {\em {CRYPTO}}, pages 429--448, 2013.

\bibitem[Kra03]{Kra03}
Hugo Krawczyk.
\newblock {SIGMA:} the '{SIGn-and-MAc}' approach to authenticated
  diffie-hellman and its use in the ike-protocols.
\newblock In {\em {CRYPTO}}, pages 400--425, 2003.

\bibitem[Kra10]{Kra10}
Hugo Krawczyk.
\newblock Cryptographic extraction and key derivation: The {HKDF} scheme.
\newblock In {\em {CRYPTO}}, pages 631--648, 2010.

\bibitem[KSW08]{KSW08}
Jonathan Katz, Amit Sahai, and Brent Waters.
\newblock Predicate encryption supporting disjunctions, polynomial equations,
  and inner products.
\newblock In {\em {EUROCRYPT}}, 2008.

\bibitem[KW14a]{KM14a}
D.~Kaiser and M.~Waldvogel.
\newblock Adding privacy to multicast dns service discovery.
\newblock In {\em {IEEE} TrustCom}, pages 809--816, 2014.

\bibitem[KW14b]{KM14b}
D.~Kaiser and M.~Waldvogel.
\newblock Efficient privacy preserving multicast dns service discovery.
\newblock In {\em {IEEE} {CSS}}, 2014.

\bibitem[KW15]{KW15}
Hugo Krawczyk and Hoeteck Wee.
\newblock The {OPTLS} protocol and {TLS} 1.3.
\newblock {\em {IACR} Cryptology ePrint Archive}, 2015:978, 2015.

\bibitem[LDB05]{LDB05}
Ninghui Li, Wenliang Du, and Dan Boneh.
\newblock Oblivious signature-based envelope.
\newblock {\em Distributed Computing}, 17(4), 2005.
\newblock Extended abstract in ACM PODC 2003.

\bibitem[LJBN15]{LJBN15}
Robert Lychev, Samuel Jero, Alexandra Boldyreva, and Cristina Nita{-}Rotaru.
\newblock How secure and quick is quic? provable security and performance
  analyses.
\newblock In {\em {IEEE} Symposium on Security and Privacy}, pages 214--231,
  2015.

\bibitem[LL06]{LL06}
Jiangtao Li and Ninghui Li.
\newblock Oacerts: Oblivious attribute certificates.
\newblock {\em {IEEE} Trans. Dependable Sec. Comput.}, 3(4):340--352, 2006.

\bibitem[LW14]{LW14}
Allison~B. Lewko and Brent Waters.
\newblock Why proving {HIBE} systems secure is difficult.
\newblock In {\em {EUROCRYPT}}, pages 58--76, 2014.

\bibitem[MS04]{MS04}
Alfred Menezes and Nigel~P. Smart.
\newblock Security of signature schemes in a multi-user setting.
\newblock {\em Des. Codes Cryptography}, 33(3):261--274, 2004.

\bibitem[NNS10]{NNS10}
Michael Naehrig, Ruben Niederhagen, and Peter Schwabe.
\newblock New software speed records for cryptographic pairings.
\newblock In {\em {LATINCRYPT}}, pages 109--123, 2010.

\bibitem[NZ96]{NZ96}
Noam Nisan and David Zuckerman.
\newblock Randomness is linear in space.
\newblock {\em J. Comput. Syst. Sci.}, 52(1):43--52, 1996.

\bibitem[PGM{\etalchar{+}}07]{PGMSW07}
Jeffrey Pang, Ben Greenstein, Damon McCoy, Srinivasan Seshan, and David
  Wetherall.
\newblock Tryst: The case for confidential service discovery.
\newblock In {\em {HotNets}}, 2007.

\bibitem[Res15]{Res15}
E.~Rescorla.
\newblock The transport layer security ({TLS}) protocol version 1.3, July 2015.

\bibitem[RL96]{RL96}
Ronald~L. Rivest and Butler Lampson.
\newblock {SDSI} - a simple distributed security infrastructure.
\newblock Technical report, 1996.

\bibitem[Rog02]{Rog02}
Phillip Rogaway.
\newblock Authenticated-encryption with associated-data.
\newblock In {\em {ACM} {CCS}}, pages 98--107, 2002.

\bibitem[Sha84]{Sha84}
Adi Shamir.
\newblock Identity-based cryptosystems and signature schemes.
\newblock In {\em {CRYPTO}}, pages 47--53, 1984.

\bibitem[Sho99]{Sho99a}
Victor Shoup.
\newblock On formal models for secure key exchange.
\newblock {\em {IACR} Cryptology ePrint Archive}, 1999:12, 1999.

\bibitem[SW05]{SW05}
Amit Sahai and Brent Waters.
\newblock Fuzzy identity-based encryption.
\newblock In {\em {EUROCRYPT}}, pages 457--473, 2005.

\bibitem[UPn15]{UPnP}
{UPnP(TM)} device architecture 2.0, 2015.

\bibitem[Van]{Vanadium}
{Vanadium}.
\newblock \url{http://vanadium.github.io/}.

\bibitem[Yao82]{Yao82a}
Andrew~Chi{-}Chih Yao.
\newblock Theory and applications of trapdoor functions (extended abstract).
\newblock In {\em {FOCS}}, pages 80--91, 1982.

\bibitem[Zer04]{Zeroconf}
{IETF} zero configuration networking (zeroconf).
\newblock \url{https://datatracker.ietf.org/doc/charter-ietf-zeroconf}, 2004.

\bibitem[ZMB{\etalchar{+}}10]{ZMBDLC10}
Feng~W. Zhu, Matt~W. Mutka, Anish Bivalkar, Abdullah Demir, Yue Lu, and
  Chockalingam Chidambarm.
\newblock Toward secure and private service discovery anywhere anytime.
\newblock {\em Frontiers of Computer Science in China}, 4(3):311--323, 2010.

\bibitem[ZMN04]{ZMN04}
Feng~W. Zhu, Matt~W. Mutka, and Lionel~M. Ni.
\newblock {PrudentExposure}: {A} private and user-centric service discovery
  protocol.
\newblock In {\em {IEEE} {PerCom}}, pages 329--340, 2004.

\bibitem[ZMN05]{ZMN05}
Feng~W. Zhu, Matt~W. Mutka, and Lionel~M. Ni.
\newblock Service discovery in pervasive computing environments.
\newblock {\em {IEEE} Pervasive Computing}, 4(4):81--90, 2005.

\bibitem[ZMN06]{ZMN06}
Feng~W. Zhu, Matt~W. Mutka, and Lionel~M. Ni.
\newblock A private, secure, and user-centric information exposure model for
  service discovery protocols.
\newblock {\em {IEEE} Trans. Mob. Comput.}, 5(4):418--429, 2006.

\end{thebibliography}


\appendix

\section{Additional Preliminaries}
\label{app:prelims-add}
In this section, we review some additional preliminaries.

\para{Identity-based encryption.} Identity-based encryption
(IBE)~\cite{Sha84,BF01,Coc01,BB04a} is a generalization of public-key encryption
where the public keys can be arbitrary strings, called {\em identities}.
An IBE scheme consists of four algorithms:
\begin{citemize}
  \item \textbf{Setup.} $\setup(1^\lambda)$ takes the security parameter
  $\lambda$ and outputs a master public key $\mpk$ and a master secret key
  $\msk$.
  \item \textbf{Extract.} $\extract(\msk, \id)$ takes
  $\msk$ and an identity $\id$ and outputs an identity secret key $\skid$.
  \item \textbf{Encrypt.} $\enc(\mpk, \id, m)$ takes $\mpk$, an identity $\id$,
  and a message $m$ and outputs a ciphertext $\ct$.
  \item \textbf{Decrypt.} $\dec(\skid, \ct)$ takes $\ct$
  and $\skid$ and outputs a message $m$ or a special symbol
  $\perp$.
\end{citemize}
The correctness requirement for an IBE scheme states that for all messages $m$,
identities $\id$ and keys $(\mpk, \msk) \gets \setup(1^\lambda)$, if
$\sk \gets \extract(\msk, \id)$ and $\ct \gets \enc(\mpk, \id, m)$, then
$\dec(\sk, \ct) = m$.
In this work, we require IBE schemes that
are secure against adaptive chosen-ciphertext attacks (CCA-secure).

\para{IBE security.} We now present the formal game-based definition for
adaptive-CCA security in the context of an IBE scheme~\cite{BF01}. First, we define the
$\indidcca$ security game between an adversary $\calA$ and a challenger. In
the following, we let $\lambda$ denote the security parameter. The game
consists of several phases:
\begin{citemize}
  \item \textbf{Setup phase:} The challenger samples parameters
  $(\mpk, \msk) \gets \setup(1^\lambda)$, and sends $\mpk$ to $\calA$.

  \item \textbf{Pre-challenge phase:} The adversary can adaptively make
  key extraction and decryption queries to the challenger. In a key-extraction
  query, the adversary submits an identity $\id$. The challenger replies with
  an identity key $\sk_\id \gets \extract(\msk, \id)$. In a decryption query,
  the adversary submits a ciphertext $\ct$ and an identity $\id$. The
  challenger then extracts an identity key $\sk_\id \gets \extract(\msk, \id)$
  and replies with $\dec(\sk_\id, \ct)$.

  \item \textbf{Challenge phase:} Once the adversary has finished making queries
  in the pre-challenge phase, it outputs two equal-length messages $\bar m_0$
  and $\bar m_1$ along with a target identity $\hid$ ($\hid$ must not have
  appeared in any of the adversary's extraction queries in the pre-challenge
  phase). The challenger chooses a random bit $b \getsr \zo$ and replies with
  a ciphertext $\hct \gets \enc(\mpk, \hid, \bar m_b)$.

  \item \textbf{Post-challenge phase:} In the post-challenge phase, the adversary
  can continue to make extraction and decryption queries as in the pre-challenge
  phase. However, it cannot request a key for the target identity
  $\hid$ or request a decryption of the challenge ciphertext $\hct$
  under identity $\hid$.

  \item {\textbf{Output phase:}} At the end of the game, the adversary outputs a
  bit $b' \in \zo$. We say the adversary ``wins'' the $\indidcca$ game if $b' = b$.
\end{citemize}
\begin{definition}[$\indidcca$ Security~\cite{BF01}]
  \normalfont
  \label{def:ibe-cca}
  An IBE scheme $(\setup, \enc, \dec, \extract)$ is $\indidcca$-secure if no
  efficient adversary can win the $\indidcca$ security game with probability
  that is non-negligibly greater than $1/2$.
\end{definition}

\para{Additional cryptographic primitives.} We also review some standard
cryptographic definitions for pseudorandom generators (PRGs), pseudorandom
functions (PRFs), strong randomness extractors, and authenticated encryption. Below,
we write $\negl(\lambda)$ to denote a negligible function in the security
parameter $\lambda$.

\begin{definition}[Pseudorandom Generator~\cite{Yao82a,BM82}]
  \normalfont
  A pseudorandom generator $G : \zo^\ell \to \zo^n$
  where $\ell = \ell(\lambda)$, $n = n(\lambda)$, $\ell < n$ is secure if for all
  efficient adversaries $\calA$,
  \[ \abs{\Pr\left[ s \getsr \zo^\ell : \calA(1^\lambda, G(s)) = 1 \right] -
          \Pr\left[ z \getsr \zo^n : \calA(1^\lambda, z) = 1 \right]} = \negl(\lambda). \]
\end{definition}

\begin{definition}[Pseudorandom Function~\cite{GGM84}]
  \normalfont
  A pseudorandom function $F : \calK \times \zo^n \to \zo^m$
  with key space $\calK$, domain $\zo^n$ and range $\zo^m$ is secure if for all efficient
  adversaries $\calA$,
  \[ \abs{\Pr\left[ k \getsr \calK : \calA^{F(k, \cdot)}(1^\lambda) = 1 \right] -
          \Pr\left[ f \getsr \Funs[\zo^n, \zo^m]: \calA^{f(\cdot)}(1^\lambda) = 1 \right]} = \negl(\lambda), \]
  where $\Funs[\zo^n, \zo^m]$ denotes the set of all functions from $\zo^n$ to $\zo^m$.
\end{definition}

\begin{definition}[{Strong Randomness Extractor~\cite[adapted]{NZ96}}]
  \normalfont
  A function $E : \calK \times \zo^n \to \zo^m$ is a strong
  randomness extractor if for all adversaries $\calA$,
  \begin{multline*}
    \Big|\Pr\left[ k \getsr \calK, x \getsr \zo^n : \calA(1^\lambda, k, E(k, x)) = 1 \right] - \\
          \Pr\left[ k \getsr \calK, z \getsr \zo^m : \calA(1^\lambda, k, z) = 1 \right] \Big| = \negl(\lambda).
  \end{multline*}
\end{definition}

\begin{definition}[Authenticated Encryption~\cite{BN00}]
  \normalfont
  A symmetric-key encryption scheme $(\Enc, \allowbreak \Dec)$ over a key space $\calK$,
  a message space $\calM$, and a ciphertext space $\calC$ is an authenticated encryption
  scheme if it satisfies the following properties:
  \begin{citemize}
    \item \textbf{Correctness:} For all messages $m \in \calM$,
    \[ \Pr\left[ k \getsr \calK : \Dec(k, \Enc(k, m)) = m \right] = 1 - \negl(\lambda). \]

    \item \textbf{Semantic Security ($\indcpa$):} For all efficient adversaries $\calA$,
    \[ \abs{\Pr\left[ k \getsr \calK, b \getsr \zo;
            b' \gets \calA^{\LoR(k, b, \cdot, \cdot)}(1^\lambda) : b' = b \right] - \frac{1}{2}} = \negl(\lambda), \]
    where $\LoR(k, b, m_0, m_1)$ is the ``left-or-right'' encryption
    oracle which on input a key $k \in \calK$, a bit $b \in \zo$, and
    two messages $m_0, m_1 \in \calM$, outputs $\Enc(k, m_b)$.

    \item \textbf{Ciphertext Integrity:} For all efficient adversaries $\calA$,
    \[ \Pr\left[ k \getsr \calK ; \ct \gets \calA^{\Enc(k, \cdot)}(1^\lambda) :
                 \ct \notin Q \wedge \Dec(k, \ct) \neq \bot \right] = \negl(\lambda), \]
    where $Q$ is the set of ciphertexts output by the $\Enc$ oracle.
  \end{citemize}
\end{definition}

\para{Key derivation.} A key derivation function
(KDF)~\cite{DGHKR04,Kra10} is an algorithm for extracting randomness from some
non-uniform entropy source to obtain a uniformly random
string suitable for use in other cryptographic primitives.
For our private mutual authentication protocol
in Figure~\ref{fig:auth-proto-2}, the key-derivation function is implemented
by a hash function: $\kdf(g^x, g^y, g^{xy}) = H(g^x, g^y, g^{xy})$, and
our security analysis relies on the hardness of the Hash
Diffie-Hellman assumption~\cite{ABR01} (described below). For our private
service discovery protocol, we defer an explicit description of our key-derivation
procedure to Appendix~\ref{app:discovery-protocol-security}, where we provide a formal
analysis of our protocol.

\para{Cryptographic assumptions.} We now state the
  Diffie-Hellman-type assumptions~\cite{ABR01} we use to show security of our protocols.
  Let $\bbG$ be a cyclic group of prime order $p$ with generator $g$, and
  let $H: \bbG \to \calZ$ be a hash function from $\bbG$ to an output space
  $\calZ$. Then, the Hash Diffie-Hellman (Hash-DH) assumption in $\bbG$
  states that
  the following distributions are computationally indistinguishable:
  \[ (g, g^x, g^y, H(g^x, g^y, g^{xy})) \quad \text{and} \quad
     (g, g^x, g^y, z), \]
  where $x, y \getsr \Z_p$
  and $z \getsr \calZ$. The Strong Diffie-Hellman (Strong-DH)
  assumption in $\bbG$ states that computing $g^{xy}$
  is hard given $(g, g^x, g^y)$ where $x, y \getsr \Z_p$, even if the
  adversary is given access to a verification oracle $\calO(g,
  g^x, \cdot, \cdot)$. The verification oracle
  $\calO$ outputs $1$ if the input $(g, g^x, g^y, g^z)$ is a
  decisional Diffie-Hellman (DDH) tuple ($z = xy
  \bmod p$) and $0$ otherwise. In other words, the Strong-DH assumption states
  that the computational Diffie-Hellman (CDH) problem is hard even
  with access to a (restricted) DDH oracle.

\section{Key-Exchange Security}
\label{app:key-exchange-model}
In this section, we describe the Canetti-Krawczyk model of key-exchange (KE)
in the ``post-specified'' peer setting~\cite{CK02,Kra03}. In the
post-specified peer setting, a peer's identity is possibly unavailable at the
beginning of the KE protocol. Instead, the identity of the peer is revealed
during the KE protocol. For example, at the beginning of the KE protocol, the
client might only know the server's IP address and nothing more about the
identity of the server. By participating in the KE protocol, the client learns
the identity of the server. Earlier works on key exchange have also
considered the ``pre-specified'' peer setting~\cite{BR93,Sho99a,CK01}, where
the peer's identity is assumed to be known at the beginning of the protocol.
Our description in this section is largely taken from the description
in~\cite{CK02}.

Following earlier works~\cite{BR93,BCK98}, Canetti and
Krawczyk~\cite{CK01,CK02,Kra03} model a KE protocol as a multi-party protocol
where each party can run one or more instances of the protocol. At the
beginning of the protocol execution experiment, the challenger sets up the
long-term secrets for each party (as prescribed by the protocol
specification). Then, the adversary is given control of the protocol
execution. In particular, the adversary can activate parties to run an
instance of the key-exchange protocol called a {\em session}. Each session is
associated with a session identifier. Within a session, parties can be
activated to either initiate a session, or to respond to an incoming message.
The purpose of the session is to agree on a {\em session key} with another
party in the network (called the {\em peer}) by exchanging a sequence of
messages. Sessions can run concurrently, and messages are directed to the
session with the associated session identifier.

In the post-specified model, a party can be activated to initiate a KE session
by a tuple $(P, \sid, d)$, where $P$ is the party at which the session is
activated, $\sid$ is a unique session identifier, and $d$ is a ``destination
address'' used only for the delivery of messages for session $\sid$. In
particular, $d$ has no bearing on the identity of the peer. A party can be
activated as either an initiator or as a responder (e.g., to respond to an
initialization message from another party). The output of a completed KE session
at a party $P$ consists of a public tuple $(P, \sid, Q)$ as well as a secret
value $\ksk$, where $\sid$ is the session id, $Q$ is the peer of the session,
and $\ksk$ is the application-traffic (or session) key. In addition, sessions
can be {\em aborted} without producing a session key. We denote the output of
an aborted session with a special symbol $\perp$. When the session produces an
output (as a result of completing or aborting), its local state is erased.
Only the tuple $(P, \sid, Q)$ and the application-traffic key $\ksk$ persists at the
conclusion of the session. Each party may also have additional {\em long-term}
state such as a signature signing key. These long-term secrets are shared
across multiple sessions and are not considered to be a part of a party's
local session state. In the protocol analysis, we will uniquely identify
sessions by a pair $(P,
\sid)$, where $P$ is the local party and $\sid$ is the session identifier.

\para{Attacker model.} The attacker is modeled to capture realistic attack
capabilities in an open network. In particular, the attacker has full control
over the network communication, and moreover, is allowed access to some of the
secrets used or generated by the key-exchange protocol.
More precisely, the attacker is
modeled as a probabilistic polynomial-time machine with full control
over the network. It is free to intercept, delay, drop, inject, or tamper
with messages sent between the parties. In addition, the attacker can
activate parties (as either an initiator or a responder)
to begin an execution of the
KE protocol. The attacker can also issue the following session exposure
queries to learn ephemeral and long-term secrets held by the parties:
\begin{citemize}
  \item $\statereveal$: The adversary can issue a
  $\statereveal$ query against any incomplete session. It is then
  given the local state associated with the targeted session.
  Importantly, the local state
  does {\em not} contain long-term secrets such as signature signing keys which
  are shared across multiple sessions involving the same party.

  \item $\keyreveal$: The adversary can issue a $\keyreveal$
  query against any completed session. It is then given the
  application-traffic key $\ksk$ associated with the targeted session.

  \item $\corrupt$: The adversary can issue a $\corrupt$ query
  on any party $P$, at which point it learns all information in the memory of
  $P$, including any long-term secrets. After corrupting a party $P$,
  the adversary dictates all subsequent actions by $P$.
\end{citemize}
Note that the adversary in this model is fully adaptive.

\para{Security definition.} To define the security of a KE protocol, Canetti and
Krawczyk give an indistinguishability-based definition that intuitively states
that an adversary cannot distinguish the real session key output by an
session from a uniformly random session key. However, since the adversary
has the ability to make session-exposure queries, this definition is only
meaningful if the power of the adversary is restricted. In particular, the
adversary can only try to distinguish the session key for sessions $(P, \sid)$
on which it did not issue one of the three session-exposure queries defined
above (otherwise, it can trivially distinguish the session key from a uniformly
random key). Formally, if the adversary makes a session-exposure query on
$(P, \sid)$, the session $(P, \sid)$ is considered exposed. Moreover, any sessions
that ``match'' $(P, \sid)$ are also considered to be exposed. This restriction
is necessary since the matching session also computes the same session
key (in an honest execution of the protocol), thus allowing again a trivial
distinguishing attack. In the post-specified model, a session $(Q, \sid)$
{\em matches} a completed session with public output $(P, \sid, Q)$ if
the session $(Q, \sid)$ is incomplete or if the session $(Q, \sid)$ completes with
public output $(Q, \sid, P)$.

To be more precise, at any point in the indistinguishability game,
the adversary can issue
a $\Test$ query on any completed session $(P, \sid)$.
Let $(P, \sid, Q)$ be the public output and $\ksk$ be the secret
output by this session. The challenger then sets
$k_0 = \ksk$ and $k_1 \getsr \calK$, where $\calK$ is the key-space for the
session key. The challenger then chooses a bit
$b \getsr \zo$ and gives the adversary
$k_b$. The adversary can continue making session exposure queries as well as
activate sessions and interact with the parties after making the $\Test$
query. At the end of the game, the adversary outputs a bit $b' \in \zo$.
The adversary wins the distinguishing game if $b' = b$.
An adversary is {\em admissible} if the session $(P, \sid)$ and all sessions
matching $(P, \sid, Q)$ are unexposed. We now state the formal security definition
from~\cite{CK02,Kra03}.
\begin{definition}[{Key-exchange security~\cite{CK02,Kra03}}]
  \normalfont
  \label{def:key-exchange-sec}
  A key-exchange protocol $\pi$ is secure if for all efficient and admissible
  adversaries $\calA$, the following properties hold:

  \begin{cenumerate}
    \item Suppose a session $(P, \sid)$ completes at an uncorrupted party $P$
    with public output $(P, \sid, Q)$ and secret output $\ksk$. Then, with
    overwhelming probability, if $Q$
    completes session $(Q, \sid)$ while $P$ and $Q$ are uncorrupted,
    the public output of $(Q, \sid)$ is $(Q, \sid, P)$, and the secret output is
    $\ksk$.

    \item The probability that $\calA$ wins the test-session distinguishing game
    is negligibly close to $1/2$.
  \end{cenumerate}
\end{definition}

\subsection{Service Discovery Model}
\label{app:service-disc-model}
In Section~\ref{sec:discovery-protocol}, we introduced a private service
discovery protocol that consisted of two sub-protocols: a broadcast protocol
where the server announced its identity and other relevant endpoint
information, and a 0-RTT mutual authentication protocol. In this section, we
describe how we can extend the Canetti-Krawczyk key-exchange framework to also
reason about the security of the service discovery protocol.

One major difference in the service discovery setting is that multiple clients
can legitimately respond to the server's broadcast message and initiate
handshakes with the server. In contrast, in the mutual authentication setting,
all messages sent between clients and servers are intended for a single
peer. Thus, the server's broadcast message might include ``semi-static''
secrets that persist for the lifetime of the broadcast message (which might
include more than one session). In our security model, the servers can
maintain three kinds of state: local session state that persists for a single
session, semi-static state that persists for the lifetime of a broadcast, and
long-term state that persists across multiple broadcasts.

In the service discovery setting, the adversary is allowed to activate servers
to produce a broadcast message. Specifically, a server can be activated to
construct a broadcast message by a tuple $(P, \bid)$ where $P$ is the party
being activated, and $\bid$ is a unique broadcast identifier. Next, the
adversary can activate parties to respond to a broadcast message by providing
a tuple $(P, \bid, \sid)$ along with a broadcast message. Here, $P$ identifies
the party that is being activated to initiate a key-exchange protocol, $\bid$
is the broadcast identifier, and $\sid$ is a session specific identifier. In
contrast to the basic key-exchange model, the session activation tuple does
not include a destination address. It is assumed that the server's discovery
broadcast already includes the relevant endpoint information. In the broadcast
setting, the pair $(\bid, \sid)$ is taken to be the unique session identifier.
The output of a successful key-exchange session at a party $P$ consists of a
public tuple $(P, \bid, \sid, Q)$, where $(\bid, \sid)$ identifies the session
and $Q$ is the peer of the session, as well as a secret application-traffic key $\atk$.
As in the key-exchange setting, sessions can also abort without producing
any output. When a local session completes (either successfully or due to
aborting), the local session state is erased. Moreover, whenever a server
is activated to initialize a new broadcast, its previous broadcast state
is also erased. In the subsequent protocol analysis, we will uniquely
identify broadcasts by a pair $(P, \bid)$ and sessions by a triple
$(P, \bid, \sid)$.

In the service discovery setting, the adversary's goal is still
to distinguish the session key output by a completed session from a session key
that is chosen uniformly at random. In addition, we allow the
adversary an additional query that allows it to learn any semi-static secrets
that persist for the lifetime of a particular broadcast:
\begin{citemize}
  \item \broadcastreveal: The adversary can issue a $\broadcastreveal$ query
  against any server that has initiated a broadcast. The adversary is given
  any semi-static state the server is maintaining for the lifetime of its
  current broadcast. However, the adversary is not given any long-term secrets
  such as signing keys which persist between broadcasts.
\end{citemize}

\para{Admissibility requirements.}
In the key-exchange security game, it was necessary to impose restrictions
on the adversary's power to prevent it from trivially winning the security game.
The same holds in the discovery setting. Specifically, we say a session
$(P, \bid, \sid)$ is ``exposed'' if the adversary makes the following queries:
\begin{citemize}
  \item The adversary makes a $\keyreveal$ query on $(P, \bid, \sid)$.

  \item The adversary makes a $\corrupt$ query on $P$ before
  the session $(P, \bid, \sid)$ has expired.

  \item $P$ is the initiator of the key-exchange protocol (the client),
  and the adversary had made a $\statereveal$ query on $(P, \bid, \sid)$.

  \item $P$ is the responder in the key-exchange protocol (the server),
  and the adversary has made both a $\statereveal$ query on $(P, \bid, \sid)$,
  and either a $\broadcastreveal$ query on $(P, \bid)$ or a
  $\corrupt$ query on $P$ before $(P, \bid)$ has expired.
\end{citemize}
Moreover, when the adversary makes one of the above queries and exposes a
session, all ``matching'' sessions are also marked as exposed. We adapt
the definition from the basic key-exchange setting: a session $(Q, \bid, \sid)$
matches a completed session with public output $(P, \bid, \sid, Q)$ if
$(Q, \bid, \sid)$ is incomplete or if $(Q, \bid, \sid)$ completes with public
output $(Q, \bid, \sid, P)$.

We give some intuition about the exposure conditions. The last requirement
effectively states that in the case of server compromise, as long as one of
the semi-static secret (from the broadcast) and the ephemeral session-specific
secret is not compromised, the adversary cannot learn anything about the
negotiated key. The above definition also captures perfect forward secrecy
since the adversary is allowed to corrupt parties after the expiration of the
test session.

We can now define security for the service discovery protocol in the same
manner as we defined key-exchange security for a mutual authentication
protocol (Definition~\ref{def:key-exchange-sec}):

\begin{definition}[Service discovery security]
  \normalfont
  \label{def:service-discovery-sec}
  A service discovery protocol $\pi$ is secure if for all efficient and admissible
  adversaries $\calA$, the following properties hold:

  \begin{cenumerate}
    \item Suppose a session $(P, \bid, \sid)$ completes at an uncorrupted
    party $P$ with public output $(P, \bid, \sid, Q)$ and secret output $\ksk$.
    Then, with overwhelming probability, if $Q$ completes session $(Q, \bid, \sid)$
    while $P$ and $Q$ are uncorrupted, the public output of $(Q, \bid, \sid)$
    is $(Q, \bid, \sid, P)$, and the secret output is $\ksk$.

    \item The probability that $\calA$ wins the test-session distinguishing game
    is negligibly close to $1/2$.
  \end{cenumerate}
\end{definition}

\section{Key-Exchange Privacy}
\label{app:key-exchange-privacy}
In this section, we formally define our notion of privacy for a key-exchange
protocol. Intuitively, we say a protocol is private if no efficient client or
server is able to learn anything about another party's identity unless they
satisfy that party's authorization policy. This property should hold even
if the adversary corrupts multiple parties; as long as none of the corrupted
parties satisfy the target's authorization policy, then they should not learn
anything about the target's identity. More formally, we operate in the
Canetti-Krawczyk key-exchange model, but make several modifications to the
adversary's capabilities and its objectives. We now describe our
model formally.

\para{Formal definition.}
Let $n$ be the number of parties participating in the protocol execution
experiment. We denote the parties $P_1, \ldots, P_n$. In addition to the $n$
parties, we also introduce a special test party $P_\ptest$ whose identity will
be hidden from the adversary at the beginning of the protocol execution
experiment. We now define two experiments $\expt_0$ and $\expt_1$. For $b \in
\zo$, the experiment $\expt_b$ proceeds as follows:
\begin{citemize}
  \item \textbf{Setup phase:} At the beginning of the experiment, the adversary
  submits a tuple of distinct identities $(\id_1, \ldots,\id_n)$ and two challenge
  identities $\sind{\id_\ptest}{0}$ and $\sind{\id_\ptest}{1}$
  (distinct from $\id_1, \ldots, \id_n$). For each $i
  \in [n]$, the challenger sets up the long-term secrets for party $P_i$. For
  instance, depending on the protocol specifications, the challenger might
  issue a certificate binding the credentials for $P_i$ to its identity
  $\id_i$. The challenger associates the identity $\sind{\id_\ptest}{b}$ with
  the special party $P_\ptest$, and performs the setup procedure for
  $P_\ptest$ according to the protocol specification.

  \item \textbf{Protocol execution:} In this phase, the adversary is free to
  activate any party ($P_1, \ldots, P_n$ and $P_\ptest$) to initiate sessions
  and respond to session messages. When the adversary activates a party $P$ to
  initiate a session $(P, \sid)$, it can also specify a policy $\pi_{(P, \sid)}$
  associated with the session $(P, \sid)$. Similarly, when the adversary
  activates a party $Q$ as a responder, it can specify a policy $\pi_{(Q, \sid)}$
  associated with the session $(Q, \sid)$. In addition, when the adversary
  activates a party $Q$ to respond to a session message for session $(Q, \sid)$,
  it also specifies the sender's policy. Party $Q$ answers the query if and
  only if it satisfies the sender's policy. If the adversary does not specify
  a policy for a session $(P, \sid)$, the party $P$ accepts all connections.
  The adversary can also issue $\statereveal$, $\keyreveal$ and $\corrupt$
  queries on any party participating in the protocol execution (with the same
  semantics as in the key-exchange security game described in
  Appendix~\ref{app:key-exchange-model}).

  \item \textbf{Output phase:} At the end of the protocol execution, the
  adversary outputs a bit $b' \in \zo$.
\end{citemize}
Intuitively, we will say that a key-exchange protocol is {\em private} if no
efficient adversary can distinguish $\expt_0$ from $\expt_1$ with probability
that is non-negligibly greater than $1/2$. This captures the property that no
active network adversary is able to distinguish interactions with a party
$\sind{\id_\ptest}{0}$ from those with a party $\sind{\id_\ptest}{1}$. Of
course, as defined currently, an adversary can trivially win the security
game. Since we have not placed any restrictions on the adversary's queries, it
can, for instance, corrupt the target party $P_\ptest$ and learn its identity.

\para{Admissibility requirements.} To preclude trivial ways of learning the identity
of the test party, we enumerate a series of constraints on the adversary's
power. We say an adversary is ``admissible'' for the privacy game if the following
conditions hold:
\begin{citemize}
  \item The adversary does not issue a $\corrupt$ query on $P_\ptest$.
  In addition, the adversary does not make a $\statereveal$ query
  on any session $(P_\ptest, \sid)$ that completes.

  \item Whenever a session $(P_\ptest, \sid)$ completes with public output
  $Q$, then the adversary has not made a $\statereveal$ query on the
  session $(Q, \sid)$, and moreover, $Q$ is not corrupt before the completion of
  $(P_\ptest, \sid)$.

  \item Let $\Pi_\ptest$ denote the set of policies the adversary
  has associated with the test party $P_\ptest$. Let $I \subseteq [n]$ be the
  indices of the parties the adversary has corrupted in the course of the
  protocol execution. Then, for all policies $\pi \in \Pi_\ptest$ and
  indices $i \in I$, it should be the case that $\id_i$ does not satisfy
  $\pi$.

  \item Whenever the adversary associated a policy $\pi$ with a particular
  session $(P, \sid)$, it must be the case that either $\sind{\id_\ptest}{0}$ and
  $\sind{\id_\ptest}{1}$ both satisfy $\pi$ or neither satisfy $\pi$.
\end{citemize}
These admissibility requirements are designed to rule out several ``trivial''
ways of breaking the privacy of the test party $P_\ptest$ in the protocol
execution experiment. Certainly, if the adversary corrupted the test party, or
exposed (via a $\statereveal$ or a $\corrupt$ query) a completed session
$(P_\ptest, \sid)$, it can easily learn the identity of $P_\ptest$. These
scenarios are captured by the first two admissibility requirements. The third
admissibility requirements states that the adversary never corrupts a party
who is authorized to talk with the test party. This is unavoidable because a
party is always willing to reveal its identity to any other party that
satisfies its policy. Thus, if the adversary obtained the credentials of a
party that satisfied the test party's policy, it should be able to learn the
test party's identity. The final admissibility requirement captures the fact that
we work in the model where a network adversary can always tell whether a
handshake completes or aborts. Thus, if the adversary is allowed to choose the
policies the protocol participants use, then it must be constrained so it
cannot choose a policy that allows it to trivially distinguish the two test
identities.

With these admissibility requirements, we now state our privacy definition
for key-exchange protocols.
\begin{definition}[Key-exchange privacy]
  \normalfont
  \label{def:key-exchange-privacy}
  A key-exchange protocol $\pi$ is private if for all efficient and admissible
  adversaries $\calA$, the probability that $\calA$ outputs $1$ in $\expt_0$
  is negligibly close to the probability that it outputs $1$ in $\expt_1$.
\end{definition}

\para{The SIGMA-I protocol.} To provide some additional intuition on the nature
of our privacy definition, we first show that the SIGMA-I
protocol~\cite{Kra03,CK02} does not satisfy our notion of privacy against {\em
active} network attackers. Consider the case where we have two parties: $P_1$
and the test party $P_\ptest$. The adversary chooses an arbitrary identity
$\id_1$ for $P_1$ and two distinct identities $\sind{\id_\ptest}{0}$ and
$\sind{\id_\ptest}{1}$ for $P_\ptest$. The adversary activates $P_1$ to
initiate a session $(P, \sid)$, and forwards the message (a DH share $g^x$) to
$P_\ptest$. The test party $P_\ptest$ responds with its DH share $g^y$ and an
encryption of its identity (along with other components) under an encryption
key derived from $g^x$ and $g^y$. At this point in the protocol execution,
none of the sessions have completed. The adversary now performs a
$\statereveal$ query on $P_1$ to learn $x$, from which it can derive the key
$P_\ptest$ used to encrypt its message. Now, the adversary simply decrypts the
response message from $P_\ptest$ and recovers its identity. Thus, we conclude
that the original SIGMA-I protocol does not provide privacy in the presence of
an active network adversary.

\section{Analysis of Private Mutual Authentication Protocol}
\label{app:mutual-auth-analysis}

In this section, we show that the protocol in Figure~\ref{fig:auth-proto-2}
from Section~\ref{sec:private-mutual-auth} is a private mutual authentication
protocol (satisfies both Definition~\ref{def:key-exchange-sec} and
Definition~\ref{def:key-exchange-privacy}). In the following description,
we will often refer to the first message from the client to the server as the
``initialization'' message, the server's response as the ``response'' message,
and the final message from the client as the ``finish'' message.

\para{Security.} Security of our protocol follows directly from the security
of the SIGMA-I protocol. Specifically, we have the following theorem.
\begin{theorem}
  The protocol in Figure~\ref{fig:auth-proto-2} is a secure key-exchange
  protocol (Definition~\ref{def:key-exchange-sec}) assuming the Hash-DH
  assumption holds in $\bbG$ and the security of the underlying
  cryptographic primitives (the signature scheme and the authenticated encryption
  scheme).
\end{theorem}
\begin{proof}
  As noted in Section~\ref{sec:private-mutual-auth}, the two differences
  between our protocol and the SIGMA-I protocol from~\cite{CK02} is that we
  use authenticated encryption to encrypt and authenticate the handshake
  messages rather than separate encryption and MAC schemes. The other
  difference is that in our protocol, the server substitutes a prefix-based
  encryption of its certificate chain for its actual certificate chain in
  the response message to the client.

  In the original SIGMA-I protocol, the requirement on the certificate chain is that it
  provides an authenticated binding between an identity and its public
  verification key. An encrypted certificate chain provides the exact same level of
  assurance. After all, an honest client only accepts the server's credential
  if the client successfully decrypts the encrypted certificate chain and verifies the
  underlying certificate chain. Therefore, we can substitute encrypted
  certificate chain into the SIGMA-I protocol with no loss in security. Security of our
  resulting protocol then follows directly from security of the original SIGMA-I
  protocol~\cite[\S 5.3]{CK02}.
\end{proof}

\para{Privacy.} We now show that our key-exchange protocol is private. In our analysis,
we construct the encryption scheme used to encrypt certificates
from an IBE scheme as described in Section~\ref{sec:crypto-primitives}. Privacy
for the client's identity then reduces to the security of the underlying key-exchange
(since the client encrypts its identity under a handshake secret key),
while privacy for the server's identity reduces to CCA-security of
the underlying IBE scheme (since the server encrypts its identity using a
prefix encryption scheme constructed from IBE).
\begin{theorem}
  \label{thm:mutual-auth-proto-private}
  The protocol in Figure~\ref{fig:auth-proto-2} is private
  (Definition~\ref{def:key-exchange-privacy}) assuming the IBE scheme used to
  construct the prefix encryption scheme is $\indidcca$-secure, the Hash-DH
  assumption holds in $\bbG$, and the underlying cryptographic primitives (the
  signature scheme and the authenticated encryption scheme) are
  secure.
\end{theorem}
\begin{proof}
  We proceed using a hybrid argument. First, we define a simulator that will
  simulate the role of the challenger for the adversary $\calA$ in the protocol
  execution environment. Then, we specify a sequence of hybrid experiments where
  we modify the behavior of the simulator.

  Specifically, the simulator $\calS$ takes as input the number of parties
  $n$, the security parameter $\lambda$, and the adversary $\calA$, and plays
  the role of the challenger in the protocol execution experiment with
  $\calA$. At the beginning of the simulation, $\calA$ submits a tuple of
  identities $(\id_1, \ldots, \id_n)$ and two test identities
  $\sind{\id_\ptest}{0}$, $\sind{\id_\ptest}{1}$. During the protocol
  execution, the simulator chooses the parameters for each of the parties,
  and handles the responses to the adversary's queries according to the
  specifications of the particular hybrid experiment. We now define our
  sequence of hybrid experiments:
  \begin{citemize}
    \item \textbf{Hybrid $\hyb_0$:} This is the real experiment $\expt_0$ (i.e.,
    the simulator responds to the adversary's queries as described in
    $\expt_0$).

    \item \textbf{Hybrid $\hyb_1$:} Same as $\hyb_0$, except whenever the test
    party $P_\ptest$ is activated to process a server's response message, if
    all of the validation checks pass, the simulator uses the identity
    $\sind{\id_\ptest}{1}$ in place of $\sind{\id_\ptest}{0}$ when constructing
    $P_\ptest$'s finish message.

    \item \textbf{Hybrid $\hyb_2$:} This is the real experiment $\expt_1$.
  \end{citemize}
  We now show that each consecutive pair of hybrid experiments is computationally
  indistinguishable. This suffices to show that $\expt_0$ and $\expt_1$ are
  computationally indistinguishable, and correspondingly, that the protocol in
  Figure~\ref{fig:auth-proto-2} provides privacy.

  \begin{claim}
    \label{claim:priv-mutual-auth-0-1}
    Hybrids $\hyb_0$ and $\hyb_1$ are computationally indistinguishable if
    the Hash-DH assumption holds in $\bbG$ and the underlying cryptographic primitives
    are secure.
  \end{claim}
  \begin{proof}
    Observe that $\hyb_0$ and $\hyb_1$ are identical experiments, except in
    $\hyb_1$, the finish messages sent by the test party $P_\ptest$ contain
    an encryption of the identity $\sind{\id_T}{1}$ under the handshake
    encryption key $\htk$ instead of $\sind{\id_T}{0}$. At a high level then,
    the claim follows from the fact that the SIGMA-I protocol ensures
    confidentiality of the client's finish message against active
    attackers~\cite[\S5.3]{CK02}.

    Specifically, Canetti and Krawczyk show that for any complete session $(P,
    s, Q)$ that is not exposed by the adversary (that is, neither this session
    nor its matching session $(Q, \sid)$ has been corrupted by a $\statereveal$
    or a $\corrupt$ query), breaking the semantic security of the information
    encrypted under $\htk$ in the finish message of session $(P, \sid)$ implies a
    distinguisher between $\htk$ and a random encryption key. This then
    implies an attack either on the handshake security of the protocol or one
    of its underlying cryptographic primitives.

    In our privacy model, the admissibility requirement stipulates that the
    test party $P_\ptest$ is never corrupted during the protocol execution.
    Suppose that in the protocol execution experiment, the adversary activates
    $P_\ptest$ to respond to a server's message in a session $(P_\ptest, \sid)$
    and $P_\ptest$ responds with a finish message (a ciphertext under the
    handshake key $\htk$). Since $P_\ptest$ is honest, if it sends the finish
    message, then it also completes the session setup by outputting the tuple
    $(P_\ptest, \sid, Q)$. Next, we use the fact that if $(P_\ptest, \sid)$
    completes with peer $Q$, by the admissibility requirement, the adversary
    must not have made a $\statereveal$ query on $(Q, \sid)$ nor was $Q$ corrupt
    before the completion of $(P_\ptest, \sid)$. Thus, we can directly invoke the
    security properties of the SIGMA-I protocol and argue that the handshake
    key $\htk$ used by $P_\ptest$ to encrypt the finish message in session
    $(P_\ptest, \sid)$ is computationally indistinguishable from a uniformly
    random key. The claim then follows by semantic security of the
    authenticated encryption scheme (strictly speaking, we require a hybrid
    argument over each session where $P_\ptest$ sends a finish message).
  \end{proof}

  \begin{claim}
    \label{claim:priv-mutual-auth-1-2}
    Hybrid $\hyb_1$ and $\hyb_2$ are computationally indistinguishable if
    the IBE scheme used to construct the prefix encryption scheme is
    $\indidcca$-secure.
  \end{claim}
  \begin{proof}
    Let $q$ be an upper bound on the number of sessions where $\calA$
    activates the test party $P_\ptest$ as the responder (the server) in the
    protocol execution experiment. We define a sequence of $q + 1$ intermediate
    hybrids $\hyb_{1,0}, \ldots, \hyb_{1, q}$ where hybrid experiment $\hyb_{1,i}$
    is defined as follows:
    \begin{citemize}
      \item \textbf{Hybrid $\hyb_{1,i}$:} Same as $\hyb_1$, except the first $i$
      times that $P_\ptest$ is activated as a responder, $P_\ptest$
      substitutes the identity $\sind{\id_\ptest}{1}$ for $\sind{\id_\ptest}{0}$
      in its response message. In all subsequent sessions where $P_\ptest$
      is activated as a responder, it uses the identity $\sind{\id_\ptest}{0}$.
    \end{citemize}
    By construction, hybrid experiment $\hyb_1$ is identical to $\hyb_{1,0}$
    and $\hyb_2$ is identical to $\hyb_{1, q}$. We now show that for all $i
    \in [q]$, hybrid $\hyb_{1, i-1}$ is computationally indistinguishable
    from hybrid $\hyb_{1, i}$, assuming that the IBE scheme is $\indidcca$-secure.

    \begin{claim}
      For all $i \in [q]$, hybrids $\hyb_{1, i - 1}$ and $\hyb_{1, i}$ are computationally
      indistinguishable assuming that the IBE scheme is $\indidcca$-secure.
    \end{claim}
    \begin{proof}
    Suppose $\calA$ is an adversary that is able to distinguish
    $\hyb_{1,i-1}$ from $\hyb_{1,i}$. We show how to use $\calA$ to build an
    adversary $\calB$ for the $\indidcca$-security. At the beginning of the
    $\indidcca$ game, algorithm $\calB$ is given the public parameters
    $\mpk$ for the IBE scheme. It starts running adversary $\calA$ and obtains
    a tuple of identities $(\id_1, \ldots, \id_n)$ and test identities
    $\sind{\id_\ptest}{0}$, $\sind{\id_\ptest}{1}$. Algorithm $\calB$
    then simulates the setup procedure in $\hyb_1$. In particular, for each
    $i \in [n]$, it chooses a signing and verification key for each party $P_i$.
    In addition, it also issues a certificate binding the identity $\id_i$
    to its associated signature verification key. Next, it prepares two certificates,
    one binding $\sind{\id_\ptest}{0}$ to the verification key
    for $\id_\ptest$, and another binding $\sind{\id_\ptest}{1}$ to the verification
    key for $\id_\ptest$. Note that $\calB$ does not make any queries to the IBE
    extraction oracle, and thus, does not
    associate any IBE identity keys with each party. We will describe how algorithm
    $\calB$ is able to answer the queries consistently in the simulation in
    spite of this fact. Finally, algorithm $\calB$ gives $\mpk$ to $\calA$
    and the protocol execution phase begins. Algorithm $\calB$
    simulates the response to each of adversary $\calA$'s queries as follows:
    \begin{citemize}
      \item \textbf{Client initialization queries.} These are handled
      exactly as in $\hyb_1$ and $\hyb_2$.

      \item \textbf{Server response queries.} When adversary $\calA$
      activates a party $P$ to respond to a client initialization query with some
      policy $\pi$, if $P \ne P_\ptest$, algorithm $\calB$ simulates the response
      message as in $\hyb_1$. This is possible
      because $\calB$ chooses all of the parameters that appears in the
      response message and also has the master public key $\mpk$ of the
      IBE scheme. If $P = P_\ptest$, then let $\ell$ be the number of times
      $\calA$ has already activated $P_\ptest$ to respond to a client initialization
      message in the protocol execution thus far. Then, algorithm $\calB$ proceeds
      as follows:
      \begin{citemize}
        \item If $\ell < i - 1$, then $\calB$ constructs the response message
        as described in $\hyb_2$, that is, using the identity
        $\sind{\id_\ptest}{1}$ in the response message.

        \item If $\ell \ge i$, then $\calB$ constructs the response message as
        described in $\hyb_1$, that is, using the identity $\sind{\id_\ptest}{0}$
        in the response message.

        \item If $\ell = i - 1$, then $\calB$ submits the tuple
        $(\sind{\id_\ptest}{0}, \sind{\id_\ptest}{1})$ with the policy $\pi$
        as its challenge identity in the IBE security game. It receives a ciphertext
        $\hct_S$ from the challenger. Algorithm $\calB$ constructs the rest of
        $P_\ptest$'s response message as in $\hyb_1$ and $\hyb_2$ using $\hct_S$ as
        its encrypted certificate chain.
      \end{citemize}

      \item \textbf{Client response queries.} When the adversary $\calA$ delivers
      a message to a client session $(P, \sid)$, $\calB$ responds as follows:
      \begin{cenumerate}
        \item Let $\pi_S$ be the sender's policy associated with the message.
        If $P \ne P_\ptest$ and $\id_P$ does not satisfy $\pi_S$, $\calB$
        aborts the session. If $P = P_\ptest$ and $\sind{\id_\ptest}{0}$ does
        not satisfy $\pi_S$, $\calB$ aborts the session. Recall that our
        admissibility requirement specifies that either $\sind{\id_\ptest}{0}$
        and $\sind{\id_\ptest}{1}$ both satisfy $\pi_S$, or neither satisfy
        $\pi_S$.

        \item $\calB$ parses the adversary's message as $(\sid, g^y, \ct)$.
        It checks that there is a local session $(P, \sid)$, and aborts the
        protocol if not. Next, it derives the
        keys $(\htk, \atk) = \kdf(g^x, g^y, g^{xy})$ where $x$ is the ephemeral DH exponent
        it chose for session $(P, \sid)$. It decrypts $\ct$ with $\htk$ and parses
        the result as a tuple ($\ct_S, \sigma_S)$, again aborting the protocol
        if decryption fails or the resulting message has the wrong form. Finally,
        $\calB$ checks that $\pi_S$ is the actual policy associated with $\ct_S$,
        aborting the protocol if not.

        \item If $\calB$ is still in the pre-challenge phase or if
        $\calB$ is in the post-challenge phase and either $\ct_S \ne \hct_S$
        or $\pi_S \ne \bar \pi$ (where $\bar \pi$ is the identity $\calB$ submitted
        to the IBE challenger and $\hct_S$ is the ciphertext $\calB$ received
        from the IBE challenger), then $\calB$
        queries the IBE decryption oracle on $\ct_S$ with identity
        $\pi_S$ to obtain a decrypted identity $\id_Q$. $\calB$ performs
        the remaining checks as would be done in hybrids $\hyb_1$ and $\hyb_2$,
        and constructs the client's finish message in the same manner.
        \\ \\
        If $\calB$ is in the post-challenge phase, $\ct_S = \hct_S$, and
        $\pi_S = \bar \pi$, then $\calB$ aborts the session if the identity
        $\sind{\id_\ptest}{0}$ does not satisfies the policy associated with
        session $(P, \sid)$. Recall again that under our admissibility
        requirement, either $\sind{\id_\ptest}{0}$ and $\sind{\id_\ptest}{1}$
        both satisfy the policy associated with session $(P, \sid)$, or
        neither satisfy the policy. If $\calB$ does not abort, then $\calB$
        verifies that the signature $\sigma_S$ is a valid signature on $(\sid,
        \ct_S, g^x, g^y)$ where $g^x, g^y$ are the session id and ephemeral DH
        shares associated with session $(P, \sid)$. If the signature verifies,
        then $\calB$ constructs $P$'s finish message as in $\hyb_1$ and
        $\hyb_2$.
      \end{cenumerate}

      \item \textbf{Server finish queries.} These are handled
      exactly as in $\hyb_1$ and $\hyb_2$.

      \item \textbf{$\statereveal$ and $\keyreveal$ queries.} These are handled
      exactly as in $\hyb_1$ and $\hyb_2$.

      \item \textbf{$\corrupt$ queries.} If $\calA$ asks to corrupt a party $P
      \ne P_\ptest$ (since $\calA$ is admissible), $\calB$ queries the IBE extraction
      oracle for the secret keys for $\id_P$ and each prefix of $\id_P$. It gives
      these secret keys to $\calA$, the long-term signing key associated with $\id_P$,
      and any ephemeral secrets for incomplete sessions currently in the local storage
      of $P$.
    \end{citemize}
    At the end of the game, $\calA$ outputs a guess $b'$. Algorithm $\calB$
    echoes this guess.

    To complete the proof, we first show that $\calB$ is an admissible IBE
    adversary in the $\indidcca$-security game. By construction, $\calB$ never
    requests the challenger to decrypt the challenge ciphertext under the
    challenge identity. It thus suffices to argue that $\calB$ never queries
    the extraction oracle on the challenge identity $\bar \pi$. This follows
    by admissibility of $\calA$. In the above specification, algorithm $\calB$
    only makes extraction queries when the adversary corrupts a party. By
    the admissibility requirement, none of the corrupted parties can satisfy
    any of the policies the adversary associated with $P_\ptest$, which
    in particular, includes $\bar \pi$. Thus, in the simulation, $\calB$
    never needs to issue an extraction query for the identity $\bar \pi$,
    and so, $\calB$ is admissible.

    Next, we show that if $\calB$ receives an encryption of
    $\sind{\id_\ptest}{0}$ from the IBE challenger in the challenge phase,
    then it has perfectly simulated hybrid $\hyb_{1,i-1}$ for $\calA$ and if
    it receives an encryption of $\sind{\id_\ptest}{1}$ from the IBE
    challenger, it has perfectly simulated $\hyb_{1, i}$ for $\calA$. It
    is easy to verify that $\calB$ correctly simulates the client initialization,
    server finish, and exposure queries.
    For the server response queries, $\calB$ constructs the encrypted identity
    as prescribed by $\hyb_{1,i-1}$ when the IBE challenger replies with an
    encryption of $\sind{\id_\ptest}{0}$. The encrypted identity is constructed
    as prescribed by $\hyb_{1,i}$ when the IBE challenger replies with an encryption
    of $\sind{\id_\ptest}{1}$.

    Finally, we argue that $\calB$ correctly simulates the client response queries.
    The only non-trivial case is when the adversary submits $\ct_S = \hct_S$ and
    $\pi_S = \bar \pi$. In all other cases, the behavior of $\calB$ is identical
    to that in the real scheme (unchanged between $\hyb_{1,i-1}$ and $\hyb_{1,i}$).
    In the case where $\ct_S = \hct_S$ and $\pi_S = \bar \pi$, then by construction,
    $\ct_S$ is either a valid encryption of the identity $\sind{\id_\ptest}{0}$
    or of $\sind{\id_\ptest}{1}$. By admissibility, the client's policy in session
    $(P, \sid)$ either accepts both $\sind{\id_\ptest}{0}$ and $\sind{\id_\ptest}{1}$,
    or rejects both. Thus, in the real experiment, the client's decision to
    abort the session or continue the handshake is {\em independent} of whether the
    server's identity was $\sind{\id_\ptest}{0}$ or $\sind{\id_\ptest}{1}$.
    The remainder of the response processing is identical to that in the real
    scheme, so we conclude that $\calB$ correctly simulates the client's behavior.
    The claim follows.
  \end{proof}
  \noindent Since each pair of hybrids $\hyb_{1,i-1}$ and $\hyb_{1,i}$ for all $i \in [q]$
  are computationally indistinguishable under the $\indidcca$-security of the
  IBE scheme, we conclude that $\hyb_1$ and $\hyb_2$ are computationally
  indistinguishable.
  \end{proof}

  \noindent Theorem~\ref{thm:mutual-auth-proto-private} now
  follows from Claims~\ref{claim:priv-mutual-auth-0-1}
  and~\ref{claim:priv-mutual-auth-1-2}.
\end{proof}

\section{Analysis of Private Service Discovery Protocol}
\label{app:discovery-protocol-security}

\begin{figure*}
\footnotesize
  \begin{framed}
    Let $\bbG$ be a group where the Hash-DH and Strong-DH assumptions hold and let
    $g$ be a generator for $\bbG$. Let $H_1, H_2$ be hash functions. Let
    $\prg$ be a secure pseudorandom generator (PRG) with seed-space $\calK$. Let $\Extract$ be a key-extraction
    algorithm with output space $\calK$. The key-extraction algorithm
    can be instantiated with the HMAC-based key derivation function~\cite{KE10} like that used
    in the OPTLS protocol~\cite{KW15}.
    \\ \\
    \textbf{Server's broadcast message:}
      \[ \bid, \id_S, g^s, \sig_S(\bid, \id_S, g^s) \]

    \textbf{Protocol messages:}
    \[ \text{Client} \hspace{32em} \text{Server} \]
    \[ \protorightlong{\bid, \sid, g^x, \set{\id_S, \id_C, \sig_C(\bid, \sid, \id_S, \id_C, g^s, g^x)}_\htk} \]
    \[ \protoleftlong{\bid, \sid, g^y, \set{\bid, \sid, \id_S, \id_C, g^s, g^x, g^y}_{\htk'}} \]

    \textbf{Broadcast description:} When a server $S$ is activated to broadcast a
    discovery message with broadcast id $\bid$, it erases any existing
    semi-static state (from a previous broadcast). It chooses $s \getsr \Z_p$
    and outputs the message $(\bid, \id_S, g^s, \sig_S(\bid, \id_S, g^s))$.
    The server stores the constant $s$ together with its current broadcast id
    $\bid$. Without loss of generality, we assume that $S$ never reuses a broadcast
    id (e.g., the broadcast ids must be in ascending order).
    \\ \\
    \textbf{Protocol actions:}
    \begin{cenumerate}
      \item When a party $C$ is activated to initialize a session $(\bid,
      \sid)$ with broadcast message $(\bid', \id_{S'}, g^{s'}, \sigma')$, it
      first checks that $\bid = \bid'$ and that no previous session $(C,
      \bid, \sid)$ has been initialized (aborting the session if the checks
      fail). If the checks succeed, it looks up the verification key for the party $S'$
      identified by $\id_{S'}$ and checks that $\sigma'$ is a valid signature on the tuple
      $(\bid, \id_{S'}, g^{s'})$ under the public key of $\id_{S'}$. If any of
      the checks fail, $C$ aborts the session. Otherwise, $C$ chooses
      a random exponent $x \getsr \Z_p$, and computes $k = H_1(g^{s'}, g^x,
      g^{s'x})$, $(\htk, \htk', \exk) = \prg(k)$. It replies to $S'$ with the
      message $(\bid, \sid, g^x, \{\id_{S'}, \id_C,
      \sig_C(\bid, \sid, \id_{S'}, \id_C, g^{s'}, g^x)\}_\htk)$. In addition, it
      stores its ephemeral exponent $x$, the server's identity $\id_{S'}$ and
      the server's DH share $g^{s'}$ in the local state of session $(C,
      \bid, \sid)$.

      \item When a party $S$ is activated as a responder with session id $\sid$
      and a message of the form $(\bid', \sid', g^{x'}, \ct')$, it looks up
      its current broadcast id $\bid$ and semi-static exponent $s$. Then, it
      checks that $\bid' = \bid$, $\sid' = \sid$, and no previous session
      $(S, \bid, \sid)$ has been initialized (aborting the session if the
      checks fail or there is no broadcast state). It then computes
      $k = H_1(g^{s}, g^{x'}, g^{sx'})$, $(\htk, \htk', \exk) = \prg(k)$,
      and tries to decrypt $\ct'$
      with $\htk$. If decryption succeeds (does not output $\perp$), $S$
      parses the output as $(\id_{S'}, \id_{C'}, \sigma')$. It checks that
      $\id_{S'} = \id_S$, and looks up the verification key for the party $C'$ identified
      by $\id_{C'}$. It
      verifies that $\sigma'$ is a valid signature on $(\bid, \sid,
      \id_{S}, \id_{C'}, g^s, g^{x'})$ under the public key of $C'$. If
      any of these checks fail, $S$ aborts the protocol. Otherwise, it
      chooses an ephemeral exponent $y \getsr \Z_p$ and computes
      $\atk = \Extract(\exk, H_2(g^{x'}, g^y, g^{x'y}))$.
      It replies to $C'$ with the message $(\bid, \sid,
      g^y, \{ \bid, \sid, \id_S, \id_{C'}, g^s, g^x, g^y \}_{\htk'})$. It also
      outputs the public tuple $(S, \bid, \sid, C')$ and the secret
      session key $\atk$.

      \item Upon receiving a response message of the form $(\bid', \sid',
      g^{y'}, \ct')$, the party $C$ checks whether there is a session
      $(C, \bid', \sid')$. If not, $C$ ignores the message. Otherwise,
      it loads the ephemeral constant $x$, the server's identity $\id_{S'}$,
      and the server's DH share $g^{s'}$ from its local storage, computes
      $k = H_1(g^{s'}, g^{x}, g^{s'x})$, $(\htk, \htk', \exk) = \prg(k)$.
      It tries to decrypt $\ct'$ with $\htk'$. If
      $\ct'$ successfully decrypts to the tuple
      $(\bid', \sid', \id_{S'}, \id_C, g^{s'}, g^x, g^{y'})$,
      where $x$ is the ephemeral exponent in the local state of session $(C, \bid, \sid)$,
      then $C$ outputs the public tuple $(C, \bid', \sid',
      S')$ and the secret session key $\atk = \Extract(\exk, H_2(g^{x},
      g^{y'}, g^{xy'}))$. Otherwise, $C$ aborts the protocol.
    \end{cenumerate}
  \end{framed}

  \caption{Formal specification of the private discovery protocol. We present the
  ``non-private'' version of the service discovery protocol without the
  prefix-based encryption. In the private version, the server's broadcast is
  encrypted under a prefix encryption scheme to its policy.}
  \label{fig:formal-discovery-protocol}
\end{figure*}

In this section, we formally argue the security of our private service
discovery protocol from Section~\ref{sec:discovery-protocol} under
Definition~\ref{def:service-discovery-sec}. First, we give a more formal
specification of the protocol in Figure~\ref{fig:formal-discovery-protocol}. In the
subsequent sections, we demonstrate our discovery protocol is secure in our
extended Canetti-Krawczyk key-exchange model
(Theorem~\ref{thm:disc-protocol-secure},
Appendix~\ref{app:disc-key-exchange-sec}), that it provides 0-RTT security
(Theorem~\ref{thm:one-pass-sec}, Appendix~\ref{app:0-rtt-security}), and that it provides privacy for both the
client and the server (Theorem~\ref{thm:disc-privacy}, Appendix~\ref{app:discovery-protocol-privacy}).

\subsection{Key-Exchange Security}
\label{app:disc-key-exchange-sec}
First, we show that the protocol in Figure~\ref{fig:formal-discovery-protocol}
is a secure key-exchange protocol in the service-discovery extension to the
Canetti-Krawczyk key-exchange model from Appendix~\ref{app:service-disc-model}.
Specifically, we prove the following theorem:

\begin{theorem}
  \label{thm:disc-protocol-secure}
  The protocol in Figure~\ref{fig:formal-discovery-protocol} is a secure
  service discovery protocol (Definition~\ref{def:service-discovery-sec}) in
  the random oracle model, assuming the Strong-DH assumption and the Hash-DH
  assumption hold in $\bbG$, as well as the
  security of the underlying cryptographic primitives (the signature scheme,
  the PRG, the authenticated encryption scheme, and the extraction algorithm).
\end{theorem}

\noindent In the following, we show that the two properties in
Definition~\ref{def:service-discovery-sec} hold for the key-exchange protocol
in Figure~\ref{fig:formal-discovery-protocol}.

\para{Proof of Property 1.} Our proof of Property~2
(Appendix~\ref{app:proof-property-2}) below will also show that our service
discovery protocol satisfies Property~1. In particular, we refer to hybrid
arguments $\hyb_0$ through $\hyb_3$ in the proof of Property~2.

\subsubsection{Proof of Property 2.}
\label{app:proof-property-2}
  Following similar analysis of key-exchange protocols
  in~\cite{KPW13,KW15}, we assume selective security in the adversary's choice
  of the test session. In our case, we require that at the beginning of the
  security game, the adversary commits to the following:
  \begin{citemize}
    \item The test session $(P, \bid, \sid)$.
    \item The peer's identity $Q$ in the test session.
    \item Whether $P$ is the initializer or the responder in the test
    session.
  \end{citemize}
  Note that selective security implies full adaptive security at
  a security loss that grows polynomially in the number of parties and the
  number of sessions the adversary initiates (where the simulator
  guesses the adversary's test session at the beginning of the protocol
  execution experiment).

  We begin by defining a simulator $\calS = \calS(\calA)$. On input $n$ (the
  number of parties), $\lambda$ (a security parameter) and an adversary
  $\calA$, the simulator $\calS$ simulates a run of the security game for
  the mutual authentication protocol. In the selective security setting, the adversary begins by
  committing to a test session $(\bar P, \hbid, \hsid)$, the peer $\bar Q$ in
  the test session, and whether $\bar P$ was the initiator or the responder in
  the test session. Then, the simulator $\calS$ initializes the $n$ parties by
  choosing a private signing key and a public verification key for each of the
  $n$ parties. When $\calA$ activates a party, the simulator $\calS$ performs
  the protocol actions (as described in
  Figure~\ref{fig:formal-discovery-protocol}) on behalf of the parties, and
  gives $\calA$ the outgoing messages as well as the public output of each
  session.

  \para{Description of simulator.} We now introduce several variants of the
  simulator $\calS$, which we generically denote $\bar \calS$. The simulator
  $\bar \calS$ behaves very similarly to $\calS$, except for the following
  differences.
  \begin{cenumerate}
    \item At the beginning of the simulation, the simulator chooses three
    exponents $\bar s, \bar x, \bar y \getsr \Z_p$ and four keys
    $\hhtk, \hhtk', \hexk, \hatk \in \calK$. The specification of
    $\hhtk, \hhtk', \hexk, \hatk$
    will determine the different variants of the simulator $\bar \calS$.

    \item In the selective security model, the adversary commits
    to a test session $(\bar P, \hbid, \hsid)$, the peer $\bar Q$, and the
    role of $\bar P$ in the test session at the beginning of the experiment.
    For notational convenience, let $\bar S \in \set{\bar P, \bar Q}$ denote
    the server the adversary commits to for the test session and similarly,
    let $\bar C \in \set{\bar P, \bar Q}$ denote the client to which it commits.

    The simulator $\bar \calS$ simulates the execution of the key-exchange
    security game exactly as $\calS$, except for the following differences.
    \begin{citemize}
      \item If the adversary activates $\bar S$ to initiate the broadcast
      $(\bar S, \hbid)$, the simulator uses $\bar s$ as the semi-static
      DH exponent in the broadcast.

      \item If the adversary activates $\bar C$ to initiate the session
      $(\bar C, \hbid, \hsid)$, the simulator uses $\bar x$ as the
      ephemeral DH exponent in the start message.

      \item If the adversary activates $\bar S$ as a responder to the session
      $(\bar S, \hbid, \hsid)$, the simulator uses $\bar y$ as the
      ephemeral DH exponent in the response message.

      \item It substitutes the keys $\hhtk, \hhtk', \hexk$
      for the keys $\htk, \htk', \exk$ whenever the DH shares
      $(g^{\bar s}, g^{\bar x}, g^{\bar s \bar x})$ are used to derive
      the keys during the simulation (that is, when the simulator
      needs to compute the quantity $\prg(H_1(g^{\bar s}, g^{\bar x}, g^{\bar s \bar x}))$).
      Similarly, it uses $\hatk$ in place of $\atk$ whenever the shares
      $(g^{\bar x}, g^{\bar s}, g^{\bar y}, g^{\bar s \bar x}, g^{\bar x \bar y})$
      are used to derive the session key (that is, when the simulator
      needs to compute the quantity
      $\extract(\hexk, H_2(g^{\bar x}, g^{\bar y}, g^{\bar x \bar y}))$).

    \end{citemize}

    \item At the end of the protocol, $\calA$ outputs a bit. The simulator
    $\bar \calS$ outputs the same bit.
  \end{cenumerate}
Our security analysis consists of two main cases, depending on whether the
adversary compromises the server's semi-static broadcast secret or not. Recall
from our admissibility requirement that as long as it is not the case that both
the server's ephemeral DH secret and the server's semi-static broadcast secret
are compromised, we say the adversary is admissible. This is very similar to
the case analysis in~\cite{KW15}. More concretely, our two cases are given
as follows:
\begin{citemize}
  \item \textbf{Case 1:} The adversary neither issues a $\broadcastreveal$
  query on $(\bar S, \hbid)$ nor corrupts $\bar S$ before $(\bar S, \hbid)$
  expires.

  \item \textbf{Case 2:} The adversary does not issue a $\statereveal$
  query on $(\bar S, \hbid)$.
\end{citemize}
For each case, we define a sequence of hybrid experiments, and then show
that each consecutive pair of hybrid experiments are computationally
indistinguishable. Before we begin the case analysis, we prove the following
lemma.
\begin{lemma}
  \label{lem:proper-client-init}
  Suppose a session $(P, \bid, \sid)$ completes with a peer $Q$.
  Assume moreover that neither $P$ nor $Q$
  has been corrupted before the completion of $(P, \bid, \sid)$. Then,
  assuming that $\sig$ is a secure signature scheme (and that the certificates
  are authenticated), the following hold:
  \begin{citemize}
    \item If $P$ is the client and $Q$ is the server, then
    the adversary initiated a broadcast $(Q, \bid)$ and
    $P$ must have been activated to initiate a session $(P, \bid, \sid)$
    with the broadcast message output by $(Q, \bid)$.

    \item If $P$ is the server and $Q$ is the client,
    the session $(Q, \bid, \sid)$ cannot
    complete with a peer $P'$ where $P' \ne P$.
  \end{citemize}
\end{lemma}
\begin{proof}
  We prove each claim separately:
  \begin{citemize}
    \item When $P$ is activated to initialize a session $(P, \bid,
    \sid)$ with a broadcast message $(\bid', \id_{S'}, g^{s'},
    \sigma')$, it first checks that $\bid' = \bid$ and that $\sigma'$ is a
    valid signature on $(\bid, \id_{S'}, g^{s'})$ under the public key
    identified by $\id_{S'}$. Since $(P, \bid, \sid)$ completes with peer $Q$,
    it must be the case that $S' = Q$. Since $Q$ has not been corrupted before
    the completion of $(P, \bid, \sid)$, it signs at most one broadcast
    message containing $\bid$. Thus, if $(P, \bid, \sid)$ completes with $Q$,
    it must have been initialized with the broadcast message output by $(Q,
    \bid)$ since this is the only message that contains a valid signature from
    $Q$ with the broadcast id $\bid$. Otherwise, $\calA$ can be used to break
    the security of the signature scheme $\sig$.

    \item Since $(P, \bid, \sid)$ completes with peer $Q$ and $P$ is the server,
    $P$ must have received a message containing a signature from $Q$ on a
    tuple containing $(\bid, \sid, \id_P, \id_Q)$. By assumption, $Q$ has not
    been corrupted before the completion of $(P, \bid, \sid)$, so it would
    only sign a message of this form if it was activated to initialize
    a session $(Q, \bid, \sid)$. Otherwise, the adversary must have forged
    a signature under $Q$'s signing key.
    \\ \\
    But if $Q$ was activated to initiate the session $(Q, \bid, \sid)$ and the session
    completes with a peer $P' \ne P$, then the only signature that an honest
    $Q$ would have produced that contains the session identifier $(\bid,
    \sid)$ is a signature on a tuple containing $(\bid, \sid, \id_{P'}, \id_Q) \ne
    (\bid, \sid, \id_P, \id_Q)$. In this case, an honest $Q$ would never have
    signed any tuple containing $(\bid, \sid, \id_P, \id_Q)$. Thus, any adversary that
    can cause $(P, \bid, \sid)$ to complete with peer $Q$ and have $(Q, \bid,
    \sid)$ complete with peer $P'$ can be used to forge signatures for $\sig$
    (in particular, under $Q$'s signing key).
    \qedhere
  \end{citemize}
\end{proof}

\noindent We now consider the two possible cases, and argue that in each case,
the adversary's advantage in distinguishing the session key of an unexposed
session from random is negligible.
\para{Case 1: $s$ is not compromised.} In this case, we rely on the security
of the server's broadcast secret for the privacy of the session key. We now
define our sequence of hybrid experiments:

\begin{citemize}
  \item \textbf{Hybrid $\hyb_0$:} This is the real protocol execution experiment.
  Specifically, the simulator $\bar \calS$ where
  $(\hhtk, \hhtk', \hexk) = \prg(H_1(g^{\bar s}, g^{\bar x}, g^{\bar s \bar x}))$,
  and $\hatk = \extract(\hexk, H_2(g^{\bar x}, g^{\bar y}, g^{\bar x \bar y}))$.

  \item \textbf{Hybrid $\hyb_1$:} Same as hybrid $\hyb_0$, except $\bar \calS$
  also aborts if $\calA$ queries $H_1$ on the input
  $(g^{\bar s}, g^{\bar x}, g^{\bar s \bar x})$.

  \item \textbf{Hybrid $\hyb_2$:} Same as hybrid $\hyb_1$, except
  $\hhtk, \hhtk', \hexk$ are uniformly random over $\calK$.

  \item \textbf{Hybrid $\hyb_3$:} Same as hybrid $\hyb_2$, except
  $\bar S$ also aborts if the session $(\bar Q, \hbid, \hsid)$
  does not match $(\bar P, \hbid, \hsid)$. In addition, $\bar S$ also aborts
  if the session key output by the test session $(\bar P, \hbid, \hsid)$ is
  not $\hatk$.

  \item \textbf{Hybrid $\hyb_4$:} Same as hybrid $\hyb_3$, except
  $\hatk$ is uniformly random over $\calK$.
\end{citemize}
We now show that each consecutive pair of hybrid experiments
described above is computationally indistinguishable.

\begin{claim}
  \label{claim:priv-disc-hyb-0-1}
  Hybrids $\hyb_0$ and $\hyb_1$ are computationally indistinguishable if the
  Strong-DH assumption holds in $\bbG$ and $H_1$ is modeled as a random oracle.
\end{claim}

\begin{proof}
  Let $(\bar P, \hbid, \hsid)$ be the session, and $\bar Q$ be the peer that the
  adversary commits to at the beginning of the experiment. By definition, this
  means that $(\bar P, \hbid, \hsid, \bar Q)$ is the public output of the test
  session.

  Let $\calA$ be a distinguisher between
  $\hyb_{0}$ and $\hyb_{1}$. We use $\calA$ to build a Strong-DH adversary
  $\calB$ as follows. Algorithm $\calB$ is given a Strong-DH challenge $(g,
  g^{\bar s}, g^{\bar x})$, as well as a DDH oracle $\calO(g^{\bar s}, \cdot,
  \cdot)$ where $\calO(g^\alpha, g^\beta, g^\gamma) = 1$ if $\gamma = \alpha
  \beta \pmod p$, and 0 otherwise. We now describe the operation of $\calB$.
  In the security reduction, $\calB$ will need to program the outputs of
  $H_1$.

  At the beginning of the simulation, algorithm $\calB$ generates a private
  signing key along with a public verification key (in the same manner as
  $\bar \calS$) for each of the $n$ parties. Algorithm $\calB$ also
  initializes two empty tables $T_1$ and $T_2$. The first table $T_1$ is used
  to maintain the mapping from tuples $(g^\alpha, g^\beta, g^\gamma)$ to
  random oracle outputs $H_1(g^\alpha, g^\beta, g^\gamma)$, and will be used
  for answering oracle queries to $H_1$. The second table $T_2$ maps between
  group elements $g^\beta$ and random oracle outputs $H_1(g^{\bar s}, g^\beta,
  g^\gamma)$ where $\gamma = \bar s \beta \pmod p$. In other words, $T_2$ maps
  between DDH tuples with first component $g^{\bar s}$ and is used to answer
  $\keyreveal$ queries consistently in the simulation.

  Next, $\calB$ chooses a random exponent $\bar y \getsr \Z_p$ and a random
  value $k' \getsr \calK$. It adds the mapping $[g^{\bar x} \mapsto k']$
  to $T_2$, computes $(\hhtk, \hhtk', \hexk) = \prg(k')$ and
  $\hatk = \extract(\hexk, H_2(g^{\bar x}, g^{\bar y}, g^{\bar x \bar y}))$.

  Algorithm $\calB$ then begins the simulation of the key-exchange security
  game for $\calA$. Algorithm $\calB$ responds to $\calA$'s actions as
  follows. Recall that in the selective security model, the adversary commits
  to the initiator of the protocol at the beginning of the security game. In
  the following description, we again write $\bar C \in \set{\bar P, \bar Q}$
  to denote the client in the test session, and $\bar S \in \set{\bar
  P, \bar Q}$ to denote the server in the test session.
  \begin{citemize}
    \item \textbf{Broadcast queries.} If the adversary activates $\bar S$ to
    initiate the broadcast $(\bar S, \hbid)$, the simulator uses $g^{\bar s}$
    from the Strong-DH challenge as the semi-static DH share in the broadcast
    message. For other broadcast queries, algorithm $\calB$ chooses a fresh DH
    exponent and constructs the broadcast message exactly as in the real protocol.

    \item \textbf{Client initialization queries.} When $\calA$ activates
    a party $P$ to initiate a session $(P, \bid, \sid)$, if
    $(P, \bid, \sid) \ne (\bar C, \hbid, \hsid)$, $\calB$ chooses a fresh
    DH exponent and prepares the message exactly as in the real scheme.
    \\ \\
    Otherwise, if $(P, \bid, \sid) = (\bar C, \hbid, \hsid)$,
    $\calB$ sets $g^{\bar x}$ from the Strong-DH challenge to be the DH
    share in its message. Next, $\calB$ sets $k = k'$ and
    $(\htk, \htk', \exk) = \prg(k)$, where $k'$ is the key $\calB$ chose at the beginning
    of the simulation. The remaining steps of the query processing
    are handled exactly as in the real experiment.

    \item \textbf{Server response queries.} When $\calA$ activates
    a server $P$ to respond to a session $(P, \bid, \sid)$, if
    $(P, \bid) \ne (\bar S, \hbid)$, then $\calB$ chooses a fresh ephemeral
    DH exponent and constructs the response message exactly as in the real
    scheme. Note that since it chose the semi-static DH share for $(P, \bid)$,
    it can construct the broadcast message exactly as in the real protocol.
    \\ \\
    If $(P, \bid) = (\bar S, \hbid)$, then let $g^x$ be the ephemeral DH share
    in the activation message. Algorithm $\calB$ checks if there is already a
    mapping $g^x \mapsto k$ in table $T_2$. If the mapping does not exist,
    then $\calB$ chooses a new key $k \getsr \calK$ and adds the mapping $g^x
    \mapsto k$ to $T_2$. Then, it derives $(\htk, \htk', \exk) = \prg(k)$
    and verifies the activation message (exactly as prescribed in the
    real protocol). If everything succeeds and $\sid = \hsid$, algorithm
    $\calB$ uses $g^{\bar y}$ (chosen at the beginning of the experiment) as
    the ephemeral DH share when constructing its response. Otherwise, it
    chooses a fresh ephemeral share and constructs the response as in the real
    scheme.

    \item \textbf{Client finish queries.} When a client receives a response
    message for session $(P, \bid, \sid)$, if $(P, \bid, \sid) \ne (\bar C,
    \hbid, \hsid)$, $\calB$ constructs the outputs as prescribed by the real
    scheme (this is possible since $\calB$ chose the client's ephemeral DH
    share in this case). If $(P, \bid, \sid) = (\bar C, \hbid, \hsid)$,
    $\calB$ proceeds as described in the real scheme, except it sets
    $(\htk, \htk', \exk) = \prg(k')$, where $k'$ is the key $\calB$ chose at
    the beginning of the simulation.

    \item \textbf{Exposure queries.} Algorithm $\calB$ answers all admissible
    $\statereveal$, $\broadcastreveal$, $\keyreveal$ and $\corrupt$ queries as
    $\bar \calS$. In the case of an admissible $\statereveal$,
    $\broadcastreveal$, or $\corrupt$ query, the simulator can answer since it
    chose all of the associated secrets. For
    $\keyreveal$ queries, we use the fact that $\calB$ is already able to
    process the server response and client finish queries, which includes the
    computation of the secret session key.

    \item \textbf{Oracle queries to $H_1$.} Whenever $\calA$ queries the
    random oracle $H_1$ at $(g^\alpha, g^\beta, g^\gamma)$, $\calB$ responds
    as follows:
    \begin{citemize}
      \item If there is already a mapping of the form
      $[(g^\alpha, g^\beta, g^\gamma) \mapsto k]$ in $T_1$, then $\calB$
      replies with $k$.

      \item If either $g^\alpha \ne g^{\bar s}$, or $g^\alpha = g^{\bar s}$ and
      $\calO(g^\beta, g^\gamma) = 0$, then $\calB$ chooses
      $k \getsr \calK$, adds the mapping
      $[(g^\alpha, g^\beta, g^\gamma) \mapsto k]$ to $T_1$
      and replies with $k$.

      \item If $g^\alpha = g^{\bar s}$ and $\calO(g^\beta, g^\gamma) = 1$,
      then $\calB$ checks whether there is a mapping of the form
      $[g^\beta \mapsto k]$ in $T_2$. If so, $\calB$ adds the mapping
      $[(g^\alpha, g^\beta, g^\gamma) \mapsto k]$ to $T_1$ and replies with $k$.
      If no such mapping exists in $T_2$, then $\calB$ samples $k \getsr \calK$,
      and adds the mapping $[(g^\alpha, g^\beta, g^\gamma) \mapsto k]$ to $T_1$,
      and the mapping $[g^\beta \mapsto k]$ to $T_2$.
    \end{citemize}
    Moreover, if $g^\alpha = g^{\bar s}$ and $g^\beta = g^{\bar x}$ and
    $\calO(g^\beta, g^\gamma) = 1$, $\calB$ aborts the simulation and outputs
    $g^{\gamma}$.
  \end{citemize}
  To complete the proof, we show that $\calB$ perfectly simulates $\hyb_0$ for
  $\calA$. First, the exponents $\bar s$ and $\bar x$ from the Strong-DH
  challenge are distributed uniformly over $\Z_p$, so they are properly
  distributed. By construction of the tables $T_1$ and $T_2$ and programming
  the outputs of the random oracle $H_1$ accordingly, $\calB$ ensures
  that all sessions involving the broadcast $(\bar S, \hbid)$ are properly simulated.
  Note that in the simulation, $k'$ plays the role of
  $H_1(g^{\bar s}, g^{\bar x}, g^{\bar s \bar x})$, so both the sessions
  $(\bar C, \hbid, \hsid)$ and $(\bar S, \hbid, \hsid)$ are properly simulated.

  Since $\calB$ perfectly simulates $\hyb_0$ for $\calA$, with non-negligible
  probability, $\calA$ will query $H_1$ on $(g^{\bar s}, g^{\bar x}, g^{\bar s \bar x})$.
  Then $\calB$ succeeds in the Strong-DH game with the same advantage.
\end{proof}

\begin{claim}
\label{claim:priv-disc-hyb-1-2}
Hybrids $\hyb_1$ and $\hyb_2$ are computationally indistinguishable if
$\prg$ is a secure pseudorandom generator.
\end{claim}

\begin{proof}
  Suppose there exists an adversary $\calA$ that can distinguish between
  $\hyb_1$ and $\hyb_2$. We use $\calA$ to build an algorithm $\calB$ that can
  break PRG security. Algorithm $\calB$ operates as follows. At the beginning of
  the game, the PRG challenger samples a seed $s \getsr \calK$, and gives $\calB$
  a string $(\hhtk, \hhtk', \hexk)$ where either $(\hhtk, \hhtk', \hexk) = \prg(s)$
  or $(\hhtk, \hhtk', \hexk) \getsr \calK^3$.

  Algorithm $\calB$ starts simulating the protocol execution experiment for
  $\calA$. It chooses all of the parameters as in the real scheme, and aborts
  if any of the abort conditions in $\hyb_1$ and $\hyb_2$ are
  triggered. During the simulation, if $\calB$ ever needs to derive
  keys $(\htk, \htk', \exk) = \prg(k)$ using
  $k = H_1(g^{\bar s}, g^{\bar x}, g^{\bar s \bar x})$, it instead uses the
  keys $(\hhtk, \hhtk', \hexk)$ from the PRG challenger.

  If $(\hhtk, \hhtk', \hexk)$
  was the output of the PRG, then $\calB$ correctly simulates $\hyb_1$. In
  particular, algorithm $\calB$ has exactly simulated an execution of the
  real scheme where the value of the random oracle $H_1$ at
  $(g^{\bar s}, g^{\bar x}, g^{\bar s \bar x})$ is the PRG seed $s$. Since
  $s$ is sampled uniformly at random by the PRG challenger, this value
  is properly distributed. Moreover, even though the simulator does not know
  the seed $s$, it can still correctly answer all of the adversary's
  queries because it never needs to
  query $H_1$ at $(g^{\bar s}, g^{\bar x}, g^{\bar s \bar x})$. If it did,
  then the experiment aborts. Thus, the adversary's view of the protocol
  execution is simulated perfectly according to the specification
  of hybrid $\hyb_1$.

  Conversely, if the keys $(\hhtk, \hhtk', \hexk)$ were chosen uniformly at
  random, then $\calB$ has correctly simulated $\hyb_2$. Thus, by security
  of the PRG, hybrids $\hyb_1$ and $\hyb_2$ are computationally indistinguishable.
\end{proof}

\begin{claim}
\label{claim:priv-disc-hyb-2-3}
Hybrids $\hyb_2$ and $\hyb_3$ are computationally indistinguishable if
the underlying encryption scheme is an authenticated encryption scheme
and that $\sig$ is a secure signature scheme (and that the certificates are
authenticated).
\end{claim}

\begin{proof}
  Suppose an adversary $\calA$ is able to distinguish between hybrids $\hyb_2$ and
  $\hyb_3$ with non-negligible probability. First, suppose the adversary is able
  to produce a protocol execution where $(\bar Q, \hbid, \hsid)$ does not match
  $(\bar P, \hbid, \hsid)$. We consider two possibilities:

  \begin{citemize}
    \item Suppose that $\bar P$ is the server in
    the session $(\bar P, \hbid, \hsid)$. By definition, the session $(\bar Q,
    \hbid, \hsid)$ matches $(\bar P, \hbid, \hsid)$ if it is incomplete. Thus,
    it suffices to argue that if $(\bar Q, \hbid, \hsid)$ completes, it
    completes with peer $\bar P$. By admissibility, $\bar Q$ must not have been
    corrupted before the completion of $(\bar P, \hbid, \hsid)$. The claim
    then follows by Lemma~\ref{lem:proper-client-init}.

    \item Suppose that $\bar P$ is the client in the
    session $(\bar P, \hbid, \hsid)$. As in the previous case, the session
    $(\bar Q, \hbid, \hsid)$ is matching if it is incomplete.
    \\ \\
    Instead, suppose that $(\bar Q, \hbid, \hsid)$ completes with a peer $P
    \ne \bar P$. Since neither $\bar P$ nor $\bar Q$ are corrupt before
    the completion of their respective session $(\bar P, \hbid, \hsid)$
    and $(\bar Q, \hbid, \hsid)$, the following must have occurred:
    \begin{citemize}
      \item If $(\bar P, \hbid, \hsid)$ completes with peer $\bar Q$,
      $\bar P$ must have received an authenticated encryption of a message
      containing $(\hbid, \hsid, \id_{\bar Q}, \id_{\bar P})$ under $\hhtk'$.

      \item If $(\bar Q, \hbid, \hsid)$ completes with peer $P$,
      then $\bar Q$ must have encrypted a message containing
      $(\hbid, \hsid, \id_{\bar Q}, \id_P)$ under some key (possibly $\hhtk'$). Moreover,
      this is the only message an honest $\bar Q$ ever encrypts under
      {\em any} key that contains $(\hbid, \hsid)$.
    \end{citemize}

    We can use $\calA$ to construct an algorithm $\calB$ that
    breaks the ciphertext integrity of the
    underlying authenticated encryption scheme. At the beginning of the
    ciphertext integrity game, the challenger samples a random encryption key.
    This key will play the role of $\hhtk'$ in the reduction. During the
    simulation, algorithm $\calB$ is given access to an encryption oracle.
    Algorithm $\calB$ responds to all queries as described in $\hyb_2$,
    choosing all parameters other than $\hhtk'$ for itself. During the
    simulation, whenever $\calB$ needs to encrypt a message under $\hhtk'$,
    it instead forwards it to the encryption oracle to obtain the ciphertext.
    When $\calA$ delivers a response message containing a ciphertext $\hct$
    to the session $(\bar P, \hbid, \hsid)$, $\calB$ submits $\hct$ to the
    ciphertext integrity challenger as its forgery. Otherwise, during the
    simulation, whenever $\calB$ needs to decrypt a ciphertext
    using the key $\hhtk'$, $\calB$
    checks whether this was one of the ciphertexts it previously
    submitted to the encryption oracle. If so, it looks up the corresponding
    plaintext and uses that value to simulate the response. If not, $\calB$
    substitutes the value $\bot$ for the decryption when simulating the
    response.
    \\ \\
    We argue first that $\calB$ correctly simulates $\hyb_2$ for $\calA$.
    Since, $\calB$ chooses all quantities other than $\hhtk'$ as described in
    $\hyb_2$ and $\hhtk'$ is properly sampled by the ciphertext integrity
    challenger, it suffices to argue that $\calB$ correctly simulates the
    client computation of the finish message where it needs to decrypt a
    ciphertext under $\hhtk'$ and verify the integrity of the underlying
    message. Certainly, if the ciphertext that needs to be verified was supplied
    to $\calB$ by the encryption oracle, $\calB$ correctly simulates the
    decryption operation. If the ciphertext was not one that it received
    from the encryption oracle, then either it is invalid, in which case
    decryption in $\hyb_2$ would also have produced $\bot$, or it is valid,
    in which case it is a forgery. Thus, assuming the underlying encryption
    scheme provides ciphertext integrity, $\calB$ correctly simulates $\hyb_2$ (this
    statement can be formalized by introducing another hybrid experiment).
    \\ \\
    Since $\calB$ correctly simulates the view of $\hyb_2$ for $\calA$ with
    non-negligible probability, $\calA$ is able to activate $\bar P$
    with a server response message containing a valid encryption $\hct$ of
    a message containing
    $(\hbid, \hsid, \id_{\bar Q}, \id_{\bar P})$ under $\hhtk'$. As argued above,
    the only ciphertext an uncorrupted $\bar Q$ would construct
    containing $(\hbid, \hsid)$ is for a tuple with the prefix
    $(\id_{\bar Q}, \id_P, \hbid, \hsid) \ne (\id_{\bar Q}, \id_{\bar P}, \hbid, \hsid)$.
    Thus, during the simulation $\calB$ never needs to query the encryption oracle
    on the tuple $(\hbid, \hsid, \id_{\bar Q}, \id_{\bar P})$, and so
    it can submit $\hct$ as its forgery. We conclude that as long as
    the underlying encryption scheme is an authenticated encryption scheme,
    then $(\bar Q, \hbid, \hsid)$ cannot complete with a peer $P \ne \bar P$.
  \end{citemize}
  Next, suppose that the adversary is able to produce a protocol execution where
  the session key output by the test session $(\bar P, \hbid, \hsid)$
  is not $\hatk$. We again consider two possibilities:
  \begin{citemize}
    \item Suppose that $\bar P$ is the server in the session $(\bar P, \hbid, \hsid)$.
    Since $(\bar P, \hbid, \hsid)$ completes with peer $\bar Q$ and $\bar P$ is the server,
    it must be the case that $\bar P$ received a message containing a signature on a tuple
    $(\hbid, \hsid, \id_{\bar P}, \id_{\bar Q}, g^{\bar s}, g^{x'})$ under the key identified by
    $\id_{\bar Q}$ for some $x' \in \Z_p$. From the above analysis,
    with overwhelming probability, the session $(\bar Q, \hbid, \hsid)$ matches
    $(\bar P, \hbid, \hsid)$. By the admissibility requirement, this means that
    $\bar Q$ has not been corrupted before the completion of
    $(\bar P, \hbid, \hsid)$. Thus, $\bar Q$ would only sign
    a message containing $(\hbid, \hsid)$ if it was activated to initiate
    a session $(\bar Q, \hbid, \hsid)$. By construction, the simulator uses
    $\bar x$ as the ephemeral exponent when initiating this session. Moreover, the only
    message an honest $\bar Q$ signs that contains $(\hbid, \hsid)$ also contains
    $g^{\bar x}$ as the client's DH share. Thus, it must be the case that $x' = \bar x$,
    since otherwise, the adversary can be used to forge
    signatures under the public key bound to $\id_{\bar Q}$.
    Finally, if the remaining validation checks pass
    and $(\bar P, \hbid, \hsid)$ completes with $\bar Q$, then $\bar P$ uses $\bar y$
    as its ephemeral exponent and derives the session key $\atk$ using
    the DH shares $g^{\bar s}, g^{\bar x}, g^{\bar y}$. In this case, $\atk = \hatk$,
    as desired.

    \item Suppose that $\bar P$ is the client in the session $(\bar P, \hbid, \hsid)$.
    This means that $\bar P$ received an ephemeral share $g^{y'}$
    and an authenticated encryption
    of the tuple
    $(\hbid, \hsid, \id_{\bar Q}, \id_{\bar P}, g^{\bar s}, g^{\bar x}, g^{y'})$ for
    some $y' \in \Z_p$ under the key $\hhtk'$. If so, then $\bar P$ derives the session
    key $\atk$ from $g^{\bar s}$, $g^{\bar x}$, and $g^{y'}$. Since $\bar Q$
    has not been corrupted before the completion of $(\bar P, \hbid, \hsid)$, it
    sends at most one encryption of a message containing $(\hbid, \hsid)$.
    From the specification of the simulator, we also have that when $\bar Q$ is activated
    as a responder to the session $(\hbid, \hsid)$, it uses $\bar y$ as its ephemeral exponent.
    Thus, the only ciphertext an uncompromised $\bar Q$ constructs that
    encrypts $(\hbid, \hsid)$ also contains $g^{\bar y}$.
    By the same argument as above, we appeal to the security of the authenticated encryption
    scheme (in particular, ciphertext integrity) to argue that $y' = \bar y$ with
    overwhelming probability. Otherwise, the adversary must have been able to create
    a new ciphertext (under $\hhtk'$)
    encrypting $(\hbid, \hsid)$ and $g^{y'}$ for $y' \ne \bar y$, thereby breaking
    ciphertext integrity. We conclude that if $(\bar P, \hbid, \hsid)$ completes, $\bar P$
    derives the shared key from $g^{\bar s}$, $g^{\bar x}$, and $g^{\bar y}$, in which case
    $\atk = \hatk$. Note that this also shows that if two matching session complete, then
    they both derive the same session key. \qedhere
  \end{citemize}
\end{proof}

\begin{claim}
\label{claim:priv-disc-hyb-3-4}
Hybrids $\hyb_3$ and $\hyb_4$ are computationally indistinguishable if
the $\extract$ function is a secure pseudorandom function.
\end{claim}

\begin{proof}
  Suppose there exists an adversary $\calA$ that can distinguish between
  $\hyb_3$ and $\hyb_4$. It is straightforward to use $\calA$ to build a PRF
  adversary for $\extract$. In hybrids $\hyb_3$ and $\hyb_4$, the extraction
  key $\hexk$ used as the key to $\extract$ is sampled
  uniformly at random (and independently of other scheme parameters).
  Moreover, the responses to the adversary's queries (other than $\corrupt$
  queries to $\bar P$ and $\bar Q$, which are not admissible), are independent
  of $\hexk$ in hybrids $\hyb_3$ and $\hyb_4$.

  We now use $\calA$ to build a PRF adversary $\calB$. Algorithm $\calB$ makes
  a single query $H_2(g^{\bar x}, g^{\bar y}, g^{\bar x \bar y})$ to the PRF oracle
  and sets $\hatk$ to be the output from the PRF oracle. In the pseudorandom world
  (where the response is the output of the PRF instantiated with a random key),
  then $\calB$ has correctly simulated $\hyb_3$ (where $\hexk$ plays the role of the
  PRF key). If the output is a uniformly random string, then $\calB$ has correctly
  simulated $\hyb_4$. Thus, by PRF security of $\extract$, $\hyb_3$ and $\hyb_4$
  are computationally indistinguishable.
\end{proof}

To conclude the proof, we argue that in $\hyb_4$, the view of the adversary is
independent of the challenge bit. First, the challenger's response to the
$\Test$ query is uniformly random over $\calK$ regardless of the challenge
bit. Finally, it suffices to argue that $\hatk$ is never used anywhere else in
$\hyb_4$, and thus, is still uniform given the adversary's view of the
protocol execution. By construction of the simulator, $\hatk$ is only used in
sessions when the shares $(g^{\bar s}, g^{\bar x}, g^{\bar y}, g^{\bar s \bar
x}, g^{\bar x \bar y})$ are used to derive the session key. By admissibility,
the adversary is not allowed to expose session $(\bar P, \hbid, \hsid)$ or its
matching session, $(\bar Q, \hbid, \hsid)$. In all other honest sessions on
which the adversary could issue a $\keyreveal$ query, the simulator must have chosen
at least one of the DH exponents $s$, $x$, or $y$. Since the adversary can only
initiate a polynomial number of sessions and the exponents $s, x, y$ are
sampled uniformly from $\Z_p$ where $p$ is super-polynomial in the security parameter
$\lambda$, we conclude that with overwhelming probability, at least one of $s
\ne \bar s$, $x \ne \bar x$, or $y \ne \bar y$. Thus, the simulator's response
to all of the adversary's queries (other than the $\Test$ query) is
independent of $\hatk$, which proves the claim. Moreover, we note that in each
of the reductions, the simulator always picked $\bar y \getsr \Z_p$ for itself.
Thus, in each case, the simulator is able to respond to a $\statereveal$ query
against the server in the test session. This completes the analysis of
Case 1.

\para{Case 2: $s$ is compromised after the handshake.} In this case, we rely
on the security of the server's ephemeral DH share to ensure confidentiality
of the session key. Our analysis is very similar to that of Case~1. We begin
by describing our hybrid experiments:

\begin{citemize}
  \item \textbf{Hybrid $\hyb_0$:} This is the real protocol execution experiment
  (same as in Case 1).

  \item \textbf{Hybrid $\hyb_1$:} Same as in Case 1, except the abort condition
  is only checked before the completion of the test session $(\bar P, \hbid, \hsid)$.

  \item \textbf{Hybrid $\hyb_2$:} Same as in Case 1.

  \item \textbf{Hybrid $\hyb_3$:} Same as in Case 1.

  \item \textbf{Hybrid $\hyb_4$:} Same as $\hyb_4$, except
  the value of $H_2(g^{\bar x}, g^{\bar y}, g^{\bar x \bar y})$ is
  replaced by a uniformly random value over $\calK$.

  \item \textbf{Hybrid $\hyb_5$:} Same as $\hyb_5$, except $\hatk$ is replaced
  by a uniformly random value over $\calK$.
\end{citemize}
Hybrids $\hyb_0$ through $\hyb_3$ are computationally indistinguishable by
applying the same arguments as in Case 1 (Claims~\ref{claim:priv-disc-hyb-0-1}
through~\ref{claim:priv-disc-hyb-2-3}). Note that the proof of
Claim~\ref{claim:priv-disc-hyb-0-1} continues to hold because we assume
that $\bar s$ is uncompromised before the completion of the test session. This
is the only part of the protocol where the claim is checked, so the same
argument applies. If suffices then to show that $\hyb_3$, $\hyb_4$, and $\hyb_5$
are computationally indistinguishable.

\begin{claim}
  Hybrids $\hyb_3$ and $\hyb_4$ are computationally indistinguishable if
  the Hash-DH assumption holds in $\bbG$.
\end{claim}
\begin{proof}
  Suppose there exists an efficient adversary $\calA$ that can distinguish
  between hybrids $\hyb_3$ and $\hyb_4$. We use $\calA$ to build a
  distinguisher $\calB$ for the Hash-DH assumption. Algorithm $\calB$ is given
  as input a Hash-DH challenge $(g^x, g^y, T)$ and must decide whether $T =
  H_2(g^x, g^y, g^{xy})$ or if $T$ is uniform over $\bbG$.

  Algorithm $\calB$ simulates the protocol execution with the exponents
  $x$ and $y$ from the Hash-DH challenge playing the roles of $\bar x$
  and $\bar y$ and $T$ playing the role of
  $H_2(g^{\bar x}, g^{\bar y}, g^{\bar x \bar y})$. Algorithm $\calB$
  chooses all of the other parameters that appear in the protocol execution
  as described in the real scheme. Since $\bar x$ and $\bar y$ are only
  used in the test session $(\bar P, \hbid, \hsid)$ and its matching
  session $(\bar Q, \hbid, \hsid)$, the simulator is able to answer
  all of the adversary's (admissible) queries exactly as in the real scheme.
  Thus, if $T = H_2(g^x, g^y, g^{xy})$, then $\calB$ has simulated
  $\hyb_3$ and if $T$ is uniform, then $\calB$ has simulated
  $\hyb_4$ for $\calA$.
\end{proof}

\begin{claim}
  Hybrids $\hyb_4$ and $\hyb_5$ are computationally indistinguishable if
  $\extract$ is a strong randomness extractor.
\end{claim}
\begin{proof}
  In $\hyb_4$, $\hatk$ is the output of $\extract$ on a uniformly random
  string while in $\hyb_5$, it is a uniformly random string. With overwhelming
  probability, the ephemeral DH exponents of all (non-corrupt) sessions other
  than the test session will not be $\bar x$ and $\bar y$ (since the simulator
  chooses the ephemeral exponents uniformly at random). Thus the value of
  $H_2(g^{\bar x}, g^{\bar y}, g^{\bar x \bar y})$ is only used once in the
  simulation and is independent of the adversary's view (by admissibility). The
  claim then follows from the fact that $\extract$ is a strong randomness
  extractor.
\end{proof}

We conclude that hybrids $\hyb_0$ and $\hyb_5$ are computationally indistinguishable.
As in Case 1, the adversary's view of the protocol execution in $\hyb_5$ no longer
depends on the test bit. Thus, the private discovery protocol
in Figure~\ref{fig:formal-discovery-protocol} is a secure service-discovery protocol
(Definition~\ref{def:service-discovery-sec}).

\subsection{Security of 0-RTT Protocol}
\label{app:0-rtt-security}

One of the main advantages of the service discovery protocol from
Figure~\ref{fig:formal-discovery-protocol} is its support for 0-RTT mutual
authentication. To enable support for application data on the first flow, when
the client prepare the initialization message, it derives an additional key
using the PRG. More precisely, when generating the initialization query,
it computes $k = H_1(g^{s}, g^x, g^{sx})$ and
$(\htk, \htk', \exk, \eadk) = \prg(k)$. In the first flow of the handshake,
it includes any early application data encrypted under $\eadk$. In this
section, we show that this 0-RTT protocol is secure in
the one-pass security model of~\cite[\S4.3]{KW15}. In~\cite[\S5.3]{KW15},
Krawczyk and Wee perform a similar analysis for the 0-RTT mode of the OPTLS
protocol in TLS~1.3.

In the one-pass setting, it is not possible to achieve perfect forward secrecy
for the lifetime of the server's broadcast (since in a 0-RTT protocol, the
server is unable to contribute an ephemeral secret to the session setup).
Moreover, the $\eadk$ in the 0-RTT protocol is vulnerable to replays of the
client's message (again, for the lifetime of the server's broadcast). One way
to address the replay problem is to have each client choose their session id
at random and have each server maintain a list of session ids that have been
used for the lifetime of each broadcast.

\para{One-pass security model.} We now specify the one-pass security model
more formally in the discovery setting. In this model, when the adversary
activates a client to initiate a session $(P, \bid, \sid)$ with a broadcast
message $B$, if the client does not abort the protocol then it outputs a
public tuple $(P, \bid, \sid, Q)$ and a secret key $\eadk$. We refer to this
tuple as the ``one-pass session output'' to distinguish it from the normal
session output. The one-pass session output for the server is defined to be its
usual session output. We note that if a client's one-pass session output
is $(P, \bid, \sid, Q)$, then if the session $(P, \bid, \sid)$ completes,
it must complete with the same output $(P, \bid, \sid, Q)$. In the one-pass
setting, we only reason about the first message in the protocol, and so
we work {\em only} with the one-pass session outputs.

As usual, the
adversary has full control over the network, and can activate parties to
initiate and respond to messages. As before, it can perform
$\broadcastreveal$, $\statereveal$, $\keyreveal$, and $\corrupt$ queries, subject
to the usual admissibility constraints. The
difference is that now the $\keyreveal$ query returns $\eadk$ instead of the
session key. We also modify the adversary's goal to be to distinguish the
early application data key of a target session rather than the session key.
In the one-pass setting, we impose an additional restriction on
the adversary:
\begin{citemize}
  \item The adversary cannot issue a $\broadcastreveal$ query on $(Q, \bid)$
  nor corrupt $Q$ before the expiration of $(Q, \bid)$.
\end{citemize}
This restriction reflects the fact that perfect forward secrecy is not
achievable for the lifetime of the server's broadcast. We now prove the
following theorem, which proceeds very similarly to the analysis of
Case~1 in the proof of Theorem~\ref{thm:disc-protocol-secure} in
Appendix~\ref{app:proof-property-2}.

\begin{theorem}
  \label{thm:one-pass-sec}
  The protocol in Figure~\ref{fig:formal-discovery-protocol} (adapted to the
  0-RTT setting) achieves one-pass security for early application data in the
  random oracle model, assuming the Strong-DH assumption in $\bbG$, and the
  security of the underlying cryptographic primitives (the signature scheme,
  the PRG, the authenticated encryption scheme, and the extraction algorithm).
\end{theorem}

\begin{proof}
  We proceed as in the proof of Theorem~\ref{thm:disc-protocol-secure}
  from Appendix~\ref{app:proof-property-2}. In particular, we define an
  analogous simulator $\bar S$ for the one-pass security experiment. We
  define the following sequence of hybrid experiments:
  \begin{citemize}
    \item \textbf{Hybrid $\hyb_0$:} This is the real
    experiment where
    $k = H_1(g^{\bar s}, g^{\bar x}, g^{\bar s \bar x})$ and
    $(\hhtk, \hhtk', \hexk, \headk) = \prg(k)$.

    \item \textbf{Hybrid $\hyb_1$:} Same as $\hyb_0$, except
    the simulator aborts if $\calA$ queries $H_1$ on the input
    $(g^{\bar s}, g^{\bar x}, g^{\bar s \bar x})$.

    \item \textbf{Hybrid $\hyb_2$:} Same as $\hyb_1$, except
    $\hhtk$, $\hhtk'$, $\hexk$, and $\headk$ are all replaced by
    uniformly random values over $\calK$.
  \end{citemize}
  Hybrids $\hyb_0$ and $\hyb_1$ are computationally indistinguishable by the
  same argument as in the proof of Claim~\ref{claim:priv-disc-hyb-0-1}.
  Next, hybrids $\hyb_1$
  and $\hyb_2$ are computationally indistinguishable by PRG security
  using the same argument as in the proof of Claim~\ref{claim:priv-disc-hyb-1-2}.

  To complete the proof, we argue that $\headk$ is independent of the
  adversary's view of the protocol execution (before it makes the $\Test$
  query). By construction $\headk$ is only used in sessions where the client's
  ephemeral DH share is $\bar x$ and the server's semi-static DH share is
  $\bar s$. Let $(\bar P, \hbid, \hsid, \bar Q)$ be the test session.
  As usual, let $\bar C, \bar S \in \set{\bar P, \bar Q}$ denote the client
  and server, respectively in the test session. We consider several cases:
  \begin{citemize}
    \item Consider a session $(P, \bid, \sid)$ where $P$ is the client
    and $(P, \bid, \sid) \ne (\bar C, \hbid, \hsid)$. Since the simulator chooses
    the ephemeral DH exponent $x$ uniformly at random in this session, with
    overwhelming probability $x \ne \bar x$. Thus, $\headk$ is independent
    of all parameters and messages associated with this session.

    \item Consider a session $(P, \bid, \sid)$ where $P$ is the client
    and $(P, \bid, \sid) = (\bar C, \hbid, \hsid)$. If $P = \bar P$, then
    this is the test session and cannot be exposed. Consider the case
    where $P = \bar Q$. Since $(\bar Q, \hbid, \hsid)$ matches
    $(\bar P, \hbid, \hsid, \bar Q)$, the adversary
    can only expose $(\bar Q, \hbid, \hsid)$ if this session completes with
    a peer $R \ne \bar P$. In this case, $\bar P$ is the server, so applying
    Lemma~\ref{lem:proper-client-init}, session $(\bar Q, \hbid, \hsid)$
    can only complete with peer $\bar P$. Thus $(\bar Q, \hbid, \hsid)$
    will always match the test session and cannot be exposed.

    \item Consider a session where $(P, \bid, \sid)$ where $P$ is the server
    and $(P, \bid) \ne (\bar S, \hbid)$. Since the simulator chooses the
    semi-static DH exponent $s$ uniformly at random for the broadcast $(P, \bid)$,
    with overwhelming probability $s \ne \bar s$. Again, $\headk$ is independent
    of the session parameters.

    \item Consider a session $(P, \bid, \sid)$ where $(P, \bid) = (\bar S, \hbid)$
    but $\sid \ne \hsid$. Suppose the session completes
    with a peer $Q$ and let $g^x$ be the client's ephemeral DH share in $(P,
    \bid, \sid)$. We argue that with overwhelming probability, $x \ne \bar x$.
    Suppose for sake of contradiction that a session $(P, \bid, \sid)$
    completes with peer $Q$ and where $Q$'s ephemeral DH share is $g^{\bar x}$.
    This means that $P$ must have received an encryption $\ct_Q$ of the
    message $(\id_P, \id_Q, \sigma_Q)$ under the handshake traffic key
    $\hhtk$ derived from $(g^{\bar s}, g^{\bar x}, g^{\bar s \bar x})$. Here,
    $\sigma_Q$ is $Q$'s signature on a string containing the tuple $(\bid,
    \sid) \ne (\hbid, \hsid)$. Whenever the adversary activates a party to initialize a session
    other than $(\bar P, \hbid, \hsid)$, the simulator chooses an ephemeral DH
    share $x$ uniformly at random. With overwhelming probability $x \ne \bar
    x$. Thus, with overwhelming probability, the simulator will only construct
    a single message encrypted under $\hhtk$, namely the message for
    session $(\bar P, \hbid, \hsid)$. In particular, $\ct_Q$ is not a
    ciphertext constructed by the simulator. Since $\hhtk$ is uniformly random
    (and unknown to the adversary), and the encryption scheme provides
    ciphertext integrity, this happens with negligible probability.
    Thus, $x \ne \bar x$ with overwhelming probability, and $\headk$ is
    independent of the session parameters.

    \item Consider a session $(P, \bid, \sid)$ where $P$ is the server and
    $(P, \bid, \sid) = (\bar S, \hbid, \hsid)$. The case $P = \bar P$
    corresponds to the test session and cannot be exposed. Suppose $P = \bar
    Q$. Then, the session $(\bar Q, \hbid, \hsid)$ matches the test session
    unless it completes with a peer $R \ne \bar P$. Suppose this happens, and
    let $g^x$ be the client's ephemeral DH share in $(\bar Q, \hbid, \hsid)$.
    By the same argument as in the previous case, we can argue that $x \ne \bar x$
    with overwhelming probability assuming that the underlying encryption
    scheme provides ciphertext integrity.
  \end{citemize}
  We have shown that the simulator's response to all admissible queries are
  independent of $\headk$ in hybrid $\hyb_3$. We conclude that the service
  discovery protocol in Figure~\ref{fig:formal-discovery-protocol} (adapted
  to the 0-RTT setting) achieves one-pass security.
\end{proof}

\subsection{Privacy of Service Discovery Protocol}
\label{app:discovery-protocol-privacy}
In this section, we show that the service discovery protocol in
Figure~\ref{fig:formal-discovery-protocol} is private if the server encrypts
its broadcast using a prefix encryption scheme (as shown in
Figure~\ref{fig:discovery-proto}). As was the case with our analysis of the
private mutual authentication scheme, we assume the prefix encryption scheme
is constructed from an IBE scheme as described in Section~\ref{sec:crypto-primitives}.
In our analysis, we use the same general
privacy model from Appendix~\ref{app:key-exchange-privacy}, but adapted to the
service discovery setting (Appendix~\ref{app:service-disc-model}). Because our
protocol supports 0-RTT mutual authentication and client identification is
provided on the {\em first} round in the protocol, we can only guarantee
a meaningful notion of privacy if we restrict ourselves to the one-pass
security setting introduced in Appendix~\ref{app:0-rtt-security}. Thus,
in the following description, we match sessions based on their one-pass
session outputs rather than their session outputs. We say a session completes
(in the one-pass sense) if it produces a one-pass session output.

In particular, in the 0-RTT service discovery model, privacy is achievable
as long as the adversary does not expose any session (in the one-pass sense)
that involves the test party. Otherwise, the adversary is able to trivially
learn the identity of the test party. We now state our admissibility requirements
more precisely:
\begin{citemize}
  \item The adversary does not corrupt $\id_\ptest$ or make a
  $\statereveal$ query on any session $(P_\ptest, \bid, \sid)$ that completes
  (in the one-pass sense).

  \item Whenever a session $(P_\ptest, \bid, \sid)$ completes (in the one-pass sense)
  with public output $Q$, then the adversary has not made a $\statereveal$ query on the
  session $(Q, \bid, \sid)$, and moreover, $Q$ is not corrupt before the completion of
  $(P_\ptest, \bid, \sid)$.

  \item Whenever a session $(P, \bid, \sid)$ completes (in the one-pass sense) with
  public output $P_\ptest$, then the adversary has not made a $\statereveal$
  query on $(P, \bid, \sid)$.

  \item Let $\Pi_\ptest$ denote the set of policies the adversary
  has associated with the test party $P_\ptest$. Let $I \subseteq [n]$ be the
  indices of the parties the adversary has either corrupted or on which it has
  issued a $\broadcastreveal$ query in the course of the
  protocol execution. Then, for all policies $\pi \in \Pi_\ptest$ and
  indices $i \in I$, it should be the case that $\id_i$ does not satisfy
  $\pi$.

  \item Whenever the adversary associated a policy $\pi$ with a
  broadcast $(P, \bid)$ or a session $(P, \bid, \sid)$, it must be the
  case that either $\sind{\id_\ptest}{0}$ and
  $\sind{\id_\ptest}{1}$ both satisfy $\pi$ or neither satisfy $\pi$.
\end{citemize}
We show the following theorem:
\begin{theorem}
  \label{thm:disc-privacy}
  The protocol in Figure~\ref{fig:formal-discovery-protocol} (where the
  broadcasts are encrypted under the server's policy using a prefix encryption
  scheme) is private in the random oracle model assuming the IBE scheme used
  to construct the prefix encryption scheme is $\indidcca$-secure, the
  Strong-DH and Hash-DH assumptions hold in $\bbG$, and the
  underlying cryptographic primitives (the signature scheme, the PRG, the
  authenticated encryption scheme, and the extraction algorithm) are secure.
\end{theorem}
\begin{proof}
  The proof is similar to that of Theorem~\ref{thm:mutual-auth-proto-private}.
  As in the proof of Theorem~\ref{thm:mutual-auth-proto-private}, we first
  define a simulator that simulates the role of the challenger for the
  adversary $\calA$ in the protocol execution environment. We then define
  a series of hybrid experiments.

  Specifically, the simulator $\calS$ takes as input the number of parties
  $n$, the security parameter $\lambda$, and the adversary $\calA$, and plays
  the role of the challenger in the protocol execution experiment with $\calA$.
  At the beginning of the simulation, $\calA$ submits a tuple of distinct identities
  $(\id_1, \ldots, \id_n)$ and two test identities $\sind{\id_\ptest}{0}$
  and $\sind{\id_\ptest}{1}$. During the protocol execution, the simulator
  chooses the parameters for each party and responds to the adversary's queries
  according to the specification of the hybrid experiment. We now define our
  sequence of hybrid experiments:
  \begin{citemize}
    \item \textbf{Hybrid $\hyb_0$:} This is the real experiment $\expt_0$.

    \item \textbf{Hybrid $\hyb_1$:} Same as $\hyb_0$, except
    $\sind{\id_\ptest}{1}$ is used in place of $\sind{\id_\ptest}{0}$
    when simulating handshake messages for sessions where $P_\ptest$
    is the client.

    \item \textbf{Hybrid $\hyb_2$:} Same as $\hyb_1$, except
    $\sind{\id_\ptest}{1}$ is used in place of $\sind{\id_\ptest}{0}$
    when simulating handshake messages for sessions where $P_\ptest$
    is the server.

    \item \textbf{Hybrid $\hyb_2$:} This is the real experiment $\expt_1$.
  \end{citemize}
  We now show that each consecutive pair of hybrid experiments is
  computationally indistinguishable.

  \begin{claim}
    \label{claim:disc-priv-hyb-0-1}
    Hybrids $\hyb_0$ and $\hyb_1$ are computationally indistinguishable in the
    random oracle model assuming the Strong-DH and Hash-DH assumptions hold in
    $\bbG$ and the security of the underlying cryptographic primitives.
  \end{claim}
  \begin{proof}
    In the service discovery protocol, the client in the protocol execution
    always encrypts both its and the server's identity under the handshake
    traffic key $\htk$. Similarly, in the server's response message, the
    server replies with an encryption of both the client's and server's
    identity under the handshake traffic key $\htk'$. Similar to the proof of
    Claim~\ref{claim:priv-mutual-auth-0-1}, it suffices to argue that the
    adversary's view of the handshake secret keys $\htk$ and $\htk'$ in all
    sessions where $P_\ptest$ is the client is computationally
    indistinguishable from uniform. The claim then follows by semantic
    security of the underlying encryption scheme.

    Formally we use the following hybrid argument. Let $q$ be a bound on the
    number of times the adversary $\calA$ activates the test party $P_\ptest$
    to initiate a session $(P_\ptest, \bid, \sid)$. We define a sequence of
    $q + 1$ hybrid experiments $\hyb_{0, 0}, \ldots, \hyb_{0, q}$ where
    in hybrid experiment $\hyb_{0, i}$, the simulator uses $\sind{\id_\ptest}{0}$
    as the identity in $P_\ptest$'s message for the first $q - i$ sessions
    and $\sind{\id_\ptest}{1}$ as the identity in the last $i$ sessions.
    By construction, $\hyb_0 \equiv \hyb_{0,0}$ and $\hyb_1 \equiv \hyb_{0,q}$.
    It suffices to argue that for each $i \in [q]$, hybrids $\hyb_{0,i-1}$ and
    $\hyb_{0,i}$ are computationally indistinguishable.

    We now argue that $\hyb_{0, i-1}$ and $\hyb_{0, i}$ are computationally
    indistinguishable. It suffices to argue that the key $\htk$ the test party
    $P_\ptest$ uses to encrypt its identity in the $\ord{(q-i+1)}$ session is
    uncompromised. Let $(P_\ptest, \sid, \bid)$ be the $\ord{(q-i+1)}$ session
    initiated at $P_\ptest$. Let $\htk$ and $\htk'$ be the handshake keys $P_\ptest$
    derives in this session to encrypt its identity. We first argue that $\htk$
    is uncompromised. An identical argument applies to $\htk'$.
    By construction, if $P_\ptest$ sent an initialization message in session $(P_\ptest, \sid, \bid)$,
    it must have completed its session (in the one-pass sense) with a peer $Q$.
    The overall argument now follows identically to that used to argue
    that the early application data key $\eadk$ in the client's message is
    uncompromised in the proof of Theorem~\ref{thm:one-pass-sec} (since $\htk$ and $\htk'$
    play an analogous role as $\eadk$ in that proof). To apply the argument from
    the proof of Theorem~\ref{thm:one-pass-sec},
    we let $(P_\ptest, \sid, \bid)$ be the test session, and verify that all the
    legal queries in the privacy setting are also admissible in the one-pass
    security setting.

    In the private discovery setting, the adversary cannot corrupt or issue a
    $\statereveal$ query on a session $(P_\ptest, \bid, \sid)$.
    This is because in the one-pass setting, if the adversary activates
    $P_\ptest$ to initiate a session $(P_\ptest, \bid, \sid)$,
    and $P_\ptest$ does not abort when processing the initialization query, then the
    session must have completed in the one-pass sense.

    It suffices to check that the adversary does not corrupt $Q$ or issue a
    $\broadcastreveal$ query on the broadcast $(Q, \bid)$ where $Q$ is
    the peer of $P_\ptest$ in the target session. Since $(P_\ptest, \bid, \sid)$
    completes, by the admissibility condition, $Q$ must not have been corrupt at the time the session
    completed. Thus, we can invoke Lemma~\ref{lem:proper-client-init} to conclude that
    $P_\ptest$ was activated to initiate a session with the broadcast message
    output by $(Q, \bid)$. Note that the proof of Lemma~\ref{lem:proper-client-init}
    only depends on the broadcast message and the first message in the mutual
    authentication handshake, so the statement applies even in our one-pass setting.

    By admissibility in the private discovery setting, if the adversary issued
    a $\broadcastreveal$ query on $(Q, \bid)$, then the policy $\pi_Q$ must not
    satisfy $P_\ptest$. But in this case, an honest $P_\ptest$ would have aborted
    the session $(P_\ptest, \bid, \sid)$. We conclude that an admissible
    adversary could not have made a $\broadcastreveal$ on the peer's broadcast $(Q, \bid)$.
    In this case, the target session $(P_\ptest, \bid, \sid)$ is admissible
    in the one-pass security setting, and so by the same argument as in the
    one-pass security proof (Theorem~\ref{thm:one-pass-sec}), we conclude that
    the handshake security keys $\htk$ and $\htk'$ are uniform and completely
    hidden to the adversary. The claim then follows by semantic security of
    the underlying authenticated encryption scheme.
  \end{proof}

  \begin{claim}
    Hybrids $\hyb_1$ and $\hyb_2$ are computationally indistinguishable in the
    random oracle model assuming Strong-DH and Hash-DH assumptions hold in
    $\bbG$ and the security of the underlying cryptographic primitives.
  \end{claim}
  \begin{proof}
    The proof of this statement follows analogously to that of
    Claim~\ref{claim:disc-priv-hyb-0-1}. We show that in all sessions
    where $P_\ptest$ is the server, the handshake encryption keys $\htk$ and
    $\htk'$ for that session have not been compromised.

    In the one-pass model, a client's initialization message
    for a session $(P, \bid, \sid)$ will contain the identity of the test
    party $\id_\ptest$ if and only if $(P, \bid, \sid)$ completes with
    one-pass session output $(P, \bid, \sid, P_\ptest)$. It suffices to show
    that $(P, \bid, \sid)$ is not exposed, in which case the handshake encryption
    keys $\htk$ and $\htk'$ are uniform and hidden from the adversary (by a similar
    argument as that used in the proof of Theorem~\ref{thm:one-pass-sec}). The claim
    then follows by a hybrid argument similar to the one used in the proof of
    Claim~\ref{claim:disc-priv-hyb-0-1}.

    Consider a session $(P, \bid, \sid)$ that completes with peer $P_\ptest$.
    Since the session completes, $P$ could not have been corrupt. In the
    discovery model, $P_\ptest$ also cannot be corrupted. Next, by
    admissibility, the adversary cannot issue a $\statereveal$ query on $(P,
    \bid, \sid)$. By one-pass security, the handshake encryption keys for
    this session are uncompromised. The claim follows by semantic security
    of the underlying authenticated encryption scheme.
  \end{proof}

  \begin{claim}
    \label{claim:disc-priv-hyb-2-3}
    Hybrids $\hyb_2$ and $\hyb_3$ are computationally indistinguishable if
    the underlying IBE scheme is $\indidcca$-secure.
  \end{claim}
  \begin{proof}
    This proof proceeds very similarly to that of
    Claim~\ref{claim:priv-mutual-auth-1-2}. In particular, we let $q$ be an
    upper bound on the number of sessions where $\calA$ activates the
    test party $P_\ptest$ to initiate a broadcast. We define a sequence of
    $q+1$ hybrid experiments $\hyb_{2,0}, \ldots, \hyb_{2,q}$ where
    hybrid experiment $\hyb_{2,i}$ is defined as follows:
    \begin{citemize}
      \item Same as $\hyb_2$ except the first $i$ times $P_\ptest$ is activated
      to initialize a broadcast, $P_\ptest$ substitutes the identity
      $\sind{\id_\ptest}{1}$ for $\sind{\id_\ptest}{0}$ in its broadcast. In
      all subsequent times $P_\ptest$ is activated to initialize a broadcast,
      it uses the identity $\sind{\id_\ptest}{0}$.
    \end{citemize}
    We now show that for all $i \in [q]$, hybrids $\hyb_{2,i-1}$ and $\hyb_{2,i}$
    are computationally indistinguishable assuming that the IBE scheme is
    $\indidcca$-secure. Suppose $\calA$ is able to distinguish $\hyb_{2,i-1}$
    from $\hyb_{2,i}$. We use $\calA$ to construct an adversary $\calB$
    for the $\indidcca$-security game. First, $\calB$ is given the public
    parameters $\mpk$ for the IBE scheme. Then, $\calB$ begins running
    $\calA$ and obtains a tuple of identities $(\id_1, \ldots, \id_n)$
    and test identities $\sind{\id_\ptest}{0}$ and $\sind{\id_\ptest}{1}$.
    Algorithm $\calB$ simulates the setup procedure in $\hyb_2$ by choosing
    signing and verification keys for each party $P_i$ and $P_\ptest$. It
    also issues certificates binding $\id_i$ to the verification key for
    each $P_i$. It prepares two certificates binding $P_\ptest$ to
    identities $\sind{\id_\ptest}{0}$ and $\sind{\id_\ptest}{1}$.
    Finally, $\calB$ gives $\mpk$ to $\calA$ and begins simulating the
    protocol execution experiment for $\calA$:
    \begin{citemize}
      \item \textbf{Server broadcast queries.} When adversary $\calA$
      activates a party $P$ to initiate a broadcast $(P, \bid)$, if $P \ne P_\ptest$,
      algorithm $\calB$ simulates the response as in the real scheme. If $P = P_\ptest$,
      then let $\ell$ be the number of times $\calA$ has activated $P_\ptest$
      to initiate a broadcast. Let $\pi$ be the policy specified by $\calA$.
      Algorithm $\calB$ then responds as follows:
      \begin{citemize}
        \item If $\ell < i - 1$, $\calB$ constructs the broadcast as in
        $\hyb_3$, that is using $\sind{\id_\ptest}{1}$.

        \item If $\ell \ge i$, $\calB$ constructs the broadcast as in $\hyb_2$,
        that is, using $\sind{\id_\ptest}{0}$.

        \item If $\ell = i - 1$, $\calB$ chooses a random DH share $s \getsr \Z_p$,
        and computes signatures $\sigma_0 = \sig_{P_\ptest}(\bid, \sind{\id_\ptest}{0}, g^s)$
        and $\sigma_1 = \sig_{P_\ptest}(\bid, \sind{\id_\ptest}{1}, g^s)$. It submits
        the tuples $(\sind{\id_\ptest}{0}, g^s, \sigma_0)$ and
        $(\sind{\id_\ptest}{1}, g^s, \sigma_1)$ to the IBE challenger
        with $\pi$ as its challenge identity. It receives a ciphertext
        $\hct_S$ from the challenger. Algorithm $\calB$ outputs the broadcast
        $(\bid, \hct_S  )$.
      \end{citemize}

      \item \textbf{Client initialization queries.} When an adversary activates
      a client $P$ to initiate a session $(P, \bid, \sid)$ with a broadcast $(\bid, \ct_S)$
      and server policy $\pi_S$, $\calB$ does the following:
      \begin{cenumerate}
        \item If there is already a session $(P, \bid, \sid)$, $\calB$ aborts the session.
        If $P \ne P_\ptest$ and $\id_P$ does not satisfy $\pi_S$, $\calB$
        aborts the session. If $P = P_\ptest$ and $\sind{\id_\ptest}{0}$ does not
        satisfy $\pi_S$, $\calB$ also aborts the session. Recall that our admissibility
        requirement states that either both $\sind{\id_\ptest}{0}$ and $\sind{\id_\ptest}{1}$
        satisfy $\pi_S$ or neither satisfy $\pi_S$.

        \item If $\calB$ is still in the pre-challenge phase, or if $\calB$ is in the
        post-challenge phase and either $\ct_S \ne \hct_S$ or $\pi_S \ne \bar \pi$
        (where $\bar \pi$ is the identity $\calB$ submitted to the IBE challenger in the
        challenge phase), then $\calB$ queries the IBE decryption oracle on $\ct_S$ and
        identity $\pi_S$ to obtain a decrypted broadcast message (or $\bot$). Algorithm
        $\calB$ performs the checks on the decrypted broadcast message and simulates the
        response as described in $\hyb_2$.
        \\ \\
        If $\calB$ is in the post-challenge phase and $\ct_S = \hct_S$ and
        $\pi_S = \bar \pi$, then $\calB$ aborts the session if
        $\sind{\id_\ptest}{0}$ does not satisfy the the client's policy
        (specified by the adversary). Again, by admissibility, either both
        $\sind{\id_\ptest}{0}$ and $\sind{\id_\ptest}{1}$ satisfy the client's
        policy or neither do. If $\calB$ does not abort the session, then
        $\calB$ simulates the client's response message as described in
        $\hyb_2$ (always using the identity $\sind{\id_\ptest}{1}$ for
        $\id_\ptest$).
      \end{cenumerate}

      \item \textbf{Server response queries.} These are handled exactly as in
      $\hyb_2$ and $\hyb_3$. They are independent of the IBE parameters.

      \item \textbf{Client finish queries.} These are handled exactly as in
      $\hyb_2$ and $\hyb_3$.

      \item \textbf{$\statereveal$ and $\keyreveal$ queries.} These are handled
      exactly as in $\hyb_2$ and $\hyb_3$.

      \item \textbf{$\corrupt$ queries.} If $\calA$ asks to corrupt a party $P
      \ne P_\ptest$ (since $\calA$ is admissible), $\calB$ queries the IBE extraction
      oracle for the secret keys for $\id_P$ and each prefix of $\id_P$. It gives
      these secret keys to $\calA$, the long-term signing key associated with $\id_P$,
      and any ephemeral secrets for incomplete sessions currently in the local storage
      of $P$.
    \end{citemize}
    At the end of the game, adversary $\calA$ outputs a guess for whether it is
    in $\hyb_2$ or $\hyb_3$. Algorithm $\calB$ echoes this guess.

    To complete the proof, we show that $\calB$ is an admissible IBE adversary in the
    $\indidcca$-security game. By construction, $\calB$ never asks the adversary to
    decrypt the challenge ciphertext. Similar to the proof of Claim~\ref{claim:priv-mutual-auth-1-2}
    we appeal to the admissibility of $\calA$ to argue that algorithm $\calB$ never
    needs to query the extraction oracle for the identity $\bar \pi$ in the challenge
    query during the simulation.

    By construction, if $\calB$ receives an encryption of $\sind{\id_\ptest}{0}$ from
    the IBE challenger, then it has correctly simulated the broadcast queries according
    to the specification of hybrid $\hyb_{2,i-1}$ for $\calA$. If
    it receives an encryption of $\sind{\id_\ptest}{1}$ from the IBE
    challenger, then it has correctly simulated the broadcast queries according to the
    specification of hybrid $\hyb_{2,i}$ for $\calA$.

    To conclude the proof, we
    check that the client initialization queries are correctly simulated. The only
    non-trivial case is when the adversary submits $\ct_S = \hct_S$ and $\pi_S = \bar \pi_S$.
    All other cases are processed exactly as in $\hyb_2$ and $\hyb_3$. In the case where
    $\ct_S = \hct_S$ and $\pi_S = \bar \pi_S$, the ciphertext is either a valid
    encryption of a broadcast from $P_\ptest$ with identity $\sind{\id_\ptest}{0}$
    or with identity $\sind{\id_\ptest}{1}$. By admissibility, the client will
    either accept both $\sind{\id_\ptest}{0}$ and $\sind{\id_\ptest}{1}$ or neither.
    Thus, in the real scheme, the client's decision to abort the session is {\em independent}
    of whether the server's identity is $\sind{\id_\ptest}{0}$ or $\sind{\id_\ptest}{1}$.
    The response generation step in $\hyb_2$ and $\hyb_3$ depends only on
    $\sind{\id_\ptest}{1}$, and is in particular, independent of
    the world bit for the IBE security game. Thus, correctness of the simulation follows.
    Thus, if the IBE scheme is $\indidcca$ secure, then $\hyb_2$ and $\hyb_3$ are computationally
    indistinguishable.
  \end{proof}

  \noindent Combining Claims~\ref{claim:disc-priv-hyb-0-1} through~\ref{claim:disc-priv-hyb-2-3},
  we conclude that the service discovery protocol in Figure~\ref{fig:formal-discovery-protocol}
  (where the broadcasts are encrypted under the server's policy using the prefix encryption
  scheme) is a private service discovery protocol.
\end{proof}

\end{document}